\documentclass[journal,twocolumn]{IEEEtran}

\usepackage[numbers,sort&compress]{natbib}
\usepackage{bm}
\usepackage{subcaption}
\usepackage{arydshln}
\usepackage{booktabs}
\usepackage{graphicx}
\usepackage{multirow}
\usepackage{makecell} 
\usepackage{placeins}
\usepackage{adjustbox}

\usepackage[cmex10]{amsmath}
\usepackage{scalerel}
\usepackage{multirow}
\usepackage[table,xcdraw]{xcolor}
\usepackage{mathtools}
\usepackage{multicol}

\usepackage{amsmath,amssymb,amsfonts}
\usepackage{float}
\usepackage{mathrsfs}
\usepackage{array}
\usepackage{caption}
\usepackage{subcaption}
\usepackage{balance} 
\usepackage{mdwmath}
\usepackage{mdwtab}
\usepackage{eqparbox}
\usepackage{url}
\usepackage{hyperref}
\hypersetup{
     colorlinks   = true,
     citecolor    = blue,
     linkcolor     = blue,
     urlcolor       =blue,
}
\newtheorem{theorem}{\color{black}{\bfseries Theorem}}
\newtheorem{remark}{\color{black}{\bfseries Remark}}

\usepackage[utf8]{inputenc}

\DeclareUnicodeCharacter{221E}{\infty}
\definecolor{applegreen}{rgb}{0.55, 0.71, 0.0}
\definecolor{bazaar}{rgb}{0.5, 1.0, 0.83}
\definecolor{bazaar}{rgb}{0.43, 0.5, 0.5}
\definecolor{bazaar}{rgb}{0.6, 0.47, 0.48}
\newcolumntype{|}{!{\vrule width 1.2pt}} 
\newcommand{\customCline}[3]{%
  \noalign{\global\arrayrulewidth=#1}%
  \cline{#2-#3}%
  \noalign{\global\arrayrulewidth=.4pt}%
}
\begin{document}
\medmuskip=0mu
\thinmuskip=0.5mu
\thickmuskip=0.5mu 
\nulldelimiterspace=0pt
\scriptspace=0pt
\bstctlcite{IEEEexample:BSTcontrol}
\title{Rigid Communication Topologies: Impact on Stability, Safety, Energy Consumption, Passenger Comfort, and Robustness of Vehicular Platoons}
  \author{Amir~Zakerimanesh$^{1,2}$,
      Tony~Zhijun~Qiu$^{2}$, Mahdi~Tavakoli$^{1}$

\thanks{
$^{1}$ Department of Electrical and Computer Engineering, University of Alberta, Edmonton, Alberta, Canada.

$^{2}$ Center for Smart Transportation, Department of Civil and Environmental Engineering, University of Alberta, Edmonton, Alberta, Canada.

}}

\markboth{	
}{A. Zakerimanesh \MakeLowercase{\textit{et al.}}:Rigid Communication Topologies: Impact on Stability, Safety, Energy Consumption, and Robustness of Vehicular Platoons}
\maketitle
\begin{abstract}
This paper investigates the impact of rigid communication topologies (RCTs) on the performance of  vehicular platoons, aiming to identify beneficial features in RCTs that enhance vehicles behavior. We introduce four performance metrics, focusing on safety, energy  consumption, passenger comfort, and robustness of vehicular platoons. The safety metric is based on momentary distances between neighboring vehicles, their relative velocities, and relative accelerations. Thus, to have access to these relative values, we formulate the coupled dynamics between pairs of neighboring vehicles, considering initial conditions (position, velocity, acceleration), leader vehicle's velocity/acceleration trajectory, deployed RCT, and vehicles' parity/disparity. By decoupling the dynamics using a mapping matrix structured on deployed RCT, vehicles' features, and control gains, precise formulations for distance errors, relative velocities, and relative accelerations between all neighboring vehicles, over the travel time, are obtained. Comparing performance metric results across RCTs highlights that downstream information transmission—from vehicles ahead, particularly the leader vehicle, to vehicles behind—significantly enhances platoon stability, safety, energy  consumption, and passenger comfort metrics. Conversely, receiving state information from vehicles behind degrades metrics, compromising safety, increasing energy  consumption, and reducing passenger comfort. These findings underscore that forward-looking, leader-centric communications between vehicles markedly enhance platoon efficiency and safety.
\end{abstract}
%
\IEEEpeerreviewmaketitle
\section{Introduction}
Vehicular platooning, the coordinated movement of a group of autonomous or semi-autonomous vehicles, has garnered significant attention in recent years due to its potential to enhance road safety, improve traffic flow, and reduce fuel consumption and emissions \cite{alam2013cooperative,xu2014impact}. Central to the successful implementation of vehicular platooning is the communication topology employed, which dictates how information is exchanged among vehicles within the platoon \cite{balador2022survey,xu2014communication,taylor2022vehicular,razzaghpour2022impact}. The advent of Connected and Automated Vehicles (CAVs), equipped with Vehicle to Vehicle (V2V) communication capabilities, marks a crucial step toward the future, promising substantial improvements in traffic safety \cite{amoozadeh2015platoon,razzaghpour2021impact}.

Recent studies have increasingly focused on investigating the impact of communication topologies on various aspects of vehicular platooning. For instance, the role of V2V communication in enhancing safety within platooning systems is explored in \cite{sidorenko2021safety}, demonstrating that emergency braking strategies utilizing V2V communication substantially outperform radar-based automatic emergency braking systems (AEBS) in terms of permissible intervehicle distances and response times. Similarly, \cite{prayitno2021v2v} discusses how communication impacts disturbance propagation and overshoot magnitude in both downstream and upstream directions.

The unreliability of V2V communication in platoons, due to factors such as hardware damage, packet loss, or cyber-attacks, can lead to instability and increased collision risks. For example, in cooperative adaptive cruise control (CACC) systems, vehicles employing the predecessor-following (PF) strategy revert to adaptive cruise control (ACC) during communication interruptions. In such cases, the sensors of the ACC system must replace V2V communication for information collection, which can compromise safety and increase the risk of collisions \cite{milanes2013cooperative,milanes2014modeling}. To mitigate these risks, a dynamic weights optimization-based CAV following model (DWOC) with a multi-predecessor following (MPF) topology is proposed in \cite{wang2024stability}, which adaptively adjusts information weights to enhance stability under varying communication conditions.

Other strategies have also been proposed to enhance stability and reliability under communication challenges. In \cite{he2021distributed}, the compromised communication in vehicle platooning is tackled by designing a observer-based distributed controller that identifies sensors under attack, ensuring the system's resilience and stability despite disrupted or malicious communication. Communication strategies combining synchronized communication slots with transmit power adaptation are studied in \cite{segata2015toward}, proving their suitability for cooperative driving applications even in highly congested freeway scenarios. The impact of random link interruptions and topology changes caused by packet loss and erasure channels, modeled using Markov chains, is analyzed in \cite{nguyen2017impact} to understand their effect on vehicle platoon formation and robustness. Additionally, the impact of communication delays on vehicle platoon safety and stability has been a significant focus, as explored in \cite{khattak2023impact,zhang2023impacts}. Studies like \cite{sun2020impacts} have investigated how attacks on communication and sensing systems can introduce delays and false data, potentially driving platoons into unsafe states and increasing the risk of high-velocity crashes through reachability analysis.

Furthermore, the impact of different communication technologies on platoon performance has been examined. For example, the impact of Dedicated Short-Range Communication (DSRC) and Fifth-Generation (5G) cellular communication on energy consumption and stability in electric truck platoon formations is analyzed in \cite{devika2023impact}. Additionally, \cite{qin2019impact} studies how optimizing traffic flow stability through V2V communication can enhance passenger comfort in CAVs.

While existing literature has made significant strides in examining the adversarial effects of non-ideal communications in vehicular platoons, there is still much to explore regarding the distinct characteristics of different communication topologies (CTs) and their effects on platoon performance. Insights into the positive and negative features arising from the communication structure itself can help researchers develop more efficient and optimized vehicular platoon systems for real-world applications. Recently, studies have begun addressing these gaps. For example, \cite{pirani2022impact} applies a general graph theory framework to explore the impact of connectivity measures within CTs on the performance of distributed algorithms, assessing their ability to mitigate communication disruptions, detect cyber-attacks, and maintain resilience. The study reveals that traditional platoon topologies based on nearest-neighbor interactions are highly vulnerable in terms of performance and security, while k-nearest neighbor topologies can achieve better security and performance levels. Furthermore, \cite{ruan2022impacts} explores distinctions among three unidirectional communication topologies (UCTs) and their implications for stability, robustness, safety, and emissions in vehicle platoons, using safety metrics such as maximum time to collision (MTTC) and deceleration rate to avoid a crash (DRAC).

This paper introduces a comprehensive analytical framework for evaluating rigid communication topologies (RCTs) in vehicular platoons. Our approach formulates the coupled dynamics between neighboring vehicles, incorporating initial conditions (position, velocity, and acceleration), the leader vehicle's trajectory, vehicle heterogeneity, and the deployed RCT. By decoupling these dynamics using a mapping matrix, we derive precise formulations for distance errors (from desired intervehicle distances), relative velocities, and relative accelerations between neighboring vehicles. Based on these relative measurements, we define and analyze four key performance metrics—Safety, Energy Consumption, Passenger Comfort, and Robustness—to evaluate the impact of different RCTs on platoon performance. 

Our study reveals that forward-looking communication strategies enhance safety, reduce energy consumption, and improve robustness and passenger comfort compared to topologies where information is primarily received from vehicles behind. Specifically, topologies such as TPFL, MPF, PFL, and BDL significantly outperform those common communications where information is predominantly received from vehicles behind, such as BD and SPTF. The latter result in degraded performance, increased safety risks, higher energy consumption, and reduced passenger comfort. These findings indicate that for larger platoon sizes, forward-looking, leader-centric communications like TPFL and PFL generally offer superior performance. Additionally, the study models platoon closed-loop dynamics based on relative measurements for both homogeneous and heterogeneous platoons, using these dynamics to analyze platoon stability.

The paper is organized as follows: Section \ref{section2} covers Preliminaries, outlining key concepts. Section \ref{section3} details Problem Formulation. Section \ref{platoonStability} examines Platoon Stability. Section \ref{pmetrcis} introduces Performance Metrics. Section \ref{simpre} presents Simulations and Results. Finally, Section \ref{conclusions} summarizes the findings and drawn conclusions.
\section{\textbf{Preliminaries}}
\label{section2}
For the sake of conciseness, we generally omit the explicit indication of the time argument in signal functions, except where clarity requires it. Additionally, we use semicolons and colons to differentiate elements of vertical and horizontal vectors, respectively. For example, $[.;.;.]$ denotes a 3-by-1 vector, while $[.,.,.]$ denotes a 1-by-3 vector.  Figure \ref{platoon} depicts a platoon including the leader vehicle, designated by $0$, and $n$ follower vehicles (FVs) where are labeled by $1,\dots,i-1,i,\dots,n-1,n$. The $x$ axis shows the movement direction of vehicles, and the notation $x_{0}$ indicates the front position of the leader vehicle, and $x_{i}$; $i=1,\dots,n$, show the front positions of the FVs, respectively. Also, $L_{i}$; $i=0,\dots,n$, show the length of the vehicles. Other notations will be touched on later.
\begin{figure}[htbp!]
\begin{center}
\resizebox{1\hsize}{!}{\includegraphics*{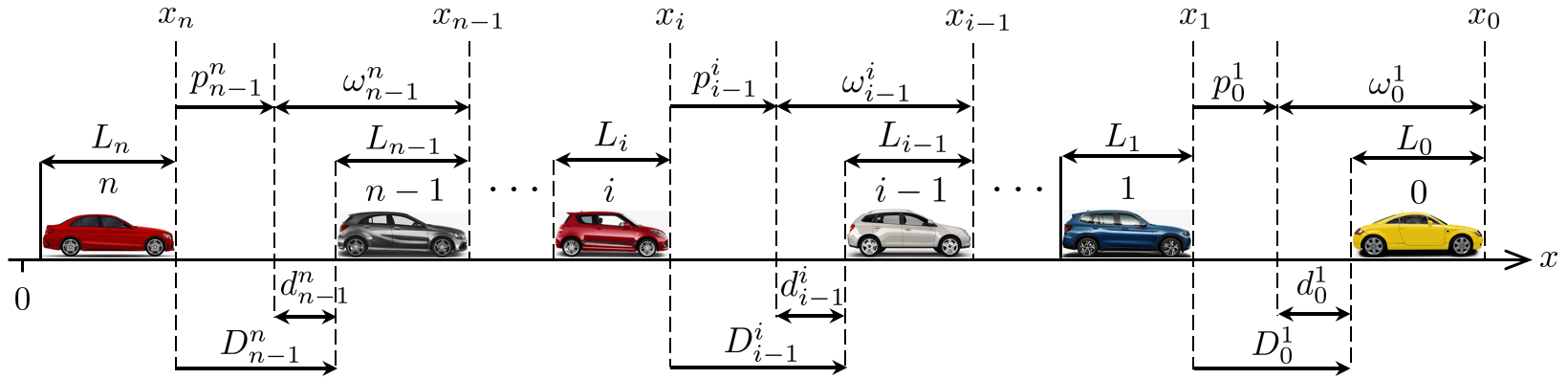}}
\caption{Schematic of a heterogeneous platoon.}
\label{platoon}
\end{center}
\end{figure}

Platooning is to uphold desired distances between vehicles while ensuring that FVs gradually align their velocities ($v_{i}$) and accelerations ($a_{i}$) with the leader's values ($v_{0}$ and $a_{0}$). Thus. the goal is to achieve
\begin{equation}
\begin{cases}
v_{i}=v_{0}; & i=1,\dots,n\\
a_{i}=a_{0}; & i=1,\dots,n
\\
        x_{i-1}-x_{i}=\omega_{i-1}^{i}; & i=1,\dots,n
\end{cases}
\label{aims}
\end{equation}
in which $\omega_{i-1}^{i}\triangleq L_{i-1}+d_{i-1}^{i}$, and the preset value $d_{i-1}^{i}$ represents the desired distance between vehicles $i-1$ and $i$.
\subsection{\textbf{Vehicle Dynamics}}
The leader vehicle undergoes no control process, but its position, velocity, and acceleration are utilized to control the FVs. Given that said, the behavior of each FV is characterized by the following formulation \cite{huang1998longitudinal}:
\begin{equation}
    \dot{a}_{i}=f_{i}\left(v_{i},a_{i}\right)+g_{i}\left(v_{i}\right)c_{i};
    \hspace{0.8cm} i=1,\dots,n
    \label{vehicle_dynamics}
\end{equation}
where $v_{i}$, $a_{i}$, and $c_{i}$ are velocity, acceleration, and engine input of the $i^{th}$ follower. 
Also, $f_{i}\left(v_{i},a_{i}\right)$ and $g_{i}\left(v_{i}\right)$ are defined as:
\begin{equation}
\begin{cases}    
\begin{split}
&f_{i}\left(v_{i},a_{i}\right)=-\frac{1}{\tau_{i}}\left(a_{i}+\frac{\sigma A_{cs_{i}}C_{d_{i}}v_{i}^{2}}{2m_{i}}+\frac{d_{m_{i}}}{m_{i}}\right)-\frac{\sigma A_{cs_{i}}C_{d_{i}}v_{i}a_{i}}{m_{i}}
\end{split}
\\
\begin{split}
&g_{i}\left(v_{i}\right)=\frac{1}{\tau_{i}m_{i}}
\end{split}
\end{cases}
\end{equation}
where the parameters $\sigma$, $A_{cs_{i}}$, $C_{d_{i}}$, $d_{m_{i}}$, $m_{i}$, and $\tau_{i}$ are the specific mass of air, vehicle cross-sectional area, drag coefficient, mechanical drag, mass, and engine time constant, respectively. The engine input $c_{i}$ is governed by the following feedback linearization controller:
\begin{equation}
c_{i}=u_{i}m_{i}+0.5\sigma A_{cs_{i}}C_{d_{i}}v_{i}^{2}+d_{m_{i}}+\tau_{i}\sigma A_{cs_{i}}C_{d_{i}}v_{i}a_{i}
\label{flc}
\end{equation}
where by substituting it into \eqref{vehicle_dynamics}, each FV within the platoon can mathematically described by a third-order linear model \cite{zheng2015stabilitym,zheng2015stability,li2017distributed,ge2022scalable,lan2021data,huang2022design}:
\begin{equation}
\tau_{i}\dot{a}_{i}+a_{i}=u_{i};
    \hspace{0.8cm} i=1,\dots,n
\label{fl}
\end{equation}
in which $u_{i}$ is the linear controller needs to be designed properly. Let ${\mathbf{X}_{i}}\triangleq\left[x_{i};\;\dot{x}_{i};\;\ddot{x}_{i}\right]$ describe the state of the $i^{th}$ follower such that $\dot{x}_{i}=v_{i}$ and $\ddot{x}_{i}=a_{i}$. Then, given \eqref{fl}, the state-space model for the $i^{th}$ follower can be expressed as follows:
\begin{equation}
\dot{\mathbf{X}}_{i} = \underbrace{\begin{bmatrix}
0&1&0\\
0&0&1\\
0&0&-\frac{1}{\tau_{i}}
\end{bmatrix}}_{\mathbf{A}_{i}}
\mathbf{X}_{i}+
\underbrace{\begin{bmatrix}
0\\
0\\
\frac{1}{\tau_{i}}
\end{bmatrix}}_{\mathbf{B}_{i}}u_{i};\quad i=1,\dots,n
\label{trss}
\end{equation}
\subsection{\textbf{Intervehicle States}}
\label{pstate}
Here we define platoon's states based on coupled states of neighboring vehicles $(i-1,i)$; $i$ indexed from $1$ to $n$.  Later, we will develop the dynamic model for the platoon using these states. Thus, let the coupled states of the pair $(i-1,i)$ be defined as $\tilde{\mathbf{X}}_{i-1}^{i}\triangleq [p_{i-1}^{i};\; v_{i-1}^{i};\; a_{i-1}^{i}]$ and so $\dot{\tilde{\mathbf{X}}}_{i-1}^{i} \triangleq [v_{i-1}^{i};\; a_{i-1}^{i};\; \mathbb{j}_{i-1}^{i}]$, where:
\begin{equation}
\begin{cases}
p_{i-1}^{i}(t)\triangleq x_{i-1}(t)-x_{i}(t)-\omega_{i-1}^{i}\\
v_{i-1}^{i}(t)\triangleq v_{i-1}(t)-v_{i}(t)\\
a_{i-1}^{i}(t)\triangleq a_{i-1}(t)-a_{i}(t)\\
\mathbb{j}_{i-1}^{i}(t)\triangleq \mathbb{j}_{i-1}(t)-\mathbb{j}_{i}(t)
\end{cases}
\label{newstates}
\end{equation}
in which, for $j\in\{i-1,i\}$, $v_{j}(t)=\dot{x}_{j}(t)$,  $a_{j}(t)=\ddot{x}_{j}(t)$, and $\mathbb{j}_{j}(t)=\dddot{x}_{j}(t)$ are the velocity, acceleration, and jerk (i.e., the derivative of acceleration) of the $j^{th}$ vehicle, respectively. It is worth noting that $\dot{p}_{i-1}^{i}=v_{i-1}^{i}$, $\dot{v}_{i-1}^{i}=a_{i-1}^{i}$, and $\dot{a}_{i-1}^{i}=\mathbb{j}_{i-1}^{i}$. The interpretations of the  coupled states are:

1. The \textbf{distance error} ($p_{i-1}^{i}(t)$) between neighboring vehicles: Given Fig. \ref{platoon}, the distances between neighboring vehicles are denoted as $D_{i-1}^{i}(t)$ where $D_{i-1}^{i}(t)=x_{i-1}(t)-x_{i}(t)-L_{i-1}=p_{i-1}^{i}(t)+d_{i-1}^{i}$ in which $p_{i-1}^{i}(t)$ are actually distance errors from the desired values $d_{i-1}^{i}$. Consequently, when $p_{i-1}^{i}(t)\rightarrow 0$, the distances will converge to their desired values.

2. The \textbf{relative velocity} ($v_{i-1}^{i}(t)$) between neighboring vehicles: The state indicates the velocity of the preceding vehicle with respect to its immediate FV.
Thus, if $v_{i-1}^{i}\rightarrow 0$, then the neighboring vehicles' velocities will be synchronized.

3. The \textbf{relative acceleration} ($a_{i-1}^{i}(t)$) between neighboring vehicles: The state quantifies the acceleration of the preceding vehicle with respect to its immediate FV.
Thus, if $a_{i-1}^{i}\rightarrow 0$, then the neighboring vehicles' accelerations will be synchronized.
\begin{remark}
The rationale behind our definitions is rooted in safety considerations. The distance errors, relative velocities and relative accelerations form the basis for the safety metrics defined later in this paper for the platoon. Later we will find the formulations of these intervehicle coupled states [see \eqref{impuleseRA}-\eqref{impulese}].
\label{re1re}
\end{remark}
\subsection{\textbf{Distributed Linear Controller}}
The controller $u_{i}$ in \eqref{fl}-\eqref{trss}, to fulfill the aims in \eqref{aims}, is defined as \cite{zakerimanesh2021heterogeneous,zakerimanesh2024stability}:
\begin{equation}
u_{i}=-\sum\limits_{\mathclap{j\in\mathbb    {I}_{i}}} \mathbf{K}\tilde{\mathbf{X}}_{i}^{j};
    \hspace{0.8cm} i=1,\dots,n
\label{contr}
\end{equation}
in which the set $\mathbb{I}_{i}$ shows the vehicles from which the $i^{th}$ follower gets information,  $\mathbf{K}\triangleq[k,\;b,\;h]$ is the control gain vector (CGV), and $\tilde{\mathbf{X}}_{i}^{j}\triangleq[p_{i}^{j};\;v_{i}^{j};\;a_{i}^{j}]$ such that $v_{i}^{j}\triangleq v_{i}-v_{j}$, $a_{i}^{j}\triangleq a_{i}-a_{j}$, and $p_{i}^{j}\triangleq x_{i}-x_{j}-d_{i,j}$ where
\begin{equation}
d_{i,j}\triangleq-sgn(i-j)\sum_{\kappa=\min(i,j)+1}^{\max(i,j)}L_{\kappa-1}+d_{\kappa-1}^{\kappa};\quad i=1,\dots,n
\label{contrdipj}
\end{equation}
where $d_{\kappa-1}^{\kappa}$ denotes the desired distance between neighboring vehicles and $L_{\kappa-1}$ is the length of the $(\kappa-1)^{th}$ vehicle.
\subsection{\textbf{Stable and Unstable Platoon}}
\label{safesa}
A key goal of platooning is to keep vehicles at safe distances from each other. To ensure safety, neighboring vehicles should not get closer than a predefined safe distance, $d_{i-1,i}^{s}$, which is less than the desired distance. The relation between intervehicle distances (IDs) and distance errors is linear, as given by $D_{i-1}^{i}(t)=p_{i-1}^{i}(t)+d_{i-1}^{i}$. Figure \ref{area_diagram} illustrates various stable scenarios.

1. \textbf{Stable-Safe Platoon}: In this scenario, all IDs, during the travel time, remain larger than the safe distances. In other words, for $i=1,\dots,n$, $D_{i-1}^{i}(t)>d_{i-1,i}^{s}$ or equivalently  $p_{i-1}^{i}(t)>d_{i-1,i}^{s}-d_{i-1}^{i}$ hold true. The portions of the $p_{i-1}^{i}(t)$ and $D_{i-1}^{i}(t)$ axes that meet this condition in the coordinate plane $(p_{i-1}^{i}(t),D_{i-1}^{i}(t))$, shown in Fig. \ref{area_diagram}, are depicted in green.

2. \textbf{Stable-Unsafe Platoon}: In this scenario, during the travel time, there exists at least one pair of neighboring vehicles where the ID becomes less than the safe distance. In other words, if $0 < D_{i-1}^{i}(t) < d_{i-1,i}^{s}$ occurs at any point during the travel time, we classify the platoon as unsafe. This condition can be equivalently checked by verifying whether $-d_{i-1}^{i} < p_{i-1}^{i}(t) < d_{i-1,i}^{s} - d_{i-1}^{i}$ occurs for at least one pair of neighboring vehicles over the travel time. The portions of the $p_{i-1}^{i}(t)$ and $D_{i-1}^{i}(t)$ axes that satisfy this condition in the coordinate plane $(p_{i-1}^{i}(t),D_{i-1}^{i}(t))$, as depicted in Fig. \ref{area_diagram}, are colored blue. 

3. \textbf{Stable-Colliding Platoon}: In this scenario, a collision occurs between at least one pair of adjacent vehicles. In other words, if $D_{i-1}^{i}(t) \leq 0$ or equivalently $p_{i-1}^{i}(t) \leq -d_{i-1}^{i}$ occurs at any point during the travel time, then a collision has occurred within the platoon. The parts of the $p_{i-1}^{i}(t)$ and $D_{i-1}^{i}(t)$ axes that satisfy this condition in the coordinate plane $(p_{i-1}^{i}(t),\;D_{i-1}^{i}(t))$, depicted in Fig. \ref{area_diagram}, are marked in red.

4. \textbf{Unstable Platoon}: If the internal stability conditions are not satisfied, then the platoon would be unstable [see Section  \ref{platoonStability}: Platoon Stability].
\begin{figure}[htbp!]
\begin{center}
\resizebox{1\hsize}{!}{\includegraphics*{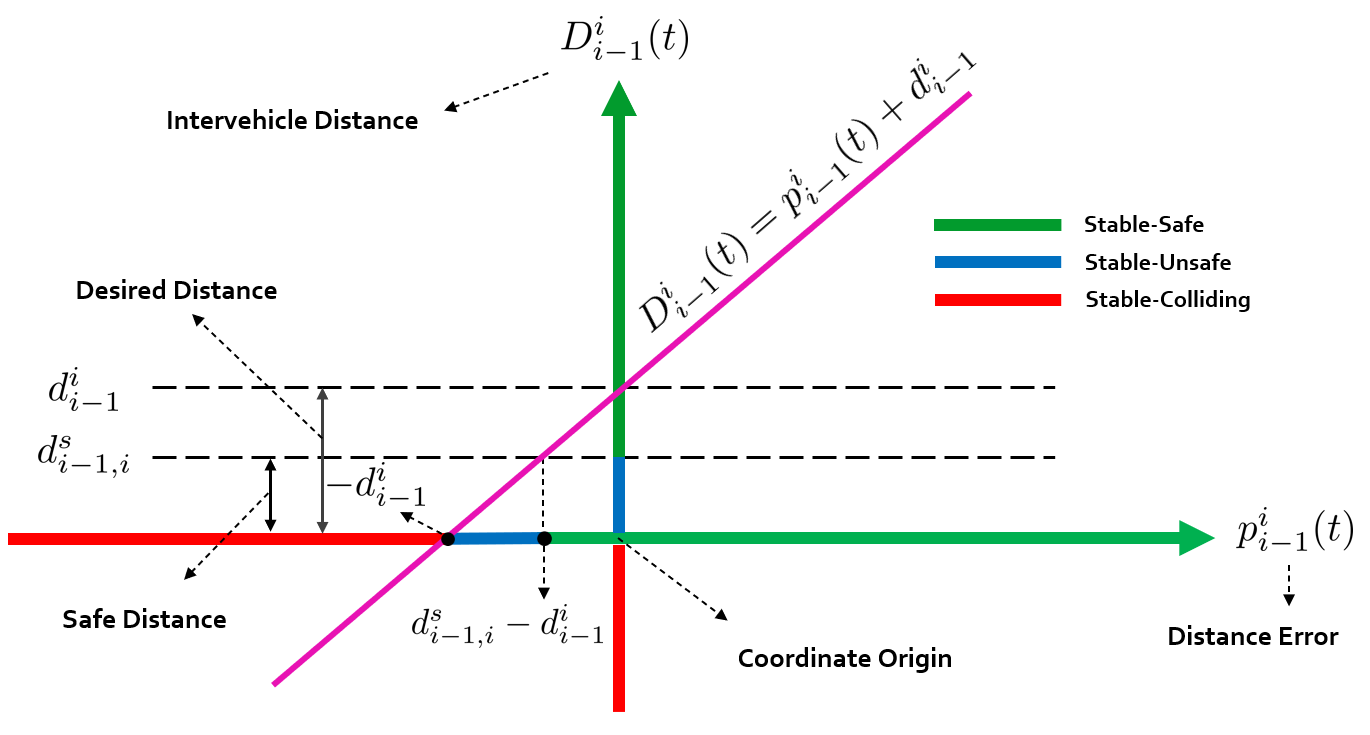}}
\caption{Illustrating Safe, Unsafe, and Colliding scenarios.}
\label{area_diagram}
\end{center}
\end{figure}
\begin{remark}
The reason for defining these scenarios is to identify later, based on different communication topologies, which configurations lead to better IDs within the platoon. Later,the distance error formulation will be provided [see equation \eqref{impulese}].  
\end{remark}
\subsection{\textbf{Connection Types Between Vehicles}}
\label{zedtaii}
Fig. \ref{interv-conn} illustrates the possible connections between adjacent vehicles. For later usage, we define:
\begin{equation}
\zeta_{i-1}^{i} = \begin{cases}
 0 & \text{ if }\hspace{1cm} z_{i-1}^{i}=0\text{ and } z_{i}^{i-1}=1 \\
 z_{i-1}^{i} & \text{ if }\hspace{1cm} z_{i-1}^{i}=1 \text{ and } z_{i}^{i-1}=1
 \\
 0 & \text{ if }\hspace{1cm} z_{i-1}^{i}=0 \text{ and } z_{i}^{i-1}=0 \\
 z_{i-1}^{i} & \text{ if }\hspace{1cm} z_{i-1}^{i}=1\text{ and } z_{i}^{i-1}=0
\end{cases}  
\label{zetaii}
\end{equation}
where, for example, $z_{i-1}^{i}=1$ indicates that the $(i-1)^{th}$ vehicle receives information from the $i^{th}$ vehicle, while $z_{i-1}^{i}=0$ indicates that it does not.
\begin{figure}[htbp!]
\begin{center}
\resizebox{0.6\hsize}{!}{\includegraphics*{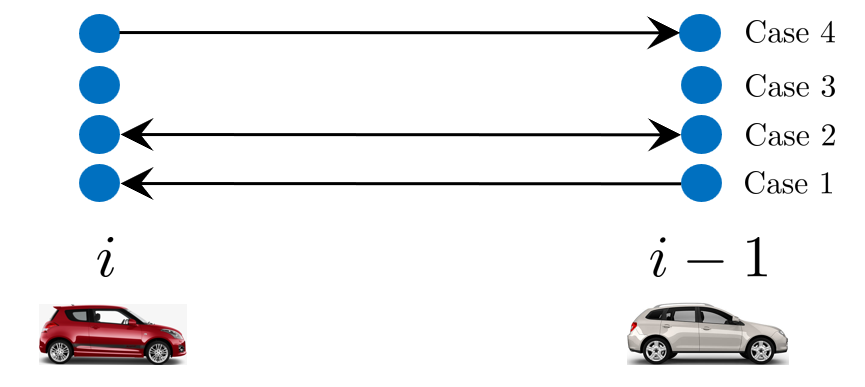}}
\caption{Connection types between neighboring vehicles.}
\label{interv-conn}
\end{center}
\end{figure}   
\subsection{\textbf{Rigid Communication Topologies (RCTs)}}
\label{comms}
Rigid communication topologies (RCTs) refer to fixed, predefined structures that dictate how vehicles within the platoon communicate with each other. These topologies define the pathways for information flow and control signals between the vehicles. We split RCTs into three categories:

1. \textbf{Typical Unidirectional RCTs (TURCTs)}: In this category, all FVs receive information only from vehicles ahead. Figs. \ref{typicalu} shows  TURCTs, in each of which there is a certain practice in the communication between vehicles. For instance, in two-predecessor-following (TPF) topology, each vehicle sends its information to its two immediate followers.
\begin{figure}[htbp!]
\begin{center}
\resizebox{0.85\hsize}{!}{\includegraphics*{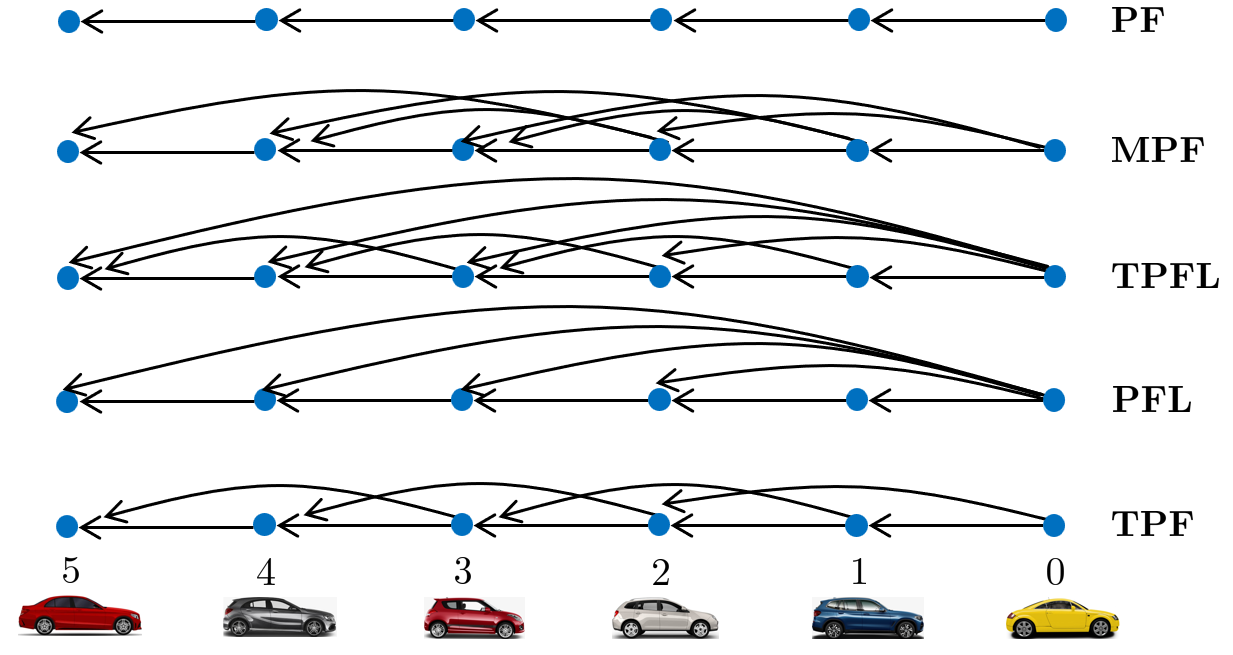}}
\caption{Some common TURCTs between vehicles.}
\label{typicalu}
\end{center}
\end{figure}

2. \textbf{Typical Bidirectional RCTs (TBRCTs)}: 
In TBRCTs, FVs can receive information from both vehicles behind and ahead. Figs. \ref{typicalb} shows TBRCTs, in each of which there is a certain pattern in the communication between vehicles. For instance, in the two-predecessor-single-following (TPSF) topology, each FV receives information from its two immediate predecessors and one immediate follower.
\begin{figure}[htbp!]
\begin{center}
\resizebox{0.85\hsize}{!}{\includegraphics*{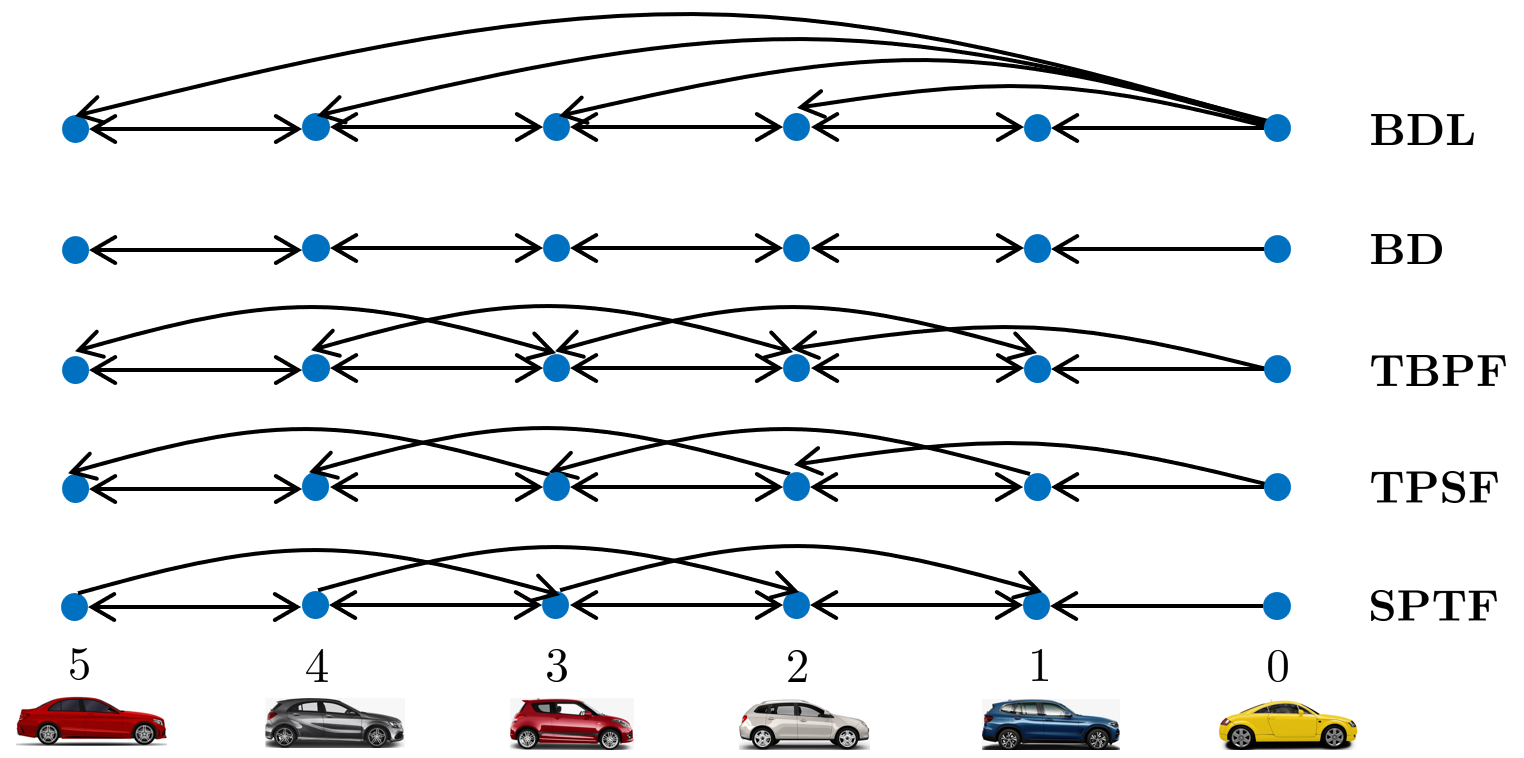}}
\caption{Some common TBRCTs between vehicles.}
\label{typicalb}
\end{center}
\end{figure}

3. \textbf{Nontypical RCTs (NRCTs)}: In NRCTs, there is no specific pattern in the communication between vehicles [see Fig. \ref{nontypical}].
\begin{figure}[htbp!]
\begin{center}
\resizebox{0.85\hsize}{!}{\includegraphics*{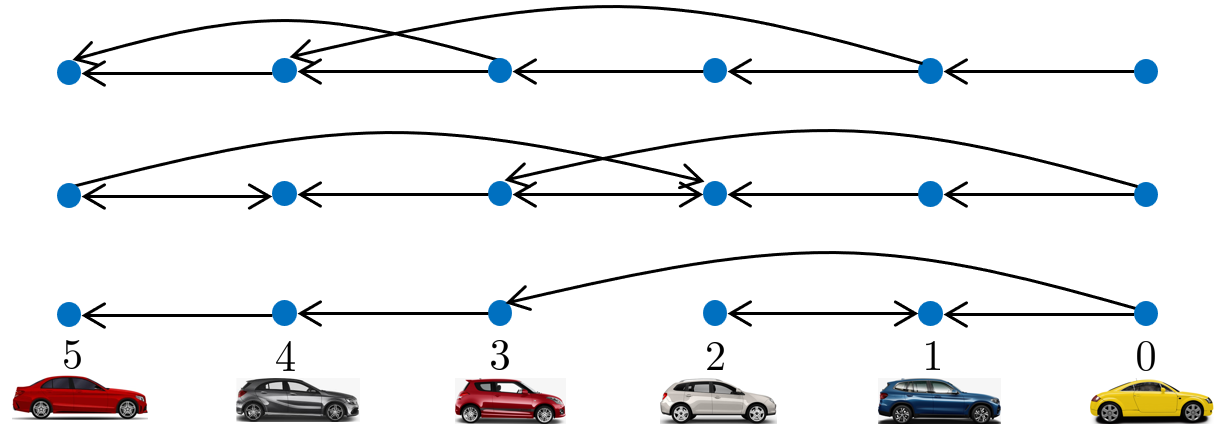}}
\caption{Some arbitrary NRCTs between vehicles.}
\label{nontypical}
\end{center}
\end{figure}
\subsection{\textbf{Definition of Sets $\mathbb{R}_{i-1},\;\mathbb{R}_{i},\;\mathbb{I}_{i-1}$, and $\mathbb{I}_{i}$}} For later usage, we introduce the following sets associated with the pair $(i-1,i)$:
\begin{itemize}[left=0pt]
\item $\mathbb{I}_{i}$: Represents vehicles from which the $i^{th}$ follower receives information.  
\item $\mathbb{I}_{i-1}$: Represents vehicles from which the $(i-1)^{th}$ follower gathers information.
\item $\mathbb{R}_{i}$: Denotes vehicles, excluding vehicle $i-1$, that share information with the $i^{th}$ follower.
\item $\mathbb{R}_{i-1}$: Denotes vehicles, excluding vehicle $i$, that provide information to the $(i-1)^{th}$ follower.
\end{itemize}
For instance, in TPLF topology from Fig. \ref{typicalu} and for the pair $(2,3)$, $\mathbb{I}_{3} = \{2,\;1,\;0\}$, $\mathbb{I}_{2} = \{1,\;0\}$, $\mathbb{R}_{3} = \{1,\;0\}$, and $\mathbb{R}_{2} = \{1,\;0\}$.
\subsection{\textbf{Accumulative CGs of Neighboring Vehicles}}
\label{accum}
For later usage, under any RCT shown in Figs. \ref{typicalu}-\ref{nontypical}, let
\begin{equation}
\begin{cases}
\begin{split}
&\bar{k}_{i-1}^{i}\triangleq\left(\frac{1}{\tau_{i}}|\mathbb{I}_{i}^{\leq i-1}|+\frac{1}{\tau_{i-1}}|\mathbb{I}_{i-1}^{\geq i}|\right)k \\&
\bar{b}_{i-1}^{i}\triangleq\left(\frac{1}{\tau_{i}}|\mathbb{I}_{i}^{\leq i-1}|+\frac{1}{\tau_{i-1}}|\mathbb{I}_{i-1}^{\geq i}|\right)b\\&
\bar{h}_{i-1}^{i}\triangleq\left(\frac{1}{\tau_{i}}|\mathbb{I}_{i}^{\leq i-1}|+\frac{1}{\tau_{i-1}}|\mathbb{I}_{i-1}^{\geq i}|\right)h+\frac{1}{\tau_{i}}
\end{split}
\end{cases}
\label{pbc}
\end{equation}
be defined as the accumulative  control gains, where the sets $\mathbb{I}_{i}^{\leq i-1}$ and $\mathbb{I}_{i-1}^{\geq i}$ are defined as $\mathbb{I}_{i}^{\leq i-1}\triangleq\{j|j\in\mathbb{I}_{i} \text{ \& }j\leq i-1\}$ and $\mathbb{I}_{i-1}^{\geq i}\triangleq\{j|j\in\mathbb{I}_{i-1} \text{ \& }j\geq i\}$, respectively. The notation $|.|$ shows the cardinality of the relevant sets. 
\section{\textbf{Problem Formulation}}
\label{section3}
For the compactness of later presentations, $\tilde{\mathbf{X}}_{i}^{j}$ [see \eqref{contr}] elements [see \eqref{newstates}] can be rewritten as $p_{i}^{j}=\tilde{x}_{i}-\tilde{x}_{j}$, $v_{i}^{j}=\dot{\tilde{x}}_{i}-\dot{\tilde{x}}_{j}$, and $a_{i}^{j}=\ddot{\tilde{x}}_{i}-\ddot{\tilde{x}}_{j}$ where, for $m\in\{j,i\}$,
$\tilde{x}_{m}\triangleq x_{m}- x_{0}+\sum_{\kappa=1}^{m}L_{\kappa-1}+d_{\kappa-1}^{\kappa}$, $\dot{\tilde{x}}_{m}\triangleq \dot{x}_{m}-v_{0}=\dot{x}_{m}-\dot{x}_{0}$, and $\ddot{\tilde{x}}_{m}\triangleq \ddot{x}_{m}-a_{0}=\ddot{x}_{m}-\ddot{x}_{0}$. Accordingly, $\tilde{\mathbf{X}}_{m}\triangleq [\tilde{x}_{m};\; \dot{\tilde{x}}_{m};\;\ddot{\tilde{x}}_{m}]$,  $\tilde{\mathbf{X}}_{i}^{j}= \tilde{\mathbf{X}}_{i}-\tilde{\mathbf{X}}_{j}$, and $\dot{\tilde{\mathbf{X}}}_{i}^{j}\triangleq[v_{i}^{j};\;a_{i}^{j};\;\mathbb{j}_{i}^{j}]$ where  $\tilde{\mathbb{j}}_{i}^{j}\triangleq \dddot{\tilde{x}}_{i}-\dddot{\tilde{x}}_{j}$. Note that for $m\in\{j,i\}$, $\dddot{\tilde{x}}_{m}\triangleq  \dddot{x}_{m}-\dot{a}_{0}=\dddot{x}_{m}-\dddot{x}_{0}$. Now, given   $\ddot{x}_{i}=\ddot{\Tilde{x}}_{i}+a_{0}$ and $\dddot{x}_{i}=\dddot{\Tilde{x}}_{i}+\dot{a}_{0}$,  plugging \eqref{contr} into \eqref{fl} and doing some mathematical rearrangements yields:
\begin{equation}
    \begin{split}
        \dddot{\Tilde{x}}_{i}&=-\frac{1}{\tau_{i}}\mathbb{K}_{i}\Tilde{\mathbf{X}}_{i}        +\sum_{j\in\mathbb    {I}_{i}}\frac{1}{\tau_{i}}\mathbf{K}\Tilde{\mathbf{X}}_{j}+\epsilon_{i};\quad i=1,\dots,n
    \end{split}
    \label{ssmx}
\end{equation}
in which $\epsilon_{i}\triangleq-\frac{1}{\tau_{i}}a_{0}(t)-\dot{a}_{0}(t)$ and $\mathbb{K}_{i}\triangleq\left[\bar{k}_{i},\;\bar{b}_{i},\;\bar{h}_{i}\right]$  such that  
\begin{equation}
\begin{split}
\bar{k}_{i}\triangleq |\mathbb{I}_{i}|k,\hspace{0.2cm} \bar{b}_{i}\triangleq |\mathbb{I}_{i}|b,\hspace{0.2cm} \bar{h}_{i}\triangleq  1+|\mathbb{I}_{i}|h; \quad i=1,\dots,n
\end{split}    
\label{kmathbb}
\end{equation} 
where $|\mathbb{I}_{i}|$ denotes the cardinality of the set $\mathbb{I}_{i}$. 
\subsection{\textbf{Coupled Dynamics of Neighboring Vehicles}}
\label{cdyn}
Recalling the section \ref{pstate}, the following theorem outlines the intervehicle coupled dynamics:
\begin{theorem}
\label{th1}
Given preliminaries \ref{comms} and \ref{accum}, assuming the relative acceleration as the output of each coupled dynamics (let it be denoted as $\tilde{y}_{i-1}^{i}$), the state-space representation of the intervehicle dynamics, under any RCT, is governed by
\begin{equation}
\begin{cases}    
\dot{\Tilde{\mathbf{X}}}_{i-1}^{i}=\underbrace{\begin{bmatrix}
0 & 1 & 0\\
0 & 0 & 1\\
-\bar{k}_{i-1}^{i} & -\bar{b}_{i-1}^{i} & -\bar{h}_{i-1}^{i}
\end{bmatrix}}_{\triangleq \mathbf{A}_{i-1}^{i}} \Tilde{\mathbf{X}}_{i-1}^{i}+\underbrace{\begin{bmatrix}
0\\0\\1    
\end{bmatrix}}_{\triangleq\mathbf{B}_{i-1}^{i}}\Tilde{u}_{i-1}^{i}\\
\tilde{y}_{i-1}^{i}=\underbrace{[0,\;0,\;1]}_{\triangleq\mathbf{C}_{i-1}^{i}}\tilde{\mathbf{X}}_{i-1}^{i}
\end{cases}
\label{col_dyn}
\end{equation}
in which $\tilde{u}_{i-1}^{i}$, the input of coupled dynamics,  is according to
\begin{equation}
\begin{split}
\tilde{u}&_{i-1}^{i}
=\frac{1}{\tau_{i-1}}
\sum_{j\in\mathbb{R}_{i-1}^{<i-1}}\sum_{\kappa=j+1}^{i-1}\mathbf{K}\Tilde{\mathbf{X}}_{\kappa-1}^{\kappa}-\frac{1}{\tau_{i-1}}
\sum_{j\in\mathbb{R}_{i-1}^{>i}}\sum_{\kappa=i+1}^{j}\mathbf{K}\Tilde{\mathbf{X}}_{\kappa-1}^{\kappa}\\&+\frac{1}{\tau_{i}}
\sum_{j\in\mathbb{R}_{i}^{>i}}\sum_{\kappa=i+1}^{j}\mathbf{K}\Tilde{\mathbf{X}}_{\kappa-1}^{\kappa}
-\frac{1}{\tau_{i}}
\sum_{j\in\mathbb{R}_{i}^{<i-1}}\sum_{\kappa=j+1}^{i-1}\mathbf{K}\Tilde{\mathbf{X}}_{\kappa-1}^{\kappa}-U(\gamma-i)\epsilon_{i} \\&-U(i-\gamma) \frac{1}{\tau_{i-1}}\mbox{\boldmath$\tau$}_{i-1}^{i}\sum_{\kappa=1}^{i-1}\Tilde{\mathbf{X}}_{\kappa-1}^{\kappa}+U(i-\gamma)\frac{\tau_{i-1}-\tau_{i}}{\tau_{i-1}\tau_{i}}a_{0}
\end{split}
\label{inputi-i+1}
\end{equation}
where in the $U(\gamma - i)$ and $U(i - \gamma)$, the parameter $\gamma$ is a real number between one and two, i.e., $1 < \gamma < 2$, where $U(\cdot)$ denotes the unit step function. Also, $\mbox{\boldmath$\tau$}_{i-1}^{i}=\left[0,\;0,\;\frac{\tau_{i-1}-\tau_{i}}{\tau_{i}}\right]$. 
\end{theorem}
\begin{proof}
Please check Appendix \ref{appendixA}.    
\end{proof}
\subsection{\textbf{Initial Conditions}}
\label{xcde}
\label{expl_inclu}Here, the aim is to extract the role of initial conditions in the input signal [see  \eqref{inputi-i+1}] of the coupled dynamics. At the time of platooning triggering, each vehicle has its own initial position, velocity, and acceleration. Thus, given \eqref{newstates}, the initial coupled states associated with relative positions ($x_{i-1}^{i}=x_{i-1}-x_{i}$), relative velocities, and relative accelerations of the adjacent vehicles are defined as $\varrho_{i-1}^{i}\triangleq x_{i-1}^{i}(0)$, $\nu_{i-1}^{i}\triangleq v_{i-1}^{i}(0)$, and $\varphi_{i-1}^{i}\triangleq a_{i-1}^{i}(0)$, respectively. Accordingly,  $\mathbf{\Xi}_{i-1}^{i}\triangleq[\varrho_{i-1}^{i};\nu_{i-1}^{i};\varphi_{i-1}^{i}]$, and thus in the Laplacian domain we get:
\begin{equation}
\begin{split}
\Tilde{\mathbf{X}}_{i-1}^{i}(s)
  =\underbrace{\begin{bmatrix}
  \frac{1}{s^{2}}\\ \frac{1}{s} \\1
  \end{bmatrix}}_{\triangleq\mathbf{T}_{1}}a_{i-1}^{i}(s)+\underbrace{\begin{bmatrix}
      \frac{1}{s}&\frac{1}{s^{2}}&0\\0&\frac{1}{s}&0\\0&0&0
  \end{bmatrix}}_{\triangleq \mathbf{T}_{2}}\mathbf{\Xi}_{i-1}^{i}-\underbrace{\begin{bmatrix}
  \frac{1}{s}
      \\
      0
      \\
      0
  \end{bmatrix}}_{\triangleq \mathbf{T}_{3}} \omega_{i-1}^{i}
\end{split}
\label{tr}
\end{equation}
using which, the input for the coupled dynamics [see \eqref{inputi-i+1}], in Laplacian domain, would be as
\begin{equation}
\begin{split}
&\Tilde{u}_{i-1}^{i}(s)=\frac{hs^2+bs+k}{s^2}\Pi_{i-1}^{i}(s)+\frac{\Theta_{i-1}^{i}s+\Gamma_{i-1}^{i}}{s^{2}}\\&-\frac{U(i-\gamma)(\tau_{i-1}-\tau_{i})}{\tau_{i-1}\tau_{i}}\left(\mathcal{M}_{i-1}^{i}(s)-a_{0}(s)\right)+\frac{U(\gamma-i)(1+\tau_{i} s)a_{0}(s)}{\tau_{i}}
\end{split}
\label{err32}
\end{equation} 
where $\mathcal{M}_{i-1}^{i}(s)=\sum_{\kappa=1}^{i-1}\Tilde{a}_{\kappa-1}^{\kappa}(s)$ and
\begin{equation}
 \begin{split}
\Pi_{i-1}^{i}(s)&\triangleq \frac{1}{\tau_{i-1}}\left[\sum_{j\in\mathbb{R}_{i-1}^{<i-1}}\sum_{\kappa=j+1}^{i-1}a_{\kappa-1}^{\kappa}(s)
-\sum_{j\in\mathbb{R}_{i-1}^{>i}}\sum_{\kappa=i+1}^{j}a_{\kappa-1}^{\kappa}(s)\right]
\\&+\frac{1}{\tau_{i}}\left[\sum_{j\in\mathbb{R}_{i}^{>i}}\sum_{\kappa=i+1}^{j} a_{\kappa-1}^{\kappa}(s)
-\sum_{j\in\mathbb{R}_{i}^{<i-1}}\sum_{\kappa=j+1}^{i-1}a_{\kappa-1}^{\kappa}(s)\right] 
\end{split}
\label{piiip}
\end{equation}
Also, in \eqref{err32}, the scalar values $\Theta_{i-1}^{i}$ and $\Gamma_{i-1}^{i}$ are defined as   
$\Theta_{i-1}^{i}\triangleq  k\theta_{t}^{i-1,i}+b\nu_{t}^{i-1,i}$ and $\Gamma_{i-1}^{i}\triangleq k\nu_{t}^{i-1,i}$
in which
\begin{equation}
\begin{split}
\theta_{t}^{i-1,i}&\triangleq \frac{1}{\tau_{i-1}}\left[\sum_{j\in\mathbb{R}_{i-1}^{<i-1}}\sum_{\kappa=j+1}^{i-1}\theta_{\kappa-1}^{\kappa}
-\sum_{j\in\mathbb{R}_{i-1}^{>i}}\sum_{\kappa=i+1}^{j}\theta_{\kappa-1}^{\kappa}\right]
\\&+\frac{1}{\tau_{i}}\left[\sum_{j\in\mathbb{R}_{i}^{>i}}\sum_{\kappa=i+1}^{j}\theta_{\kappa-1}^{\kappa}
-\sum_{j\in\mathbb{R}_{i}^{<i-1}}\sum_{\kappa=j+1}^{i-1}\theta_{\kappa-1}^{\kappa}\right] 
\end{split}
\label{erty3}
\end{equation}
and
\begin{equation}
\begin{split}
\nu_{t}^{i-1,i}&\triangleq \frac{1}{\tau_{i-1}}\left[\sum_{j\in\mathbb{R}_{i-1}^{<i-1}}\sum_{\kappa=j+1}^{i-1}\nu_{\kappa-1}^{\kappa}
-\sum_{j\in\mathbb{R}_{i-1}^{>i}}\sum_{\kappa=i+1}^{j}\nu_{\kappa-1}^{\kappa}\right]
\\&+\frac{1}{\tau_{i}}\left[\sum_{j\in\mathbb{R}_{i}^{>i}}\sum_{\kappa=i+1}^{j}\nu_{\kappa-1}^{\kappa}
-\sum_{j\in\mathbb{R}_{i}^{<i-1}}\sum_{\kappa=j+1}^{i-1}\nu_{\kappa-1}^{\kappa}\right] 
\end{split}
\label{erty32}
\end{equation}
where $\theta_{\kappa-1}^{\kappa}\triangleq\varrho_{\kappa-1}^{\kappa}-\omega_{\kappa-1}^{\kappa}$. Note that, for the pair $(0,1)$, $\mathbb{R}_{i-1}^{<i-1}=\mathbb{R}_{i-1}^{>i}=\mathbb{R}_{i}^{<i-1}=\emptyset$ in which $\emptyset$ denotes the empty set.
\subsection{\textbf{Distance Error, Relative Velocity, and Relative Acceleration Between Neighboring Vehicles}}
\label{CDPSs}
Given the coupled dynamics for neighboring vehicles [see section \ref{cdyn}] and the their input signals [see section \ref{expl_inclu}], in this section we formulate the coupled dynamics that dictate the relative acceleration; $a_{i-1}^{i}(s)$, relative velocity; $v_{i-1}^{i}(s)$, and distance error; $p_{i-1}^{i}(s)$, between the adjacent vehicles. Before that, please note that given \eqref{col_dyn}, the transfer function, denoted by $\mathcal{G}_{i-1}^{i}(s)$, between the input and output of the coupled dynamics would be $\mathcal{G}_{i-1}^{i}(s)=\mathbf{C}_{i-1}^{i}(s\mathbf{I}_{3}-\mathbf{A}_{i-1}^{i})^{-1} \mathbf{B}_{i-1}^{i}=s^{2}\Upsilon_{i-1}^{i}(s)$  
where $\Upsilon _{i-1}^{i}(s)\triangleq 1/(s^{3}+\bar{h}_{i-1}^{i}s^{2}+\bar{b}_{i-1}^{i}s+\bar{k}_{i-1}^{i})$.
\begin{theorem}
For $i=1,\dots,n$, $a_{i-1}^{i}(s)$, $v_{i-1}^{i}(s)$, and $p_{i-1}^{i}(s)$ are governed according to  
\begin{equation}
\begin{split}
&a_{i-1}^{i}(s)=\underbrace{(hs^2+bs+k)\Pi_{i-1}^{i}(s)\Upsilon_{i-1}^{i}(s)}_{\text{out of the coupled states where the pair$\neq (i-1,i)$}}\\&-\underbrace{\frac{U(i-\gamma)(\tau_{i-1}-\tau_{i})}{\tau_{i-1}\tau_{i}}\left(\mathcal{M}_{i-1}^{i}(s)-a_{0}(s)\right)s^{2}\Upsilon_{i-1}^{i}(s)}_{\text{out of the heterogeneity of the pair $(i-1,i)$}}\\&+\underbrace{\mathcal{H}_{i-1,i}^{a}(s)\Upsilon _{i-1}^{i}(s)}_{\text{ out of all initial conditions}}+\underbrace{\frac{U(\gamma-i)(1+\tau_{i} s)a_{0}(s)s^{2}\Upsilon_{i-1}^{i}(s)}{\tau_{i}}}_{\text{ out of the leader's acceleration/jerk}}
\end{split}
\label{err3c2}
\end{equation} 
$v_{i-1}^{i}(s)= (a_{i-1}^{i}(s)+\nu_{i-1}^{i})/s$, and $p_{i-1}^{i}(s)=(v_{i-1}^{i}(s)+\theta_{i-1}^{i})/s$ where $\mathcal{H}_{i-1,i}^{a}(s)\triangleq\Lambda_{2,a}^{{i-1,i}}s^2+\Lambda_{1,a}^{i-1,i}s+\Lambda_{0,a}^{i-1,i}$ in which 
\begin{equation}
\begin{cases}    
\Lambda_{2}^{i-1,i}\triangleq\varphi_{i-1}^{i}\\
\Lambda_{1}^{i-1,i}\triangleq\Theta_{i-1}^{i}-\theta_{i-1}^{i}\bar{k}_{i-1}^{i} -\nu_{i-1}^{i}\bar{b}_{i-1}^{i}
\\
\Lambda_{0}^{i-1,i}\triangleq \Gamma_{i-1}^{i}-\nu_{i-1}^{i}\bar{k}_{i-1}^{i}
\end{cases}
\label{ini_param}
\end{equation}
and $\Theta_{i-1}^{i}$ and $\Gamma_{i-1}^{i}$ are already defined [see \ref{xcde}]. Also, $\bar{k}_{i-1}^{i}$ and $\bar{b}_{i-1}^{i}$ are defined in subsection \ref{accum}.
\end{theorem}
\begin{proof}
Please check Appendix \ref{appendixB}.   
\end{proof}
\subsection{\textbf{Decoupled Dynamics of Distance Error, Relative Velocity, and Relative Acceleration Between Neighboring Vehicles}} In section \ref{CDPSs}, we formulated the coupled dynamics. To shed more light on the them, note that in \eqref{err3c2}, the calculation of $a_{i-1}^{i}(s)$ itself is dependant on knowing the term $\Pi_{i-1}^{i}(s)$ [see \eqref{piiip}] which comes from the relative accelerations of the other pairs. Thus, we need to decouple the relative accelerations of all pairs. To this end,  in \eqref{err3c2},  we define $\Psi_{i-1,i}^{a}(s)$ and $\mho_{i-1,i}^{a}(s)$, as the responses to initial conditions and the leader vehicle's acceleration/jerk, according to $\Psi_{i-1,i}^{a}(s)\triangleq\mathcal{H}_{i-1,i}^{a}(s)\Upsilon _{i-1}^{i}(s)
$ and 
\begin{equation}
\mho_{i-1,i}^{a}(s)=\begin{cases}
\begin{split}  
\frac{U(\gamma-i)(1+\tau_{i} s)a_{0}(s)s^{2}\Upsilon_{i-1}^{i}(s)}{\tau_{i}}
\end{split}; &\text{ if }\quad i=1
\\
\begin{split}    
\frac{U(i-\gamma)(\tau_{i-1}-\tau_{i})a_{0}(s)s^{2}\Upsilon_{i-1}^{i}(s)}{\tau_{i-1}\tau_{i}}
\end{split}
; &\text{ if }\quad i\neq 1
\end{cases}   
\end{equation}
Expanding \eqref{err3c2} for $i=1,\dots,n$ and constructing them all in a matrix form, a mapping can be derived as shown below:
\begin{equation}
\begin{cases}
\mathbf{a}(s) \triangleq \mathbf{Q}^{-1}(s)\left(\mathbf{\Psi}_{a}(s)+\mbox{\boldmath$\mho$}_{t}(s)\right)
\\
 \mathbf{v}(s)=\frac{1}{s}(\mathbf{a}(s)+\mbox{\boldmath$\nu$}) \\
  \mathbf{p}(s)=\frac{1}{s}(\mathbf{v}(s)+\mbox{\boldmath$\theta$})=\frac{1}{s^{2}}\mathbf{a}(s)+\frac{1}{s^{2}}\mbox{\boldmath$\nu$}+\frac{1}{s}\mbox{\boldmath$\theta$}
 \end{cases}
\label{weror}    
\end{equation}
where $\mathbf{\Psi}_{a}(s)\triangleq\left[\Psi_{0,1}^{a}(s);\; \Psi_{1,2}^{a}(s);\;\dots;\; \Psi_{n-1,n}^{a}(s)\right]$ and $\mathbf{a}(s)\triangleq\left[a_{0}^{1}(s);\; a_{1}^{2}(s);\;\dots;\; a_{n-1}^{n}(s)\right]$. Also, $\mathbf{Q}(s)\in\mathbb{C}^{n\times n}$ is the mapping matrix whose elements are according to ($i,\kappa=1,\dots,n$)
\begin{equation}
 q_{i\kappa}(s)=\begin{cases}
-\mathbf{K}_{i-1,i}^{\kappa -}\mathbf{T}_{1}\mathcal{G}_{i-1}^{i}(s)    & \text{ if  }\hspace{0.6cm} i>\kappa\\ 
1    & \text{ if  }\hspace{0.6cm} i=\kappa\\
-\mathbf{K}_{i-1,i}^{\kappa +}\mathbf{T}_{1}\mathcal{G}_{i-1}^{i}(s)    & \text{ if  }\hspace{0.6cm} i<\kappa
 \end{cases}  
 \label{qmatrix}
\end{equation}
in which 
\begin{equation}
\begin{cases}
\begin{split}
\mathbf{K}_{i-1,i}^{\kappa-}\triangleq\frac{1}{\tau_{i-1}}|\mathbb{R}_{i-1}^{\leq \kappa-1}|\mathbf{K}-\frac{1}{\tau_{i}}|\mathbb{R}_{i}^{\leq \kappa-1}|\mathbf{K}-\frac{U(i-\gamma)}{\tau_{i-1}}\mbox{\boldmath$\tau$}_{i-1}^{i}
\end{split}
\\
\begin{split}
\mathbf{K}_{i-1,i}^{\kappa +}\triangleq\frac{1}{\tau_{i}}|\mathbb{R}_{i}^{\geq \kappa}|\mathbf{K}-\frac{1}{\tau_{i-1}}|\mathbb{R}_{i-1}^{\geq \kappa}|\mathbf{K}
\end{split}
\end{cases}
\label{kiip}
\end{equation}
where $|.|$ denotes the cardinality of sets, and considering the pair 
$(i-1,i)$, $\mathbb{R}_{i-1}^{\leq \kappa-1}\triangleq\{\epsilon|\epsilon\in\mathbb{R}_{i-1} \text{ \& } \epsilon\leq\kappa-1\}$,  $\mathbb{R}_{i}^{\leq \kappa-1}\triangleq\{\epsilon|\epsilon\in\mathbb{R}_{i} \text{ \& } \epsilon\leq\kappa-1\}$, $\mathbb{R}_{i}^{\geq \kappa}\triangleq\{\epsilon|\epsilon\in\mathbb{R}_{i} \text{ \& } \epsilon\geq\kappa\}$, and $\mathbb{R}_{i-1}^{\geq \kappa}\triangleq\{\epsilon|\epsilon\in\mathbb{R}_{i-1} \text{ \& } \epsilon\geq\kappa\}$. Also, in \eqref{weror}, $\mbox{\boldmath$\mho$}_{t}(s)\in\mathbb{C}^{n\times 1}$ such that $\mbox{\boldmath$\mho$}_{t}(s) \triangleq [\mho_{0,1}^{a}(s);\mho_{1,2}^{a}(s);\dots;\mho_{n-1,n}^{a}(s)]$.
Moreover, in \eqref{weror}, the second and third relations are due to the fact that $v_{i-1}^{i}(s)=(a_{i-1}^{i}(s)+\nu_{i-1}^{i})/s$ and  $p_{i-1}^{i}(s)=(v_{i-1}^{i}(s)+\theta_{i-1}^{i})/s$. Finally, in \eqref{weror}, $\mathbf{v}(s)\triangleq[v_{0}^{1}(s);\; v_{1}^{2}(s);\;\dots;\; v_{n-1}^{n}(s)]$,  $\mbox{\boldmath$\nu$}\triangleq[\nu_{0}^{1};\; \nu_{1}^{2};\;\dots;\; \nu_{n-1}^{n}]$, $\mathbf{p}(s)\triangleq[p_{0}^{1}(s);\; p_{1}^{2}(s);\;\dots;\; p_{n-1}^{n}(s)]$, and $\mbox{\boldmath$\theta$}\triangleq[\theta_{0}^{1};\; \theta_{1}^{2};\;\dots;\; \theta_{n-1}^{n}]$.
\subsubsection{\textbf{Formulations of Distance Error, Relative Velocity, and Relative Acceleration 
 Between Vehicles $i-1$ and $i$}} As mentioned earlier [see Remark \ref{re1re}], one of the main objectives of this work was to formulate distance errors, relative velocities, and relative accelerations between neighboring vehicles. Here, we present these formulations. Given \eqref{weror}, we get ($i=1,\dots,n$)
 \begin{equation}
\begin{split} 
a_{i-1}^{i}(s)=\nu_{i-1}^{i} +\sum_{\kappa=1}^{n}q^{-1}_{i\kappa }(s)\left(\Psi_{\kappa-1,\kappa}^{a}(s)+\mho_{\kappa-1,\kappa}^{a}(s)\right)
\end{split}
\label{impuleseRA}
\end{equation}
where $a_{i-1}^{i}(s)$ is the relative acceleration in the Laplacian domain. Therefore, it is possible to obtain $a_{i-1}^{i}(t)$ through the impulse response of the $a_{i-1}^{i}(s)$. Similarly, the relative velocity would be according to ($i=1,\dots,n$)
\begin{equation}
\begin{split} 
v_{i-1}^{i}(s)=\frac{1}{s}\left[\nu_{i-1}^{i} +\sum_{\kappa=1}^{n}q^{-1}_{i\kappa }(s)\left(\Psi_{\kappa-1,\kappa}^{a}(s)+\mho_{\kappa-1,\kappa}^{a}(s)\right)\right]
\end{split}
\label{impuleseRV}
\end{equation}
and the distance error is calculated by ($i=1,\dots,n$)
\begin{equation}
\begin{split} 
p_{i-1}^{i}(s)&=\frac{1}{s}\theta_{i-1}^{i}+\frac{1}{s^2}\left[\nu_{i-1}^{i} +\sum_{\kappa=1}^{n}q^{-1}_{i\kappa }(s)\left(\Psi_{\kappa-1,\kappa}^{a}(s)+\mho_{\kappa-1,\kappa}^{a}(s)\right)\right]
\end{split}
\label{impulese}
\end{equation}
where $q_{i\kappa}^{-1}(s)$ are the elements of the matrix $\mathbf{Q}^{-1}(s)$.
\subsubsection{\textbf{Affecting Parameters}}The mapping matrix, influenced by control gains, vehicle features, and communication topology, is also affected by initial conditions through $\mathbf{\Psi}_{a}(s)$. Additionally, distance errors, relative velocities, and accelerations are influenced by the leader vehicle's acceleration/jerk as well.
\subsection{\textbf{Vehicles' Control/Engine Inputs, and Accelerations/Jerks}}
\label{CEIAJ}
To analyze energy consumption and passenger comfort, we need to formulate vehicle control/engine inputs, accelerations, and jerks. By splitting $\mathbb{I}_{i}$ into $\mathbb{I}_{i}^{>i}$ (vehicles ahead) and $\mathbb{I}_{i}^{<i}$ (vehicles behind), the control input can be rewritten as\begin{equation}
u_{i} = -\mathbf{K}\sum_{j\in\mathbb{I}_{i}^{>i}}\sum_{\kappa=i+1}^{j}\Tilde{\mathbf{X}}_{\kappa-1}^{\kappa}+\mathbf{K}\sum_{j\in\mathbb{I}_{i}^{<i}}\sum_{\kappa=j+1}^{i}\Tilde{\mathbf{X}}_{\kappa-1}^{\kappa}
\label{ytt5}
\end{equation}
By $\Tilde{\mathbf{X}}_{t}(s)\triangleq\left[\Tilde{\mathbf{X}}_{0}^{1}(s);\; \Tilde{\mathbf{X}}_{1}^{2}(s);\;\dots;\; \Tilde{\mathbf{X}}_{n-1}^{n}(s)\right]$ and noting \eqref{tr}, we get 
\begin{equation}
\Tilde{\mathbf{X}}_{t}(s)=\left(\mathbf{I}_{n}\otimes\mathbf{T}_{1}\right)\mathbf{a}(s)-\left(\mathbf{I}_{n}\otimes\mathbf{T}_{2}\right)\mathbf{\Xi}_{t}-\left(\mathbf{I}_{n}\otimes\mathbf{T}_{3}\right)\mbox{\boldmath$\omega$}_{t}
\label{eeeqq}
\end{equation}
in which $\mathbf{\Xi}_{t}\triangleq\left[\mathbf{\Xi}_{0}^{1};\mathbf{\Xi}_{1}^{2};\dots;\mathbf{\Xi}_{n-1}^{n}\right]$ and $\mbox{\boldmath$\omega$}_{t}\triangleq\left[\omega_{0}^{1};\omega_{1}^{2};\dots;\omega_{n-1}^{n}\right]$. Now, substituting $\mathbf{a}(s)$ from  \eqref{weror} into \eqref{eeeqq}, yields
\begin{equation}
\begin{split}  
\Tilde{\mathbf{X}}_{t}(s)&=\left(\mathbf{I}_{n}\otimes\mathbf{T}_{1}\right)\mathbf{Q}^{-1}(s)\left(\mathbf{\Psi}_{a}(s)+\mbox{\boldmath$\mho$}_{t}(s)\right)\\&-\left(\mathbf{I}_{n}\otimes\mathbf{T}_{2}\right)\mathbf{\Xi}_{t}-\left(\mathbf{I}_{n}\otimes\mathbf{T}_{3}\right)\mbox{\boldmath$\omega$}_{t}
\end{split}
\label{xs}
\end{equation}
using which and defining $\mathbf{I}^{\kappa}\in R^{3\times 3n}$ as the matrix comprised of $(3\kappa-2)^{th}$ till $(3\kappa)^{th}$ rows of the matrix $\mathbf{I}_{3n}$, then for $\kappa=1,\dots,n$ we get $\Tilde{\mathbf{X}}_{\kappa-1}^{\kappa}(s)=\mathbf{I}^{\kappa}\Tilde{\mathbf{X}}_{t}(s)$, plugging which into \eqref{ytt5}, results in
\begin{equation}
u_{i}(s) = -\mathbf{K}\sum_{j\in\mathbb{I}_{i}^{>i}}\sum_{\kappa=i+1}^{j}\mathbf{I}^{\kappa}\Tilde{\mathbf{X}}_{t}(s)+\mathbf{K}\sum_{j\in\mathbb{I}_{i}^{<i}}\sum_{\kappa=j+1}^{i}\mathbf{I}^{\kappa}\Tilde{\mathbf{X}}_{t}(s)
\label{uis}
\end{equation}
Note that $\mathbf{I}_{3n}$ is the identity matrix of size $3n$. 
Given \eqref{fl} and \eqref{uis}, the acceleration/jerk of FVs would be $a_{i}(s) = \frac{1}{1+\tau_{i}s}u_{i}(s)$ and $\mathbb{j}_{i}(s) = \frac{s}{1+\tau_{i}s}u_{i}(s)$, respectively. Also, we can get vehicles' velocities by $v_{i}(s)=(a_{i}(s)+v_{i}(0))/s$. Given that $v_{i}(t)$ and $a_{i}(t)$ are the impulse responses of $v_{i}(s)$ and $a_{i}(s)$, respectively, vehicles' engine inputs can be derived using \eqref{flc}.
\section{\textbf{Platoon Stability}}
\label{platoonStability}
Let's define the platoon's state vector as $\Tilde{\mathbf{X}}_{t} \in \mathbb{R}^{3n \times 1} \triangleq [\Tilde{\mathbf{X}}_{0}^{1};\; \Tilde{\mathbf{X}}_{1}^{2}; \dots; \; \Tilde{\mathbf{X}}_{n-1}^{n}]$, where each element $\Tilde{\mathbf{X}}_{i-1}^{i} \in \mathbb{R}^{3 \times 1}$; $i = 1, \dots, n$, represents the relative measurements (i.e., the distance error, the relative velocity, and the relative acceleration) between neighboring vehicles $i-1$ and $i$ [see Section \ref{pstate}]. Considering the coupled dynamics of all neighboring vehicles [i.e., equations \eqref{col_dyn}-\eqref{inputi-i+1} for $i = 1, \dots, n$], we can formulate the closed-loop dynamics of the platoon based on these relative measurements. This leads to the following state-space representation of the closed-loop dynamic model for the platoon:
\begin{equation}
\dot{\Tilde{\mathbf{X}}}_{t}=\Tilde{\mathbf{A}}_{t}\Tilde{\mathbf{X}}_{t}+\mathbf{B}_{t}\mbox{\boldmath$\epsilon$}_{t}=
 \begin{bmatrix}
\mathbf{A}_{11}    &  \mathbf{A}_{12} & \dots & \mathbf{A}_{1n}
\\\\
\mathbf{A}_{21}  & \mathbf{A}_{22} & \dots & \mathbf{A}_{2n}
\\
\vdots & \dots &  \ddots & \vdots\\
\mathbf{A}_{n1} & \mathbf{A}_{n2} & \dots & \mathbf{A}_{nn}
\end{bmatrix} \Tilde{\mathbf{X}}_{t}+\mathbf{I}_{3n}\mbox{\boldmath$\epsilon$}_{t}  
\label{oiu2}
\end{equation}
in which $\mathbf{B}_{t}=\mathbf{I}_{3n}$ is the identity matrix of size $3n$ and $\mbox{\boldmath$\epsilon$}_{t}=[\mbox{\boldmath$\epsilon$}_{0}^{1};\;\dots;\;\mbox{\boldmath$\epsilon$}_{n-1}^{n}]$; $i=1,\dots,n$, in which  $\mbox{\boldmath$\epsilon$}_{i-1}^{i}=[0;\;0;\;\epsilon_{i-1}^{i}]$ and $\epsilon_{i-1}^{i}=\epsilon_{i-1}-\epsilon_{i}$ [see \eqref{ssmx}]. Also, $\Tilde{\mathbf{A}}_{t}\in R^{3n\times 3n}$ is the platoon closed-loop system matrix such that its elements $\mathbf{A}_{i\kappa}$ for $i,\kappa=1,\dots,n$ are according to [see \eqref{col_dyn}, \eqref{kiip}, and \eqref{trss}]: 
\begin{equation}
\mathbf{A}_{i\kappa}=
\begin{cases}
\mathbf{B}_{i-1}^{i}{\mathbf{K}_{i-1,i}^{\kappa-}} & \text{ if } \quad\quad\kappa<i
\\
\mathbf{B}_{i-1}^{i}\mathbf{K}_{i-1,i}^{t} &\text{ if } \quad\quad i=\kappa\\
\mathbf{B}_{i-1}^{i}{\mathbf{K}_{i-1,i}^{\kappa+}}&  \text{ if }\quad\quad \kappa>i
\end{cases}
\label{totsx2}
\end{equation}
where $\mathbf{K}_{i-1,i}^{t}\triangleq\left[\bar{k}_{i-1}^{i}, \bar{b}_{i-1}^{i}, \bar{h}_{i-1}^{i}\right]$. Note that $\mathbf{B}_{i-1}^{i}\mathbf{K}_{i-1,i}^{t}=\mathbf{A}_{i-1}^{i}$ [see \eqref{pbc} and \eqref{col_dyn}]. Given \eqref{oiu2}, the platoon will be internally stable if all eigenvalues of the system matrix $\tilde{\mathbf{A}}_{t}$ are negative.
\begin{remark}
For the case that vehicles have identical engine time constants, i.e., $\tau_{i-1}=\tau_{i}$, platoon's closed dynamics based on relative measurements is governed by [see \eqref{trss}]
\begin{equation}
\begin{split}
\dot{\Tilde{\mathbf{X}}}_{t}=\Tilde{\mathbf{A}}_{t}\Tilde{\mathbf{X}}_{t}=\left[\mathbf{I}_{n}\otimes \mathbf{A}-\bar{\mathbf{P}}\otimes\mathbf{BK}\right]\Tilde{\mathbf{X}}_{t}+\mathbf{I}_{3n}\mbox{\boldmath$\epsilon$}_{t}
\end{split}
\label{homo_dyn}
\end{equation}
in which $\mathbf{A}=\mathbf{A}_{i-1}=\mathbf{A}_{i}$, $\mathbf{B}=\mathbf{B}_{i-1}=\mathbf{B}_{i}$,  $\mbox{\boldmath$\epsilon$}_{t}=[\mbox{\boldmath$\epsilon$}_{0}^{1};\;\mathbf{0}_{3\times 1},\;\dots;\;\mathbf{0}_{3\times 1}]$ such that $\mbox{\boldmath$\epsilon$}_{0}^{1}=[0;\;0;\;\epsilon_{0}^{1}]$ and $\epsilon_{0}^{1}=-\epsilon_{1}$ [see \eqref{ssmx}], and $\mathbf{0}_{3\times 1}$ is the zero vector of size $3$. Also,  $\bar{\mathbf{P}}\in R^{n\times n}$ whose elements $\bar{p}_{i\kappa}$ are according to [see \eqref{pbc} and \eqref{kiip}]
\begin{equation}
\bar{p}_{i\kappa}=
\begin{cases}
|\mathbb{R}_{i-1}^{\leq \kappa-1}|-|\mathbb{R}_{i}^{\leq \kappa-1}| & \text{ if } \quad\quad \kappa<i
\\|\mathbb{I}_{i}^{\leq i-1}|+|\mathbb{I}_{i-1}^{\geq i}| & \text{ if } \quad\quad i=\kappa \\
|\mathbb{R}_{i}^{\geq \kappa}|-|\mathbb{R}_{i-1}^{\geq \kappa}|  & \text{ if } \quad\quad \kappa>i
\end{cases}   
\end{equation}
Note that $\left|.\right|$ denotes the cardinality of the relevant sets.
\end{remark} 
\begin{remark}
Please note that the closed-loop dynamic description in \eqref{homo_dyn} differs from those in the literature as we have defined the relative states between neighboring vehicles as the states of the platoon.     
\end{remark}
\section{\textbf{Metrics for Platoon Performance}}
\label{pmetrcis}
We present four groups of platoon performance metrics: Safe Control Gain Deficiency Index (SaCGDI) [see \eqref{SaCGDI}], Intervehicle Distance Safety (including AAPMTTC [see \eqref{AAPMTTC}] and AAMDRAC [see \eqref{AADRAC}]), Energy Consumption (AAMEEI [see \eqref{AAIE}]), and Passenger Comfort (including AAMEA [see \eqref{AAAE}] and AAMEJ [see \eqref{AAJE}]). These metrics will be used to quantitatively assess and compare the performance of platoon under different RCTs.
\subsection{\textbf{Safe Control Gain Deficiency Index (SaCGDI)}}
\label{SaCGDI}
Considering section \ref{safesa}, different communication topologies can lead to variations in the proportion of control gains that meet safety criteria. To study this, we define SaCGDI metric to evaluate the performance of vehicular platoons across various communication topologies. The metric is defined as follows: 
\begin{equation}
\text{SaCGDI}=\left(1-\frac{\#_{\text{SaCGVs}}}{\#_{\text{TCGVs}}}\right)\times 100  
\label{SaCGDI}
\end{equation}
which represents the percentage of Control Gain Vectors (CGVs) that do not maintain safe distances between vehicles. A safe distance between neghboring vehicles, i.e., $d_{i-1,i}^{s}$, is a preset value we consider for intervehicle distances, e.g., $d_{i-1,i}^{s}=3$ m. In \eqref{SaCGDI}, $\#_{\text{SaCGVs}}$ denotes the number of Safe CGVs and  $\#_{\text{TCGVs}}$ denotes the Total number of CGVs.
\begin{remark}
A lower SaCGDI value suggests that a higher proportion of control gains are safe, which indicates better performance under the given communication topology.
\end{remark}
\subsection{\textbf{Intervehicle Distance Safety}}
\label{IDS}
We introduce two metrics: Accumulative Average Penalty of Minimum Time to Collision (AAPMTTC) and  Accumulative Average  Minimum Deceleration Rate to Avoid Collision (AAMDRAC).

\textbf{AAPMTTC}: First, let's calculate Minimum Time to Collision between neighboring vehicles $i-1$ and $i$ at time $t=j\Delta t$; let it be denoted as $MTTC_{i}(t)$, where $\Delta t$ is the sampling time, and $j$ is time steps over the travel time. At time $t$ seconds after platooning triggered, based on the distance; $D_{i-1}^{i}(t)$, relative velocity; $v_{i-1}^{i}(t)$, and relative acceleration; $a_{i-1}^{i}(t)$, $MTTC_{i}(t)$ computes the minimum duration (in seconds) it would take for the following vehicle $i$ to traverse the intervehicle distance; $D_{i-1}^{i}(t)$, to its preceding vehicle, i.e., vehicle $i-1$. The MTTC value is calculated according to ($i=1,\dots,n-1$)
\begin{equation}
MTTC_{i}(t)=
\begin{cases}
\frac{-v_{i-1}^{i}(t)-\sqrt{\Delta_{i-1}^{i}(t)}}{a_{i-1}^{i}(t)}>0 & \text{ if }  \Delta_{i-1}^{i}(t)>0 
\\
+ \infty & \text{ } Otherwise
\end{cases} 
\label{mttc}
\end{equation}
where $\Delta_{i-1}^{i}(t)=(v_{i-1}^{i}(t))^{2}-2a_{i-1}^{i}(t)D_{i-1}^{i}(t)$. Note that the second case actually implies that mathematically there is no real positive value for $MTTC_{i}(t)$. In other words, given the values of $D_{i-1}^{i}(t)$, $v_{i-1}^{i}(t)$, and $a_{i-1}^{i}(t)$, theoretically the following vehicle FV cannot reach to its preceding vehicle, then we consider $MTTC_{i}(t)=+\infty$. Based on \eqref{mttc}, the Penalty of MTTC (PMTTC) for the $i^{th}$ FV, is defined as:
\begin{equation}
PMTTC_{i}(t) = 100 e^{\frac{-0.1MTTC_{i}(t)}{T_{s}}}; \hspace{0.5cm} i=1,\dots,n  
\label{pfunction}
\end{equation}
where $T_{s}\triangleq 1$ second. The penalty function implies that, for the pair $(i-1,i)$, the more the vehicle  $i$ becomes closer to vehicle $i-1$, the more it is penalized which, at maximum, would be approximately equal to $100$. Note that $MTTC_{i}(t)=0$ implies collision between the neighboring vehicles which is not the case here. On the other hand, if theoretically they cannot reach to each other (for instance, when simultaneously we have $v_{i-1}^{i}(t)>0$ and $a_{i-1}^{i}(t)>0$), then the FV will not be penalized, i.e.,  $MTTC_{i}(t)=+\infty$ and thus $PMTTC_{i}(t)=0$. The Average PMTTC (APMTTC) of platoon, under a given RCT and a set of Stable and Non-Colliding CGVs; let it be denotd as $\mathbf{K}_{s}$, is:
\begin{equation}
APMTTC(t) =\frac{1}{\#_{\mathbf{K}_{s}}}\sum_{\mathbf{K}_{s}}\sum_{i=1}^{n}PMTTC_{i}(t) 
\label{APMTTC}
\end{equation}
where  $\#_{\mathbf{K}_{s}}$ denotes the number of CGVs in the set $\mathbf{K}_{s}$. To quantify the safety for platoon over the travel time from $t=0$ till $t$ seconds, Accumulative APMTTC (AAPMTTC) is:
\begin{equation}
AAPMTTC(t) =\sum_{i=0}^{j} APMTTC(i\Delta t)
\label{AAPMTTC}
\end{equation}
where $t=j\Delta t$ and $j$ is time steps over the travel time.
\begin{remark}
Given \eqref{pfunction}, the penalty function values for $MTTC_{i}(t) \geq 60$ seconds are nearly zero, indicating that a 60-second time duration is considered safe for the follower vehicle, especially when compared to the average human reaction time of 1.5 seconds \cite{green2000long}. However, to select a penalty function that is more responsive to $MTTC_{i}(t) \geq 60$, we can use the alternative penalty function shown with a dashed line in Fig. \ref{pf_alt}. It is important to note that the penalty function represented by the solid line is more sensitive to smaller values of $MTTC_{i}(t)$, imposing greater penalties on FVs as they come closer. Therefore, we choose the solid line function.
\begin{figure}[htbp!]
\begin{center}
\resizebox{0.8\hsize}{!}{\includegraphics*{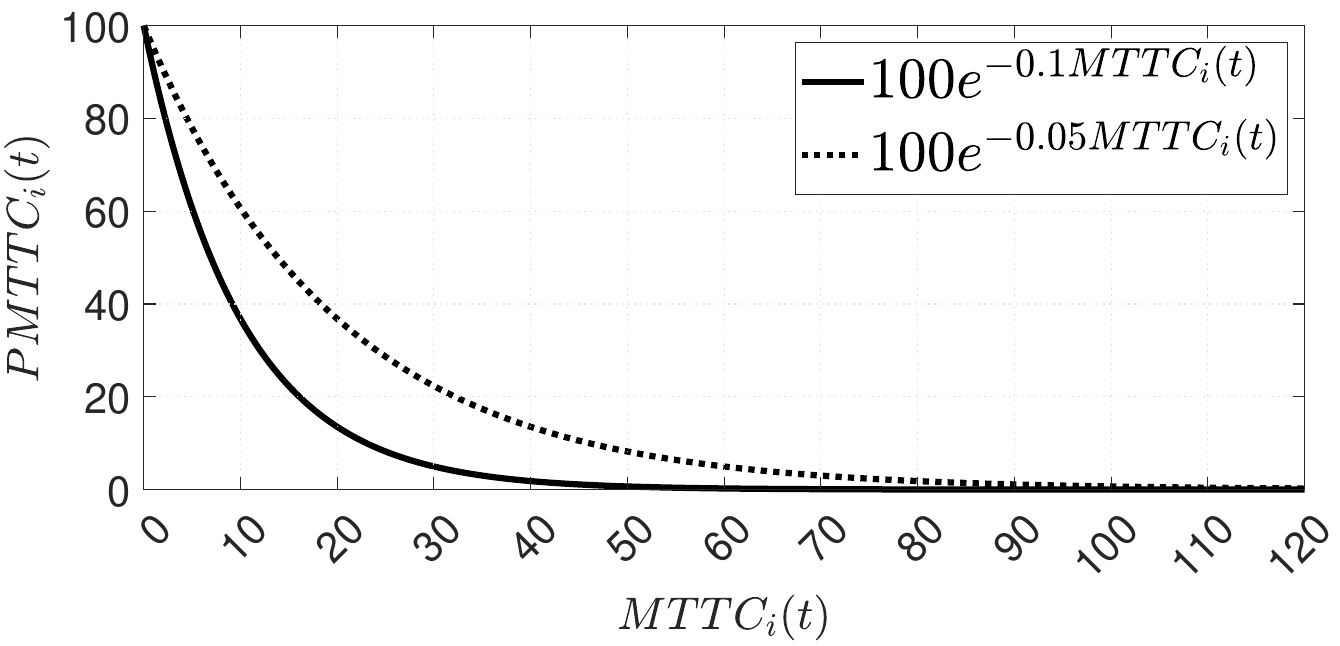}}
\caption{Penalty functions for MTTC.}
\label{pf_alt}
\end{center}
\end{figure}    
\end{remark}

\textbf{AAMDRAC}: Let's first elaborate on the calculation of Minimum Deceleration Rate to Avoid Collision (MDRAC) between neighboring vehicles $i-1$ and $i$, denoted as $MDRAC_{i}(t)$. Suppose that at time $t=j\Delta t$, we have the values of $D_{i-1}^{i}(t)$, $v_{i-1}^{i}(t)$, and $a_{i-1}^{i}(t)$ for the pair $(i-1,i)$ as the distance, relative velocity, and relative acceleration between the neghboing vehicles $i-1$ and $i$. With time paused, $MDRAC_{i}(t)$ computes the minimum constant deceleration rate the FV should apply to avoid collision with its preceding vehicle. The minimum value can be obtained considering that FV will needs zero velocity after traversing the intervehicle distance $D_{i-1}^{i}(t)$. If $v_{i-1}^{i}(t)\geq 0$ and $a_{i-1}^{i}(t)\geq 0$, the FV theoretically can not reach the preceding vehicle, and thus does not need any deceleration. However, if $v_{i-1}^{i}(t)\geq 0$ and $a_{i-1}^{i}(t)<0$, the FV can theoretically reach the preceding vehicle. Therefore, we set $MDRAC_{i}(t)=-a_{i-1}^{i}(t)/T_{a}$ which is the minimum deceleration rate to make the FV keep the same velocity difference from its preceding vehicle. Note that $T_{a}\triangleq 1$ $\mathrm{m/s^2}$. Also, for the case that $v_{i-1}^{i}(t)<0$, instant deceleration is needed. Thus, $MDRAC_{i}(t)$ for the FV $i$ would be ($i=1,\dots,n$):
\begin{equation}
MDRAC_{i}(t)=
\begin{cases}
\frac{(v_{i-1}^{i}(t))^{2}}{2D_{i-1}^{i}(t)T_{a}}; & \hspace{0.05cm}v_{i-1}^{i}(t)<0
\\
-a_{i-1}^{i}(t)/T_{a}; &  \hspace{0.05cm} v_{i-1}^{i}(t)>0 \text{ \& } a_{i-1}^{i}(t)<0\\
0; & \text{Otherwise}
\end{cases}   
\end{equation}
using which, and similiar to relations \eqref{APMTTC}-\eqref{AAPMTTC}, the  Average MDRAC (AMDRAC) of platoon is defined as:
\begin{equation}
AMDRAC(t) =\frac{1}{\#_{\mathbf{K}_{s}}}\sum_{\mathbf{K}_{s}}\sum_{i=1}^{n}MDRAC_{i}( t) 
\label{ADRAC}
\end{equation}
and thus Accumulative AMDRAC (AAMDRAC) for the platoon over the travel time can be derived as
\begin{equation}
AAMDRAC(t) =\sum_{i=0}^{j} AMDRAC(i\Delta t)
\label{AADRAC}
\end{equation}
where $t=j\Delta t$ and $j$ is time steps over the travel time.
\begin{remark}
The APMTTC and AMDRAC values offer momentary safety insights at any instant of travel time, while the  AAPMTTC and AAMDRAC metrics provide accumulative safety insights for the entire travel time. Smaller values  correspond to safer conditions. 
\end{remark}
\subsection{\textbf{Energy Consumption}}
\label{EC}
Rapid fluctuations in engine input typically lead to higher energy consumption \cite{al2021driving,ma2021eco,yao2020vehicle,sun2020optimal,yu2022eco}. Given \eqref{flc}, the engine input is directly related to the force applied, likely through a throttle or similar mechanism that regulates the power output of the engine. Thus, rapid  changes in this force input signal would  lead to increased energy consumption. To quantify these fluctuations, reminding  subsection \ref{CEIAJ} and equation \eqref{flc}, and similiar to relations \eqref{APMTTC}-\eqref{AAPMTTC}, we first calculate Average Momentary Energy of the Engine Input (AMEEI) according to:
\begin{equation}
AMEEI(t) =\frac{1}{\#_{\mathbf{K}_{s}}}\sum_{\mathbf{K}_{s}}\sum_{i=1}^{n}\left(c_{i}(t)/T_{f}\right)^{2}
\label{enCT}
\end{equation}
in which $T_{f}\triangleq 1$ $\mathrm{kg .m/s^2}$. Using \eqref{enCT}, Accumulative AMEEI (AAMEEI) for the platoon over the travel time would be:
\begin{equation}
AAMEEI(t) =\sum_{i=0}^{j} AMEEI(i\Delta t)
\label{AAIE}
\end{equation}
where $t=j\Delta t$ and $j$ is total time steps over the travel time.
\begin{remark}
The AMEEI value offers momentary energy efficiency  insight at any instance of travel time, while the AAMEEI  metric provides accumulative energy consumption insight for the entire travel time. Similar to previous metrics, smaller values for the momentary and accumulative values correspond to better platoon performance.   \end{remark}
\subsection{\textbf{Passenger Comfort}}
\label{PC}
Acceleration refers to the rate of change of velocity, and jerk is the rate of change of acceleration. Both metrics are commonly used to evaluate ride comfort in vehicles \cite{liu2021dynamic,mohajer2020enhancing,wang2020research,yan2022cooperative,hang2020integrated}. Smooth acceleration and minimal jerk contribute to a more comfortable ride, reducing discomfort caused by sudden or sharp movements. We use acceleration and jerk to measure movement smoothness and, consequently, driving comfort for passengers.  Thus, given subsection \ref{CEIAJ},  Average Momentary Energy of Acceleration and Jerk signals, i.e.,  AMEA and AMEJ, are calculated as follows:
\begin{subequations}
\begin{equation}   
AMEA_{i}(t)=\frac{1}{\#_{\mathbf{K}_{s}}}\sum_{\mathbf{K}_{s}}\sum_{i=1}^{n}\left(a_{i}(t)/T_{a}\right)^{2}
\label{enJA1}
\end{equation}
\begin{equation}   
AMEJ_{i}(t)=\frac{1}{\#_{\mathbf{K}_{s}}}\sum_{\mathbf{K}_{s}}\sum_{i=1}^{n}\left(\mathbb{j}_{i}(t)/T_{j}\right)^{2}
\label{enJA2}
\end{equation}
\label{enJA}
\end{subequations}
where $T_{a}\triangleq 1$ $\mathrm{m/s^{2}}$ and $T_{j}\triangleq 1$ $\mathrm{m/s^{3}}$. Using \eqref{enJA} Accumulative AMEA (AAMEA) and Accumulative AMEJ (AAMEJ) for the platoon are determined by:
\begin{subequations}
\begin{equation}
AAMEA(t) = \sum_{i=0}^{j} AMEA(i\Delta t)
\label{AAAE}
\end{equation}
\begin{equation}
AAMEJ(t) = \sum_{i=0}^{j} AMEJ(i\Delta t)
\label{AAJE}
\end{equation}
\end{subequations}
where $t=j\Delta t$ and $j$ is time steps over the travel time.
\begin{remark}
the AMEA and AMEJ values offer momentary comfort insights at any instant of travel time, while the  AAMEA and AAMEJ metrics provide accumulative comfort insights for the travel time. Smaller values for the momentary and accumulative values correspond to safer conditions.   
\end{remark}
\section{{\color{black}\textbf{Simulations and Results}}}
\label{simpre}
\subsection{\textbf{Preliminaries 1 (Simulations Setup)}}
\label{simprer}
We consider a platoon consisting of one leader and four FVs with the following considerations:

1. \textbf{Initial Conditions}: When the leader and FVs are moving, platooning is initiated, and thus each vehicle has its own initial position, velocity, and acceleration. This results in initial distance error $p_{i-1}^{i}(0)$, initial relative velocity $v_{i-1}^{i}(0)$, and initial relative acceleration $a_{i-1}^{i}(0)$ between neighboring vehicles. Initial conditions are summarized in Table \ref{table2}.
\begin{table}[h]
\centering
\renewcommand{\arraystretch}{1.5}
\caption{Initial conditions for individual and coupled states.}
\begin{adjustbox}{width=0.48\textwidth}
\centering
\begin{tabular}{|cccc|cccc|}
\hline
\multicolumn{4}{|c|}{Vehicles' Initial Conditions} &
  \multicolumn{4}{c|}{Coupled  Initial Conditions} \\ \Xhline{2pt}
\multicolumn{1}{|c|}{Vehicle} &
  \multicolumn{1}{c|}{\begin{tabular}[c]{@{}c@{}}Position (m)\end{tabular}} &
  \multicolumn{1}{c|}{\begin{tabular}[c]{@{}c@{}}Velocity ($\mathrm{m/s}$)\end{tabular}} &
  \begin{tabular}[c]{@{}c@{}}Acceleration ($\mathrm{m/s^{2}}$) \end{tabular} &
  \multicolumn{1}{c|}{\multirow{2}{*}{Pair}} &
  \multicolumn{1}{c|}{\multirow{2}{*}{\begin{tabular}[c]{@{}c@{}}Distance\\Error\end{tabular}}} &
  \multicolumn{1}{c|}{\multirow{2}{*}{\begin{tabular}[c]{@{}c@{}}Relative\\Velocity\end{tabular}}} &
  \multirow{2}{*}{\begin{tabular}[c]{@{}c@{}}Relative\\ Acceleration\end{tabular}} \\ \customCline{1.1pt}{1}{4}
\multicolumn{1}{|c|}{0} &
  \multicolumn{1}{c|}{2.832} &
  \multicolumn{1}{c|}{4.760} &
  4.000 &
  \multicolumn{1}{c|}{} &
  \multicolumn{1}{c|}{} &
  \multicolumn{1}{c|}{} &
   \\ \hline\customCline{1.1pt}{5}{8}
\multicolumn{1}{|c|}{1} &
  \multicolumn{1}{c|}{-11.424} &
  \multicolumn{1}{c|}{7.313} &
  5.841 &
  \multicolumn{1}{c|}{(0,1)} &
  \multicolumn{1}{c|}{5.256} &
  \multicolumn{1}{c|}{-2.553} &
  -1.841 \\ \hline
\multicolumn{1}{|c|}{2} &
  \multicolumn{1}{c|}{-28.065} &
  \multicolumn{1}{c|}{7.806} &
  6.405 &
  \multicolumn{1}{c|}{(1,2)} &
  \multicolumn{1}{c|}{7.640} &
  \multicolumn{1}{c|}{-0.493} &
  -0.563 \\ \hline
\multicolumn{1}{|c|}{3} &
  \multicolumn{1}{c|}{-41.661} &
  \multicolumn{1}{c|}{10.738} &
  8.533 &
  \multicolumn{1}{c|}{(2,3)} &
  \multicolumn{1}{c|}{4.597} &
  \multicolumn{1}{c|}{-2.932} &
  -2.128 \\ \hline
\multicolumn{1}{|c|}{4} &
  \multicolumn{1}{c|}{-57.081} &
  \multicolumn{1}{c|}{10.384} &
  9.599 &
  \multicolumn{1}{c|}{(3,4)} &
  \multicolumn{1}{c|}{6.420} &
  \multicolumn{1}{c|}{0.354} &
  -1.066 \\ \hline
\end{tabular}
\end{adjustbox}
\label{table2}
\end{table}

2. \textbf{Leader Vehicle's Acceleration}: We consider three acceleration trajectories for the leader after the platooning is triggered. These acceleration trajectories, with their corresponding velocity trajectories, are illustrated in Fig. \ref{leader_av}. The relevant functions are provided in \eqref{accvel}. Note that $v_{0}(0)=4.760$ $\mathrm{m/s}$.
\begin{equation}
\begin{aligned}[b]
a_{0}(s)=
\begin{cases}
   \frac{4s+14}{s^{2}+1.5s+1} & \quad\text{ Acc. 1}\\
   \frac{(4s+1)(s+1)}{(s+2)(s^{2}+2s+12)} &\quad \text{ Acc. 2}\\
   \frac{4s+1}{s^{2}+3s+2} &\quad \text{ Acc. 3}
\end{cases}
\end{aligned}
\label{accvel}
\end{equation}

Furthermore, note that the whole travel time is assumed to be $25$ s, and sampling time is considered as $\Delta t=0.01$ s.

3. \textbf{Engine Time Constants (ETCs)}: Three sets of different ETCs are considered, as shown in Table \ref{table1}.
\begin{table}[h]
\renewcommand{\arraystretch}{1}
\centering
\caption{Different sets of engine time constants (ETCS).}
\begin{adjustbox}{width=0.48\textwidth}
\begin{tabular}{c|cccc|}
\cline{2-5}
\multicolumn{1}{c|}{} &
  \multicolumn{1}{c|}{Follower 1} &
  \multicolumn{1}{c|}{Follower 2} &
  \multicolumn{1}{c|}{Follower 3} &
  Follower 4 \\ \customCline{1.1pt}{1}{5}
\rowcolor[HTML]{EAD1DC} 
\multicolumn{1}{|c|}{\cellcolor[HTML]{EAD1DC}Case 1} &
  \multicolumn{1}{c|}{\cellcolor[HTML]{EAD1DC}1} &
  \multicolumn{1}{c|}{\cellcolor[HTML]{EAD1DC}1} &
  \multicolumn{1}{c|}{\cellcolor[HTML]{EAD1DC}1} &
  1 \\ \hline
\rowcolor[HTML]{D9D2E9} 
\multicolumn{1}{|c|}{\cellcolor[HTML]{D9D2E9}Case 2} &
  \multicolumn{1}{c|}{\cellcolor[HTML]{D9D2E9}0.7} &
  \multicolumn{1}{c|}{\cellcolor[HTML]{D9D2E9}0.6} &
  \multicolumn{1}{c|}{\cellcolor[HTML]{D9D2E9}1} &
  0.9 \\ \hline
\rowcolor[HTML]{CFE2F3} 
\multicolumn{1}{|c|}{\cellcolor[HTML]{CFE2F3}Case 3} &
  \multicolumn{1}{c|}{\cellcolor[HTML]{CFE2F3}0.7} &
  \multicolumn{1}{c|}{\cellcolor[HTML]{CFE2F3}0.8} &
  \multicolumn{1}{c|}{\cellcolor[HTML]{CFE2F3}0.4} &
  0.5 \\ \hline
\end{tabular}
\end{adjustbox}
\label{table1}
\end{table}

4. \textbf{Physical Parameters}: Vehicles' masses ($m_{i}$), cross sectional areas ($A_{cs_{i}}$), drag coefficients ($C_{d_{i}}$), mechanical drags ($d_{m_{i}}$), and lengths ($L_{i}$) are provided in Table \ref{phpa}. Also, specific mass of air is $1.204$ $\mathrm{kg/m^{3}}$.
\begin{table}
\renewcommand{\arraystretch}{1.1}
\centering
\caption{Physical parameters.}
\begin{adjustbox}{width=0.48\textwidth}
\begin{tabular}{c|c|c|c|c|}
\cline{2-5}
\multicolumn{1}{c|}{}                                         & Follower 1 & Follower 2 & Follower 3 & Follower 4 \\ \Xhline{1.1pt}
\multicolumn{1}{|c|}{$m_{i}$ ($\mathrm{kg}$)}                    & 1900.258   & 1800.036   & 1950.98    & 2000.877   \\ \hline
\multicolumn{1}{|c|}{$A_{cs_{i}}$ ($\mathrm{m^{2}}$)} & 2.444      & 2.713      & 2.543      & 3.791      \\ \hline
\multicolumn{1}{|c|}{$C_{d_{i}}$}                        & 0.412      & 0.311      & 0.359      & 0.511      \\ \hline
\multicolumn{1}{|c|}{$d_{m_{i}}$ ($\mathrm{kg.m/s^{2}}$)} & 4.111      & 3.831      & 3.902      & 4.001      \\ \hline
\multicolumn{1}{|c|}{$L_{i}$ ($\mathrm{m}$)}                   & 4.000      & 4.000      & 4.000      & 4.000      \\ \hline
\end{tabular}
\end{adjustbox}
\label{phpa}
\end{table} 

5. \textbf{Desired and Safe Distances}: The  desired and safe distances between neighboring vehicles are set to $5$ m, and $3$ m, respectively, i.e., $d_{i-1}^{i}=5$ m and $d_{i-1,i}^{s}=3$ m; $i=1,\dots,4$.

6. \textbf{Control Gain Vectors (CGVs)}: The total CGVs involve the $k$ and $b$ parameters, which vary from $0.1$ to $20$ in increments of $0.5$. The parameter $h$ is set to $4$. Therefore, a total of $1600$ different sets of CGVs; 
$\mathbf{K}=[k,\;b,\;h]$, are utilized [see Table \ref{CGS}]. 
\begin{table}[]
\renewcommand{\arraystretch}{1.1}
\centering
\caption{Selected, in total, 1600 control gain vectors.}
\begin{adjustbox}{width=0.48\textwidth}
\begin{tabular}{c|ccc|}
 \cline{2-4} 
 & \multicolumn{3}{c|}{$\mathbf{K}= [k,\;b,\;h]$} \\ \customCline{1.1pt}{1}{4}
\multicolumn{1}{|c|}{Control Gain}                  & \multicolumn{1}{c|}{$k$}        & \multicolumn{1}{c|}{$b$}        & $h$ \\ \hline
\multicolumn{1}{|c|}{Selected Values}               & \multicolumn{1}{c|}{0.1:0.5:20} & \multicolumn{1}{c|}{0.1:0.5:20} & 4   \\ \hline
\multicolumn{1}{|c|}{Total Number of Control Gains} & \multicolumn{1}{c|}{40}         & \multicolumn{1}{c|}{40}         & 1   \\ \hline
\end{tabular}
\end{adjustbox}
\label{CGS}
\end{table}

7. \textbf{Types of CGVs}: Based on subsection \ref{safesa}, for each CGV, we can determine whether it results in an Unstable, Stable-Colliding, Stable-Unsafe, or Stable-Safe platoon. Thus, we categorize the corresponding CGVs as Unstable, Stable-Colliding, Stable-Unsafe, and Stable-Safe CGVs, respectively. The categories and their conditions are summarized in Table \ref{table-colors}. Note that we will refer to Stable-Safe and Stable-Unsafe CGVs as Non-Colliding CGVs. 
\begin{table}
\centering
\renewcommand{\arraystretch}{1.3}
\caption{Types of control gain vectors (CGVs) and relevant conditions. Note that $p_{i-1}^{i}(t)$ follows \eqref{impulese}, and colors are used later in the plot of CGVs [see Fig. \ref{homo_areas}-\ref{hete2_areas}].}
\begin{adjustbox}{width=0.48\textwidth}
\centering
\begin{tabular}{|c|c|c|}
\hline
$\mathbf{K}=[k,\;b,\;h]$               & \multicolumn{1}{c|}{Condition} & Color \\ \customCline{1.1pt}{1}{3}
Unstable        &  \text{at least one positive eigenvalue for system matrix}        & Yellow         \\ \hline
Stable-Colliding &   $\exists t$, $p_{i-1}^{i}(t)\leq-d_{i-1}^{i}$                       & Red            \\ \hline
Stable-Unsafe   &   $\exists t$, $-d_{i-1}^{i}<p_{i-1}^{i}(t)<d_{i-1,i}^{s}-d_{i-1}^{i}$    & Blue           \\ \hline
Stable-Safe     &  $\forall t$,  $p_{i-1}^{i}(t)>d_{i-1,i}^{s}-d_{i-1}^{i}$                      & Green          \\ \hline
\end{tabular}
\end{adjustbox}
\label{table-colors}
\end{table}

8. \textbf{Rigid Communication Topologies (RCTs)}: All formulations in this paper apply to any RCT, but for comparisons, we consider 10 common RCTs shown in Figs. \ref{typicalu}-\ref{typicalb}.

9. \textbf{Shared Control Gain Vectors (SCGVs)}:
\label{SCGVs}
Refer to the CGVs that fall into the same category [see Table \ref{table-colors}] across considered RCTs. For example, the green points in Fig. \ref{shared1} represent  the intersection of all Stable-Safe CGVs shown in Fig. \ref{homo_Acc1} (green points). In other words, they are CGVs that result in Stable-Safe platoons, under all considered RCTs.
\begin{remark}
In some cases [see Figs. \ref{homo_Acc2}, \ref{homo_Acc3}, \ref{hete1_Acc2}, and \ref{hete1_Acc3}], no SCGV exists between SPTF and other communications due to the absence of green points under SPTF. In such cases, SCGVs are identified without considering SPTF topology [see Figs. \ref{shared2}, \ref{shared3}, \ref{shared5}, and \ref{shared6}], and metrics for SPTF are marked as N.A. (Not Applicable, see Tables \ref{safety_table}, \ref{energy_table}, \ref{comfort_table}, and \ref{overall_table}) and corresponding PI and Average PI values (performance index, which will be discussed later), in Figs. \ref{barplots_aheadeffect}-\ref{barplots_leader_vs_ahead}, ared denoted as $\infty$ (i.e., a very large number). This indicates that SPTF topology is the worst among the considered RCTs, as no SCGV results in a Stable-Safe platoon under SPTF topology. 
\label{shcgvs}
\end{remark}
\begin{figure}[htbp!]
\begin{center}
\resizebox{0.8        \hsize}{!}{\includegraphics*{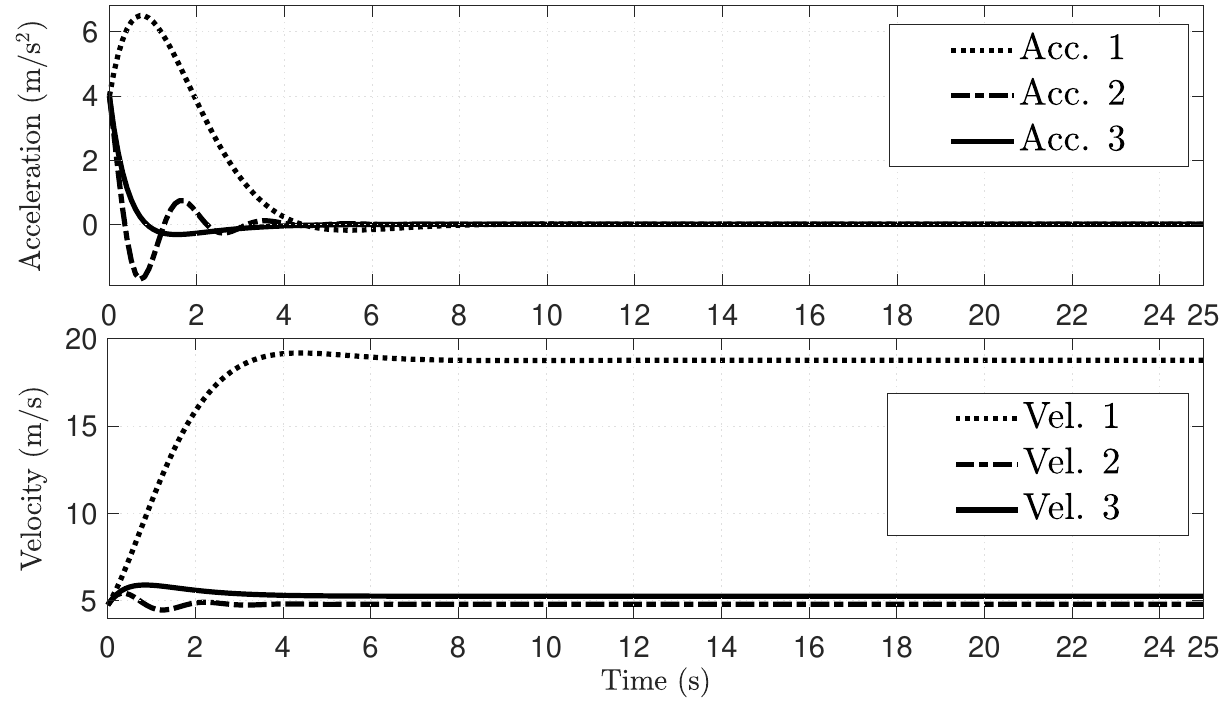}}
\caption{Leader's accelerations/velocities (see \eqref{accvel}). Note that $v_{0}(s) = (a_{0}(s) + v_{0}(0))/s$.}
\label{leader_av}
\end{center}
\end{figure} 
\begin{figure}[htbp!] 
    \centering
  \subfloat[Under Case 1 (see Table \ref{table1}) and Acc. 1 (see Fig. \ref{leader_av}).\label{homo_Acc1}]{%
       \includegraphics[width=0.8\linewidth]{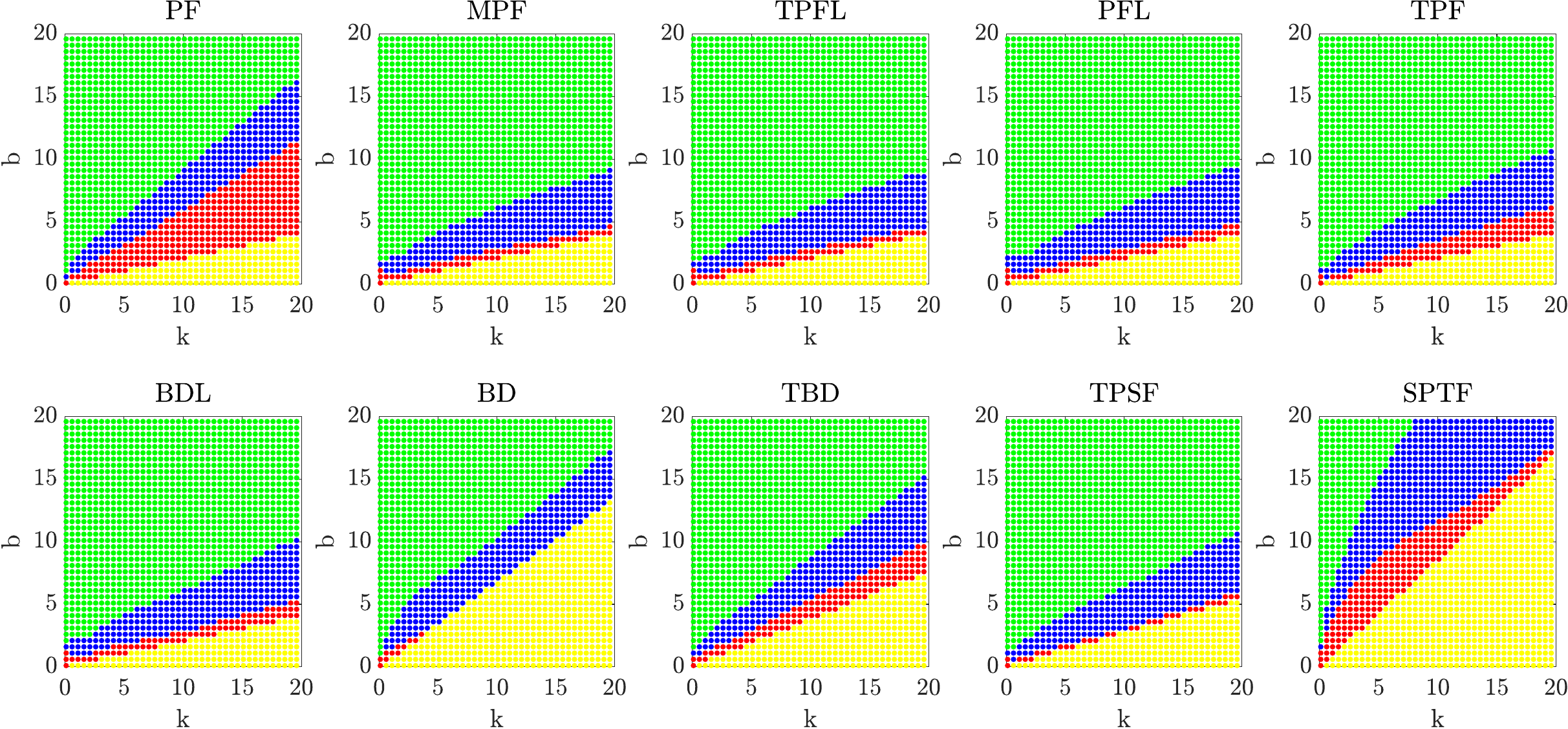}}
\vspace{0.4em}
\subfloat[Under Case 1 (see Table \ref{table1}) and Acc. 2 (see Fig. \ref{leader_av}). \label{homo_Acc2}]{%
        \includegraphics[width=0.8\linewidth]{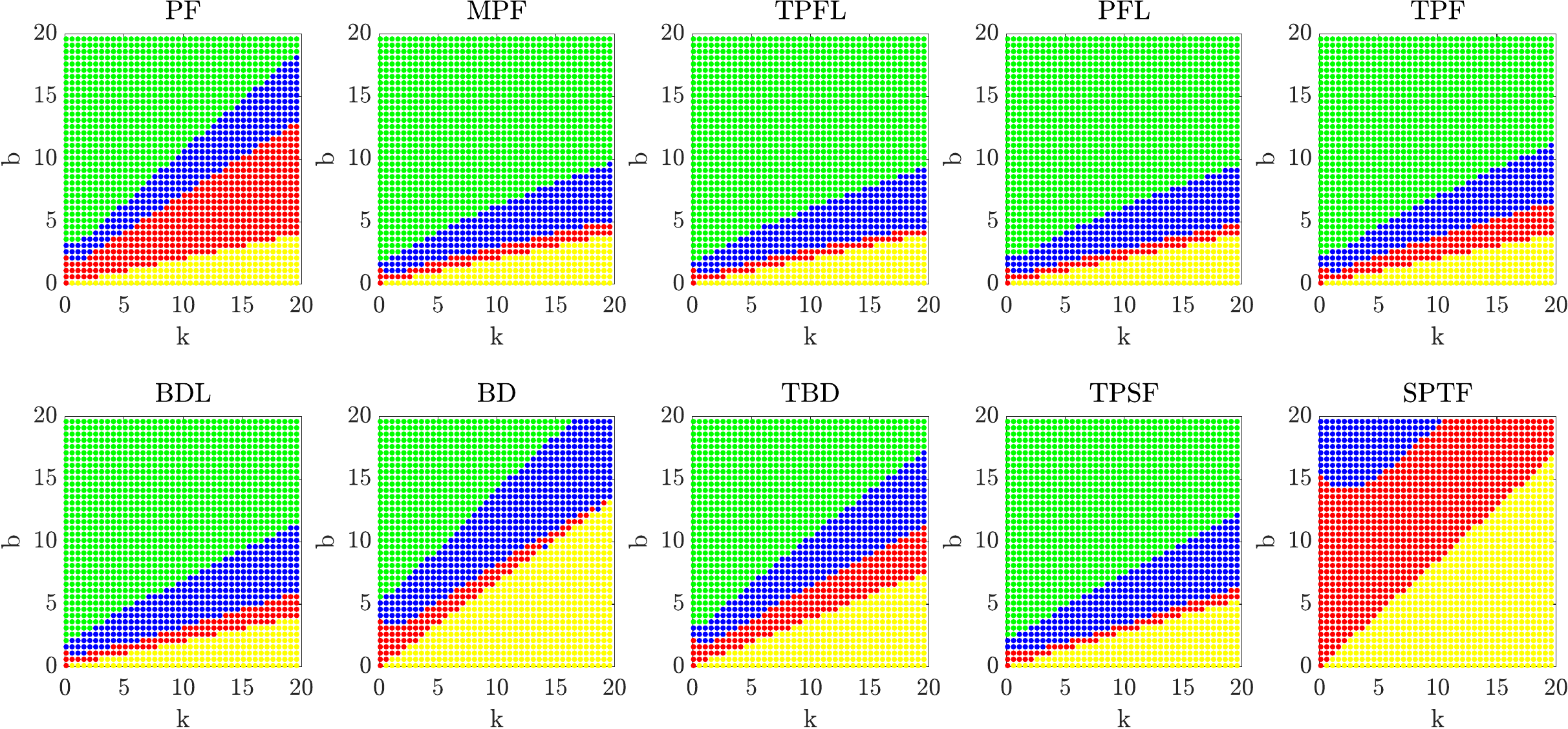}}
\vspace{0.4em}
  \subfloat[Under Case 1 (see Table \ref{table1}) and Acc. 3 (see Fig. \ref{leader_av}). \label{homo_Acc3}]{%
        \includegraphics[width=0.8\linewidth]{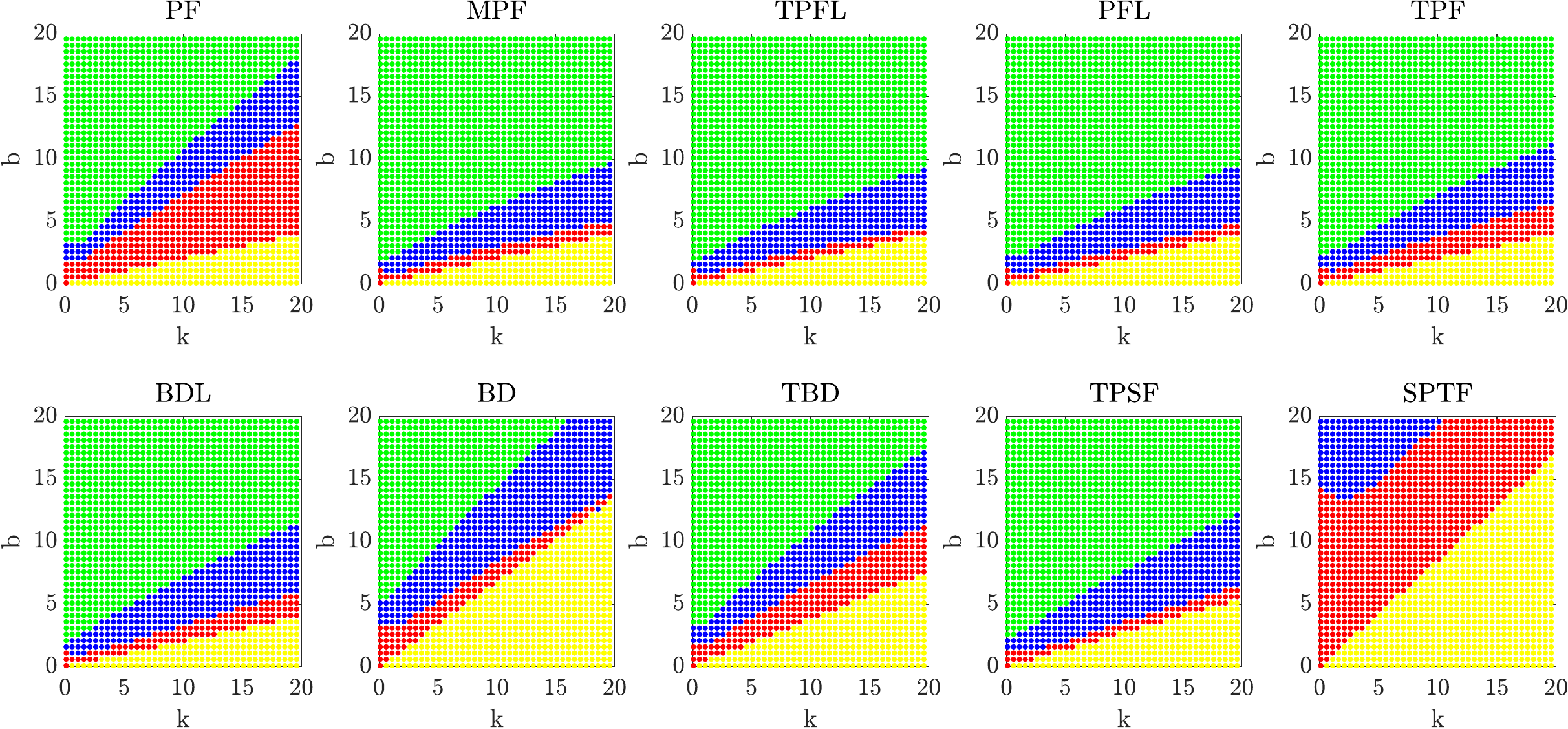}}
\vspace{0.4em}
\caption{Types of CGVs (see Table \ref{table-colors}) under Case 1 (see Table \ref{table1}) and common RCTs (see Figs. \ref{typicalu}-\ref{typicalb}). Given $\mathbf{K}=[k,\;b,\;h]$ (see Table \ref{CGS}), the X-axis and Y-axis represent control gains $k$ and $b$, respectively.}
  \label{homo_areas} 
\end{figure}
\begin{figure}[htbp!] 
    \centering
  \subfloat[Under Case 2 (see Table \ref{table1}) and Acc. 1 (see Fig. \ref{leader_av})  .\label{hete1_Acc1}]{%
       \includegraphics[width=0.8\linewidth]{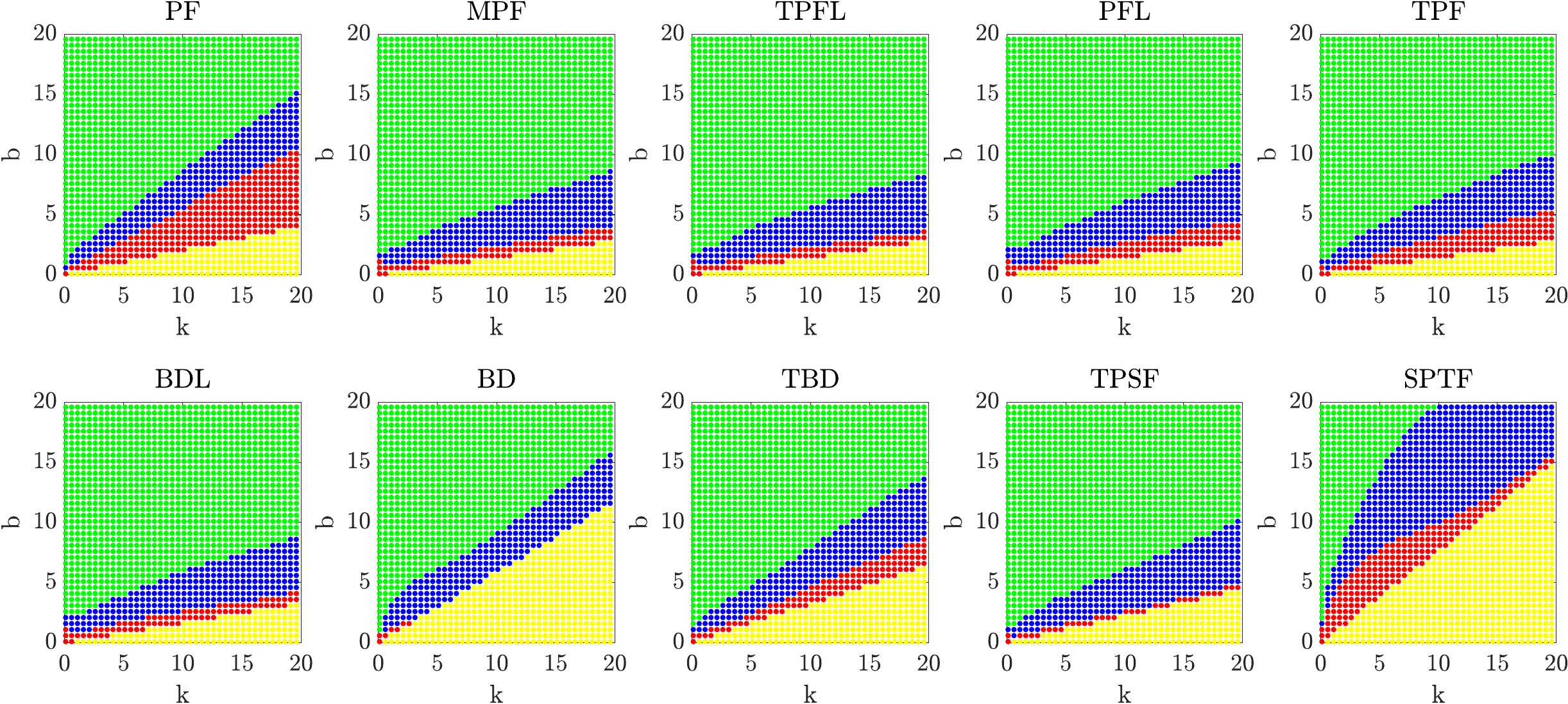}}
\vspace{0.4em}
\subfloat[Under Case 2 (see Table \ref{table1}) and  Acc. 2 (see Fig. \ref{leader_av}). \label{hete1_Acc2}]{%
        \includegraphics[width=0.8\linewidth]{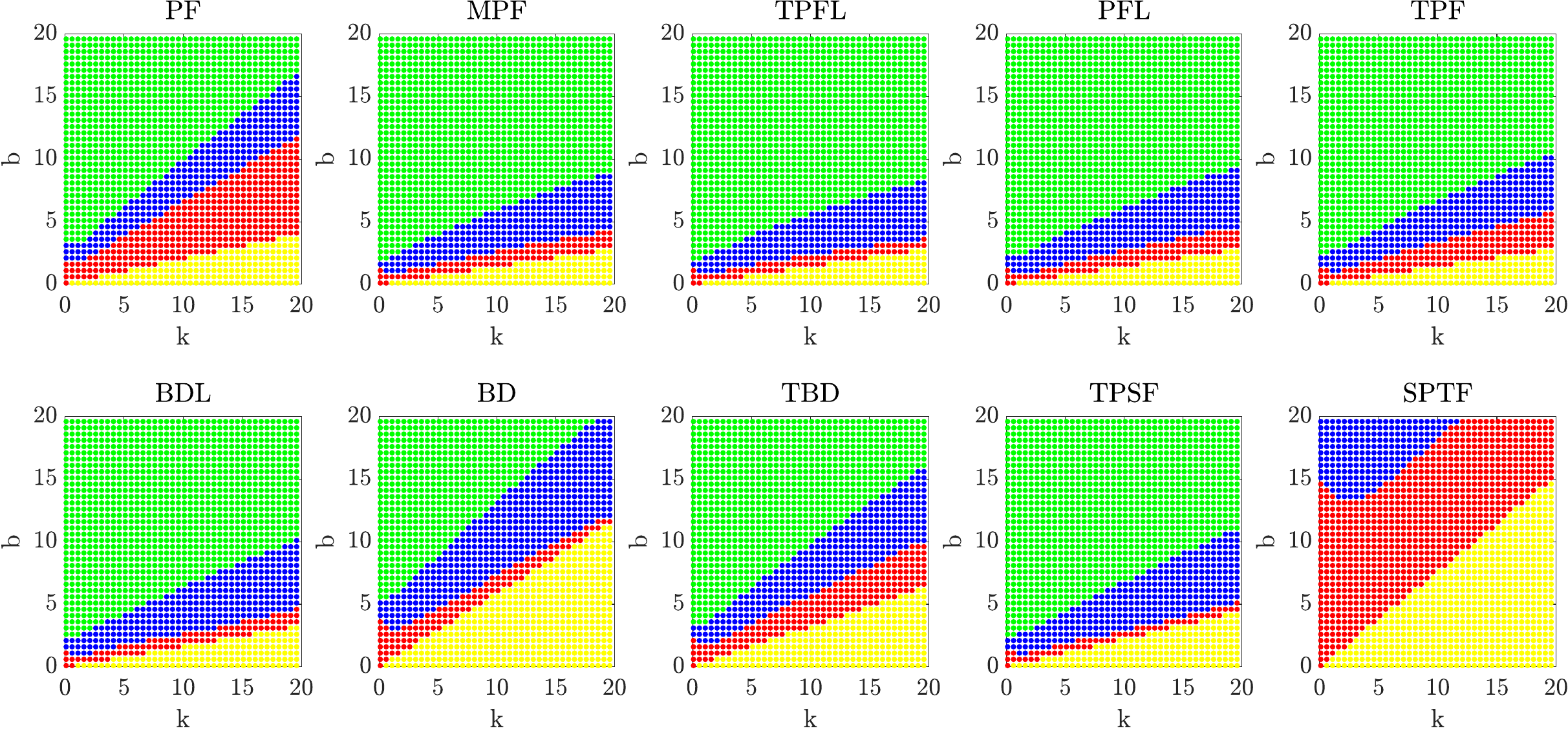}}
\vspace{0.4em}
  \subfloat[Under Case 2 (see Table \ref{table1}) and Acc. 3 (see Fig. \ref{leader_av}). \label{hete1_Acc3}]{%
        \includegraphics[width=0.8\linewidth]{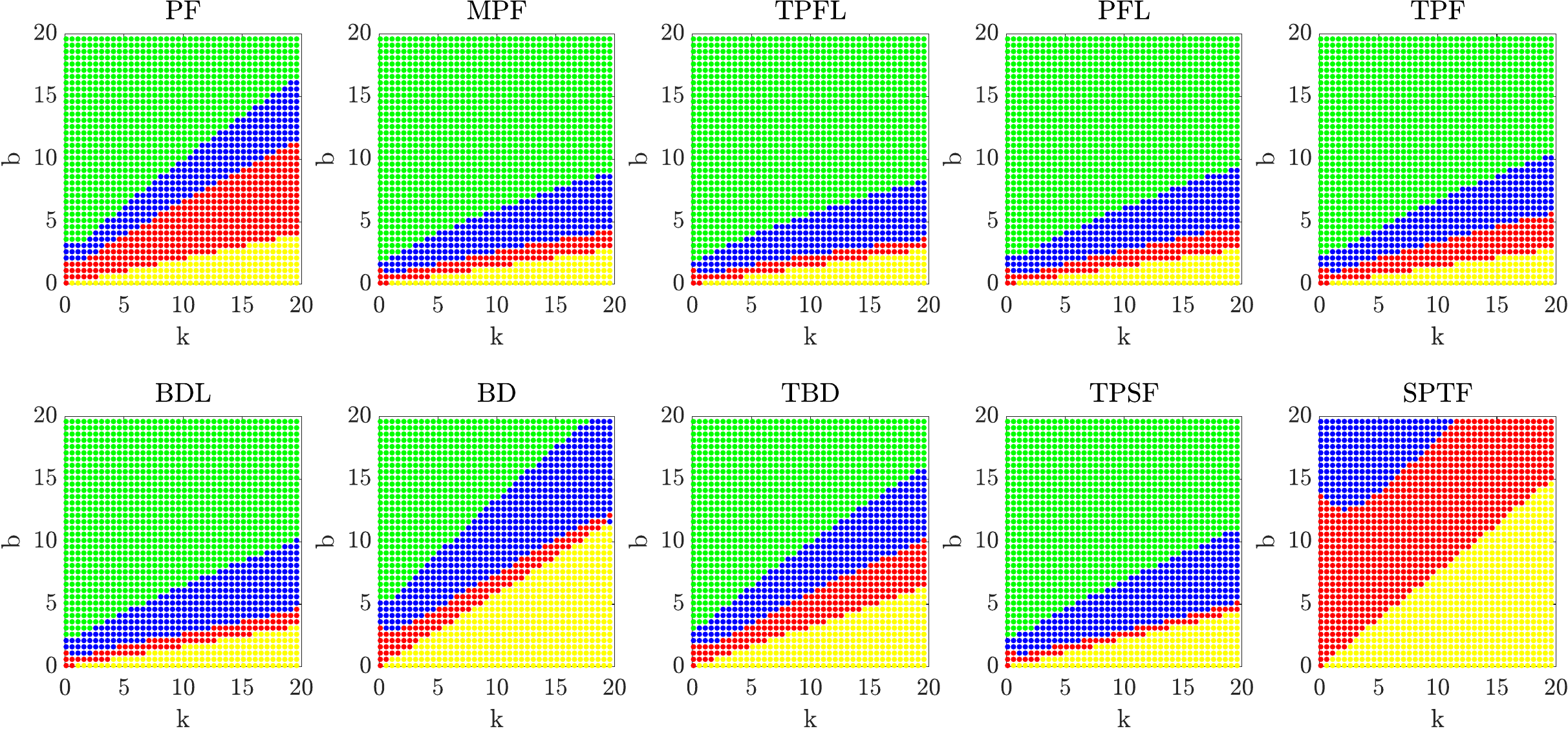}}
\vspace{0.4em}
\caption{Types of CGVs (see Table \ref{table-colors}) under Case 2 (see Table \ref{table1}) and common RCTs (see Figs. \ref{typicalu}-\ref{typicalb}). Given $\mathbf{K}=[k,\;b,\;h]$ (see Table \ref{CGS}), the X-axis and Y-axis represent control gains $k$ and $b$, respectively.}
  \label{hete1_areas} 
\end{figure}
\begin{figure}[htbp!] 
\centering
  \subfloat[Under Case 3 (see Table \ref{table1}) and Acc. 1 (see Fig. \ref{leader_av}).\label{hete2_Acc1}]{%
       \includegraphics[width=0.8\linewidth]{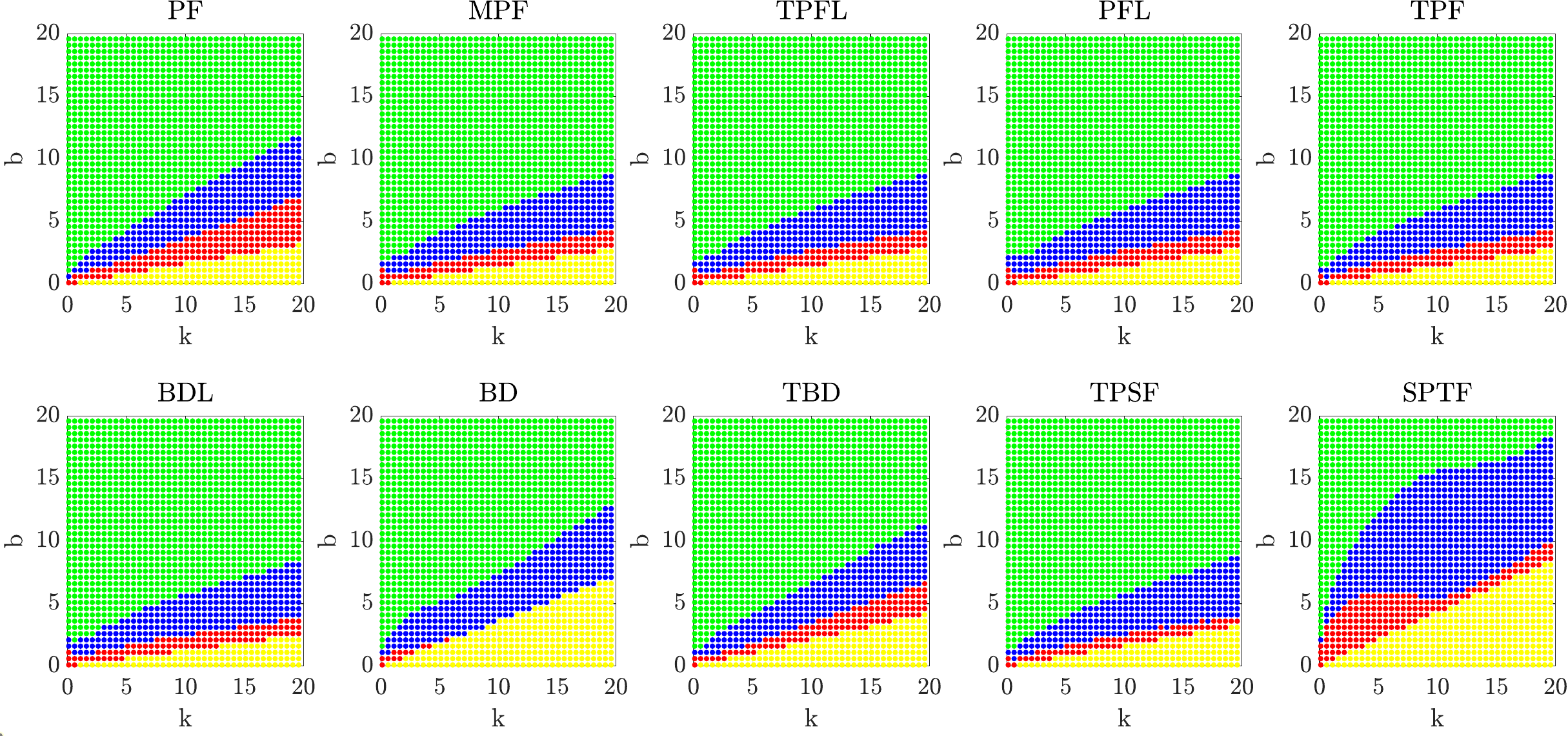}}
\vspace{0.4em}
\subfloat[Under Case 3 (see Table \ref{table1}) and Acc. 2 (see Fig. \ref{leader_av}). \label{hete2_Acc2}]{%
        \includegraphics[width=0.8\linewidth]{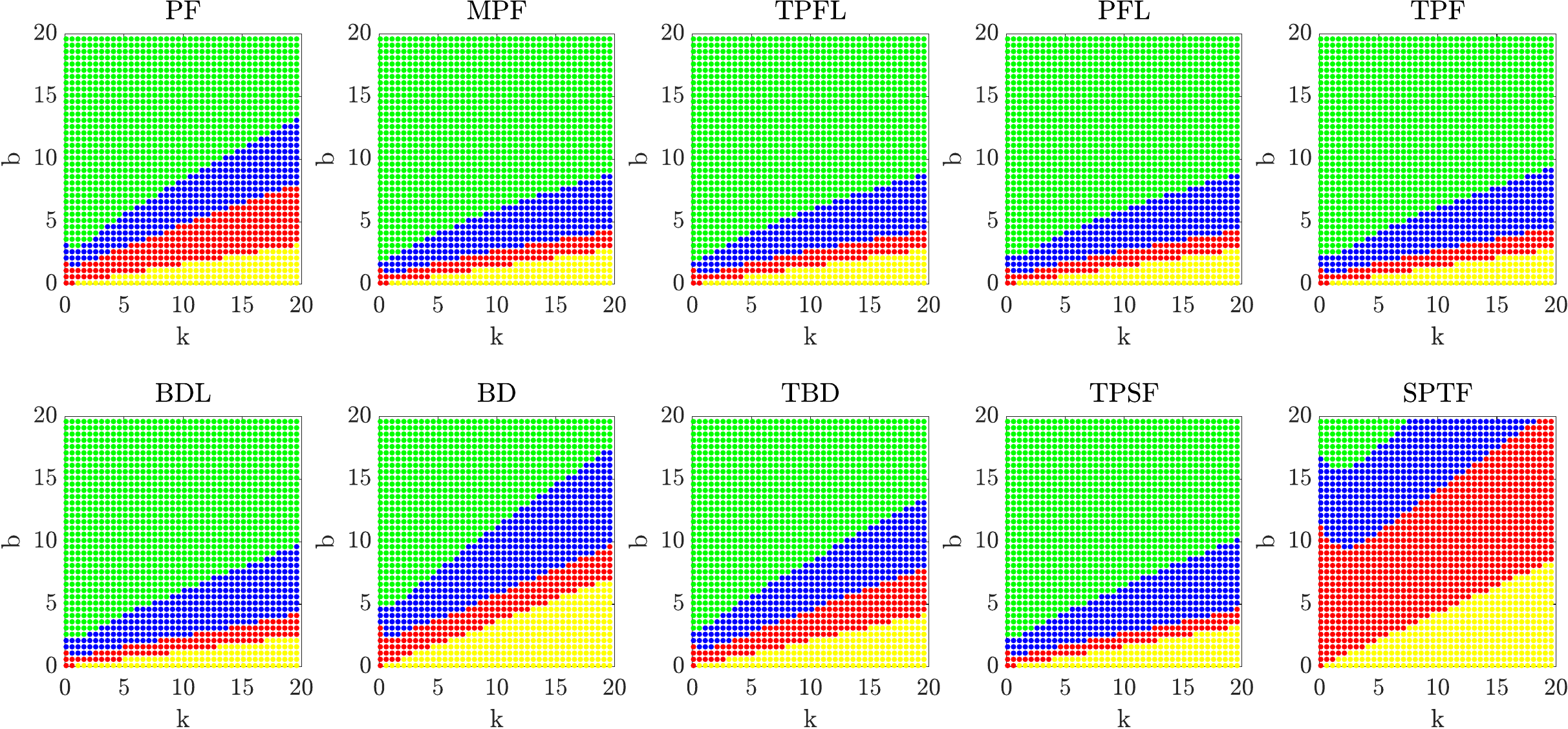}}
\vspace{0.4em}
  \subfloat[Under Case 3 (see Table \ref{table1}) and Acc. 3 (see Fig. \ref{leader_av}). \label{hete2_Acc3}]{%
        \includegraphics[width=0.8\linewidth]{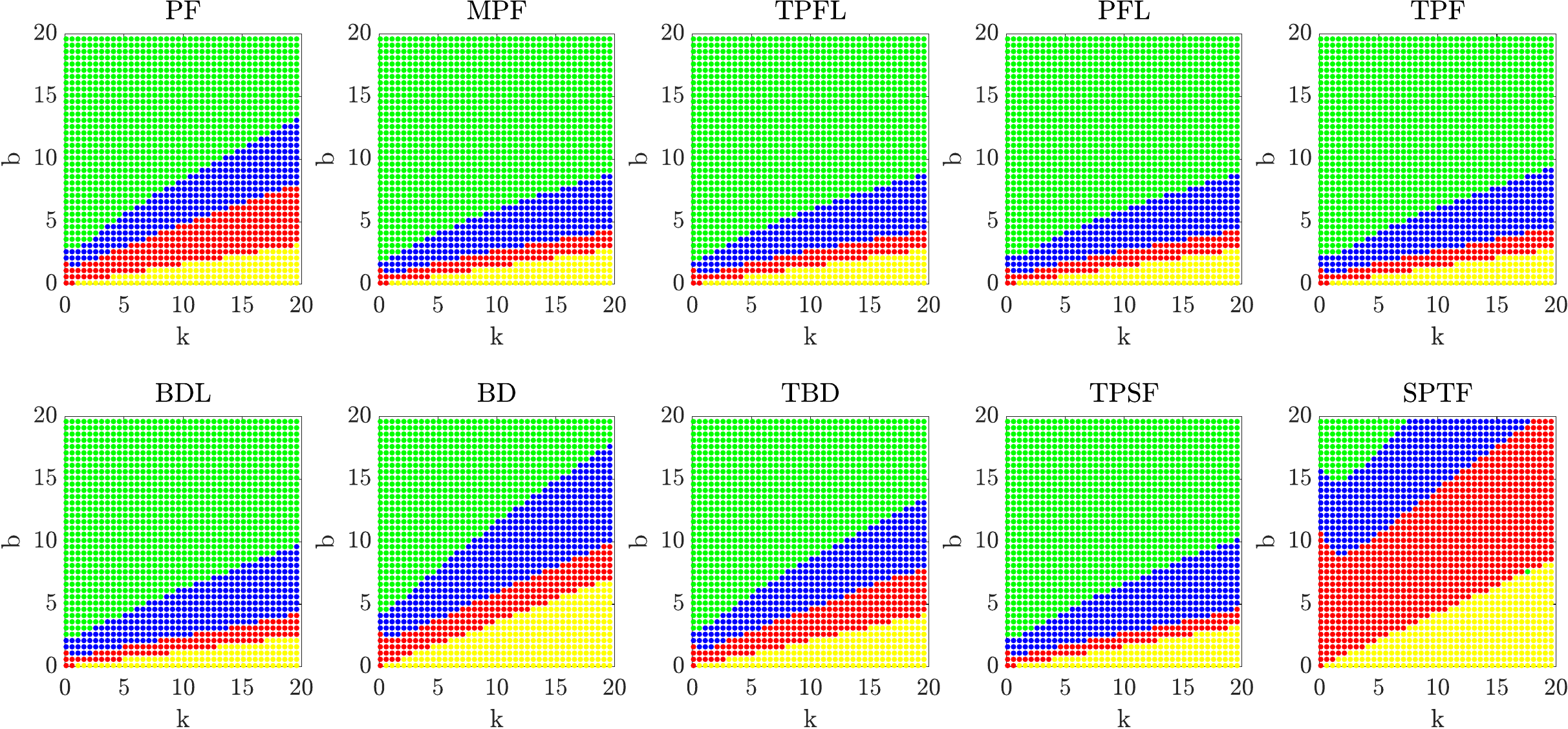}}
\vspace{0.4em}
\caption{Types of CGVs (see Table \ref{table-colors}) under Case 3 (see Table \ref{table1}) and common RCTs (see Figs. \ref{typicalu}-\ref{typicalb}). Given $\mathbf{K}=[k,\;b,\;h]$ (see Table \ref{CGS}), the X-axis and Y-axis represent control gains $k$ and $b$, respectively.}
  \label{hete2_areas} 
\end{figure}
\subsection{\textbf{Preliminaries 2 (Statistical Tools)}}
Pooled mean (PM), pooled standard deviation (PSD), and coefficient of variation (CV) are tools we use to combine data from different groups  \cite{borenstein2021introduction}. Also, we use a performance index (PI) that combines the mean and CV of metrics to integrate robustness into our performance metrics evaluation. 

1. \textbf{Pooled Mean (PM)}: Calculates the weighted average of means from different groups. This gives larger groups a proportionately greater influence.

2. \textbf{Pooled Standard Deviation (PSD)}: Combines the standard deviations of different groups into one measure. It accounts for sample sizes (weights) and variances, providing an unbiased estimate of the overall standard deviation. 

Thus, for $r$ groups with means $\mu_{i}$; $i=1,\dots,r$, standard deviations $SD_{i}$; $i=1,\dots,r$, and sample sizes $s_{i}$, $i=1,\dots,r$, PM and PSD values are calculated according to:
\begin{equation}
\text{PM} = \frac{\sum_{i=1}^{r} s_{i} \mu_{i}}{\sum_{i=1}^{r} s_{i}},
\hspace{1cm}
\text{PSD} =  \sqrt{\frac{\sum_{i=1}^{r} (s_{i}-1)SD_{i}^{2}}{\sum_{i=1}^{r} (s_{i}-1)}}
\label{pmpstd}
\end{equation}

3. \textbf{Coefficient of Variation (CV)}: It is a statistical measure representing the ratio of the SD to the mean. It provides a standardized measure of dispersion, indicating how much variability exists in relation to the average value. For the pooled values, it is calculated as $\text{CV} = \text{PSD}/\text{PM}$.
\begin{remark}
A lower CV suggests more consistency and less variability, which is a key characteristic of robustness. Conversely, a higher CV indicates greater variability and potential sensitivity to changes, implying lower robustness.   
\end{remark}
4.  \textbf{Performance Index (PI)}: To incorporate robustness into performance analysis, we combine PM and CV of each performance metric through the definition of PI
= PM+CV.
\begin{remark}
Smaller means are favorable for SaCGDI, AAPMTTC, AAMDRAC, AAMEEI, AAMEA, and AAMEJ metrics, and the PI rewards lower means. The PI also favors datasets with lower CV, indicating less variability. Thus, a lower PI value is better.
\end{remark}
\subsection{\textbf{Platoon Performance Analysis}}
Reminding  simulations setup [see section \ref{simprer}], we have three trajectories for  leader's acceleration [see \eqref{accvel}], three sets of ETCs (see Table \ref{table1}), and 1600 CGVs [see Table \ref{CGS}]. First, based on subsection \ref{safesa} and Table \ref{table-colors}, the types of CGVs for Cases 1-3 [see Table \ref{table1}], depending on leader's acceleration and the deployed RCT, are depicted in Figs. \ref{homo_areas}-\ref{hete2_areas}, respectively. Next, we need to find shared control gain vectors (SCGVs) [see \ref{SCGVs} and Remark \ref{shcgvs}] under each scenario, shown in Figs. \ref{homo_areas}-\ref{hete2_areas}. SCGVs are shown in Fig. \ref{commonCGs} where for instance, Fig. \ref{shared1} depicts SCGVs  under Case 1 and Acc. 1., which has been obtained using results in Fig. \ref{homo_Acc1}. 
\begin{figure}[htbp!]
    \centering
  \subfloat[Between all RCTs (see Fig. \ref{homo_Acc1}). \label{shared1}]{%
       \includegraphics[width=0.31\columnwidth]{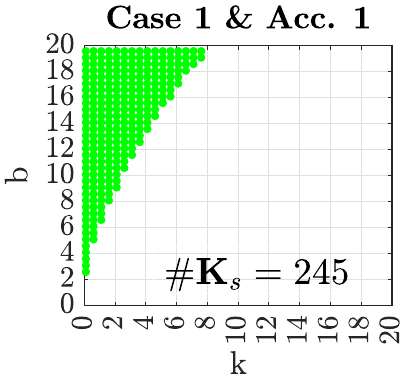}}
       \hfill
\subfloat[Between all RCTs, except SPTF (see Fig. \ref{homo_Acc2}). \label{shared2}]{%
        \includegraphics[width=0.31\columnwidth]{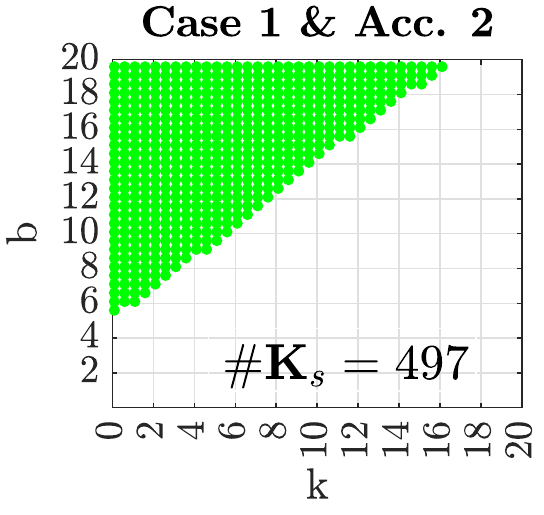}}
\hfill
  \subfloat[Between all RCTs, except SPTF (see Fig. \ref{homo_Acc3}). \label{shared3}]{%
        \includegraphics[width=0.31\columnwidth]{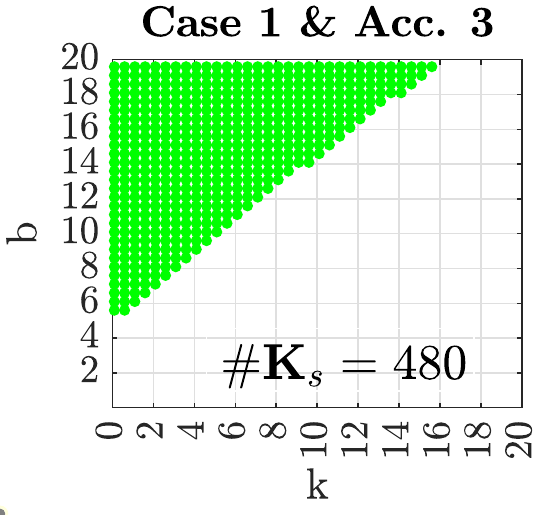}}
\hfill
  \subfloat[Between all RCTs (see Fig. \ref{hete1_Acc1}). \label{shared4}]{%
       \includegraphics[width=0.31\columnwidth]{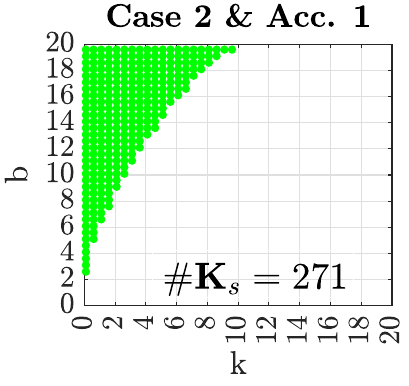}}
       \hfill
\subfloat[Between all RCTs, except SPTF (see Fig. \ref{hete1_Acc2}). \label{shared5}]{%
        \includegraphics[width=0.31\columnwidth]{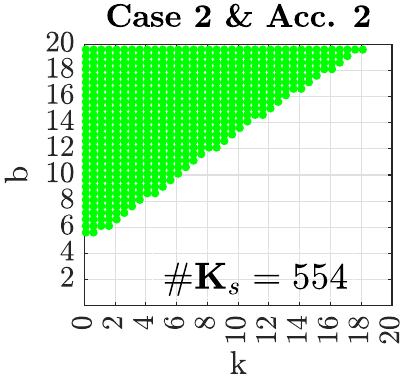}}
\hfill
  \subfloat[Between all RCTs, except SPTF (see Fig. \ref{hete1_Acc3}). \label{shared6}]{%
        \includegraphics[width=0.31\columnwidth]{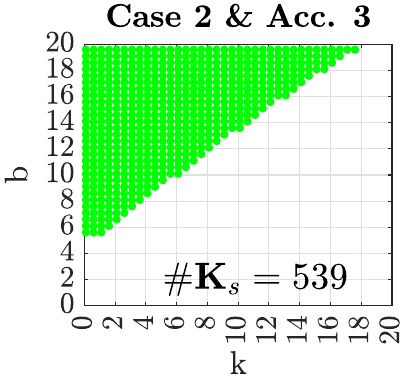}}
\hfill
  \subfloat[Between all RCTs (see Fig. \ref{hete2_Acc1}). \label{shared7}]{%
       \includegraphics[width=0.31\columnwidth]{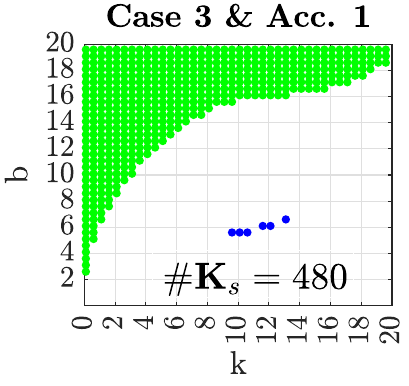}}
       \hfill
\subfloat[Between all RCTs (see Fig. \ref{hete2_Acc2}). \label{shared8}]{%
        \includegraphics[width=0.31\columnwidth]{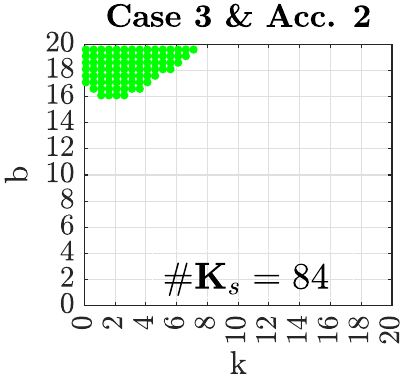}}
\hfill
  \subfloat[Between all RCTs (see Fig. \ref{hete2_Acc3}). \label{shared9}]{%
        \includegraphics[width=0.31\columnwidth]{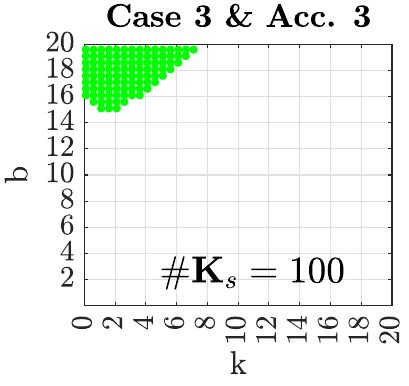}}
\caption{$\#\mathbf{K}_{s}$; $\mathbf{K}_{s}=[k,\;b,\;4]$ shows the number of shared control gain vectors (green/blue points), referred to as Weight in Tables \ref{safety_table}, \ref{energy_table}, \ref{comfort_table}, and \ref{overall_table}.}
\label{commonCGs} 
\end{figure}

Now, to highlight the differences in platoon performance under different RCTs [see Figs. \ref{typicalu}-\ref{typicalb}], we use four metrics defined earlier, and the results are as follows:

1. \textbf{Results for SaCGDI}:
\setlength{\tabcolsep}{2pt}
\begin{table*}[h]
\centering
\caption{Values of SaCGDI (see \eqref{SaCGDI}) metric under different common RCTs. The numbers corresponding to Case 1, Case 2, and Case 3 are associated with Fig. \ref{homo_areas}, Fig. \ref{hete1_areas}, and Fig. \ref{hete2_areas}, respectively. The RCTs are illustrated in Figs. \ref{typicalu}-\ref{typicalb}. Also, Acc. 1, Acc. 2, and Acc. 3; leader vehicle's accelerations, are depicted in Fig. \ref{leader_av} where their functions are provided in \eqref{accvel}. The term 'Rank (Asce.)' implies that the communications are ranked in an ascending order.}
\renewcommand{\arraystretch}{1.4}
\resizebox{1\textwidth}{!}{%
\begin{tabular}{ccc|cccccccccc|c}
\cline{4-14}
\multicolumn{1}{c}{} &
  \multicolumn{1}{c}{} &
  \multicolumn{1}{c|}{} &
  \multicolumn{10}{c|}{\textit{Rigid Communication Topologies (RCTs)}} &
  \multicolumn{1}{c|}{} \\ \cline{4-13}
\multicolumn{1}{c}{\multirow{-2}{*}{}} &
  \multicolumn{1}{c}{\multirow{-2}{*}{}} &
  \multicolumn{1}{c|}{\multirow{-2}{*}{}} &
  \multicolumn{1}{c|}{\textbf{PF}} &
  \multicolumn{1}{c|}{\textbf{MPF}} &
  \multicolumn{1}{c|}{\textbf{TPFL}} &
  \multicolumn{1}{c|}{\textbf{PFL}} &
  \multicolumn{1}{c|}{\textbf{TPF}} &
  \multicolumn{1}{c|}{\textbf{BDL}} &
  \multicolumn{1}{c|}{\textbf{BD}} &
  \multicolumn{1}{c|}{\textbf{TBPF}} &
  \multicolumn{1}{c|}{\textbf{TPSF}} &
  \textbf{SPTF} &
  \multicolumn{1}{c|}{\multirow{-2}{*}{Weight}} \\ \Xhline{2pt}
\multicolumn{1}{|c|}{{\color[HTML]{0000FF} }} &
  \multicolumn{1}{c|}{} &
  \textbf{Acc. 1} &
  \multicolumn{1}{c|}{46.062} &
  \multicolumn{1}{c|}{29.875} &
  \multicolumn{1}{c|}{29.812} &
  \multicolumn{1}{c|}{30.625} &
  \multicolumn{1}{c|}{33.188} &
  \multicolumn{1}{c|}{31.188} &
  \multicolumn{1}{c|}{51.875} &
  \multicolumn{1}{c|}{43.625} &
  \multicolumn{1}{c|}{32.063} &
  84.625 &
  \multicolumn{1}{c|}{1} \\ \cline{3-14} 
\multicolumn{1}{|c|}{{\color[HTML]{0000FF} }} &
  \multicolumn{1}{c|}{} &
  \textbf{Acc. 2} &
  \multicolumn{1}{c|}{53.500} &
  \multicolumn{1}{c|}{31.000} &
  \multicolumn{1}{c|}{30.125} &
  \multicolumn{1}{c|}{30.688} &
  \multicolumn{1}{c|}{35.125} &
  \multicolumn{1}{c|}{35.000} &
  \multicolumn{1}{c|}{68.938} &
  \multicolumn{1}{c|}{51.937} &
  \multicolumn{1}{c|}{37.250} &
  100.000 &
  \multicolumn{1}{c|}{1} \\ \cline{3-14} 
\multicolumn{1}{|c|}{{\color[HTML]{0000FF} }} &
  \multicolumn{1}{c|}{\multirow{-3}{*}{\textbf{\begin{tabular}[c]{@{}c@{}}Case 1:\\ $\mbox{\boldmath{$\tau$}}=[1,\;1,\;1,\;1]$\end{tabular}}}} &
  \textbf{Acc. 3} &
  \multicolumn{1}{c|}{53.188} &
  \multicolumn{1}{c|}{31.063} &
  \multicolumn{1}{c|}{30.063} &
  \multicolumn{1}{c|}{30.688} &
  \multicolumn{1}{c|}{35.188} &
  \multicolumn{1}{c|}{35.562} &
  \multicolumn{1}{c|}{70.000} &
  \multicolumn{1}{c|}{52.875} &
  \multicolumn{1}{c|}{38.062} &
  100.000 &
  \multicolumn{1}{c|}{1} \\ \cline{2-14} 
\multicolumn{1}{|c|}{{\color[HTML]{0000FF} }} &
  \multicolumn{1}{c|}{} &
  \textbf{Acc. 1} &
  \multicolumn{1}{c|}{43.250} &
  \multicolumn{1}{c|}{28.062} &
  \multicolumn{1}{c|}{27.688} &
  \multicolumn{1}{c|}{30.188} &
  \multicolumn{1}{c|}{31.437} &
  \multicolumn{1}{c|}{21.188} &
  \multicolumn{1}{c|}{46.500} &
  \multicolumn{1}{c|}{39.938} &
  \multicolumn{1}{c|}{30.563} &
  83.000 &
  \multicolumn{1}{c|}{1} \\ \cline{3-14} 
\multicolumn{1}{|c|}{{\color[HTML]{0000FF} }} &
  \multicolumn{1}{c|}{} &
  \textbf{Acc. 2} &
  \multicolumn{1}{c|}{49.750} &
  \multicolumn{1}{c|}{29.375} &
  \multicolumn{1}{c|}{27.750} &
  \multicolumn{1}{c|}{30.375} &
  \multicolumn{1}{c|}{33.250} &
  \multicolumn{1}{c|}{32.000} &
  \multicolumn{1}{c|}{65.375} &
  \multicolumn{1}{c|}{48.687} &
  \multicolumn{1}{c|}{34.188} &
  100.000 &
  \multicolumn{1}{c|}{1} \\ \cline{3-14} 
\multicolumn{1}{|c|}{{\color[HTML]{0000FF} }} &
  \multicolumn{1}{c|}{\multirow{-3}{*}{\textbf{\begin{tabular}[c]{@{}c@{}}Case 2:\\ $\mbox{\boldmath{$\tau$}}=[0.7,\;0.6,\;1,\;0.9]$\end{tabular}}}} &
  \textbf{Acc. 3} &
  \multicolumn{1}{c|}{49.562} &
  \multicolumn{1}{c|}{29.500} &
  \multicolumn{1}{c|}{27.812} &
  \multicolumn{1}{c|}{30.437} &
  \multicolumn{1}{c|}{33.312} &
  \multicolumn{1}{c|}{32.375} &
  \multicolumn{1}{c|}{66.312} &
  \multicolumn{1}{c|}{49.187} &
  \multicolumn{1}{c|}{34.625} &
  100.000 &
  \multicolumn{1}{c|}{1} \\ \cline{2-14} 
\multicolumn{1}{|c|}{{\color[HTML]{0000FF} }} &
  \multicolumn{1}{c|}{} &
  \textbf{Acc. 1} &
  \multicolumn{1}{c|}{35.125} &
  \multicolumn{1}{c|}{29.250} &
  \multicolumn{1}{c|}{29.250} &
  \multicolumn{1}{c|}{29.375} &
  \multicolumn{1}{c|}{29.125} &
  \multicolumn{1}{c|}{28.125} &
  \multicolumn{1}{c|}{38.500} &
  \multicolumn{1}{c|}{33.625} &
  \multicolumn{1}{c|}{28.438} &
  70.375 &
  \multicolumn{1}{c|}{1} \\ \cline{3-14} 
\multicolumn{1}{|c|}{{\color[HTML]{0000FF} }} &
  \multicolumn{1}{c|}{} &
  \textbf{Acc. 2} &
  \multicolumn{1}{c|}{41.438} &
  \multicolumn{1}{c|}{29.312} &
  \multicolumn{1}{c|}{29.312} &
  \multicolumn{1}{c|}{29.438} &
  \multicolumn{1}{c|}{30.312} &
  \multicolumn{1}{c|}{30.688} &
  \multicolumn{1}{c|}{56.000} &
  \multicolumn{1}{c|}{42.250} &
  \multicolumn{1}{c|}{32.375} &
  94.750 &
  \multicolumn{1}{c|}{1} \\ \cline{3-14} 
\multicolumn{1}{|c|}{\multirow{-9}{*}{\begin{tabular}[c]{@{}c@{}}\rotatebox[origin=c]{90}{\textbf{\color{bazaar}SaCGDI}} \end{tabular}}}  &
  \multicolumn{1}{c|}{\multirow{-3}{*}{\textbf{\begin{tabular}[c]{@{}c@{}}Case 3:\\ $\mbox{\boldmath{$\tau$}}=[0.7,\;0.8,\;0.4,\;0.5]$\end{tabular}}}} &
  \textbf{Acc. 3} &
  \multicolumn{1}{c|}{41.438} &
  \multicolumn{1}{c|}{29.312} &
  \multicolumn{1}{c|}{29.312} &
  \multicolumn{1}{c|}{29.438} &
  \multicolumn{1}{c|}{30.312} &
  \multicolumn{1}{c|}{31.250} &
  \multicolumn{1}{c|}{57.000} &
  \multicolumn{1}{c|}{42.562} &
  \multicolumn{1}{c|}{32.812} &
  92.688 &
  \multicolumn{1}{c|}{1} \\ \hline
\multicolumn{3}{|c|}{\textit{\begin{tabular}[c]{@{}c@{}}(\textbf{\color{magenta}PM}±\textbf{\color{teal}SD}) of \textbf{\color{bazaar}SaCGDI}\end{tabular}}} &
  \multicolumn{1}{c|}{45.906±6.569} &
  \multicolumn{1}{c|}{29.609±0.987} &
  \multicolumn{1}{c|}{28.914±1.020} &
  \multicolumn{1}{c|}{30.078±0.572} &
  \multicolumn{1}{c|}{32.258±2.298} &
  \multicolumn{1}{c|}{31.774±2.581} &
  \multicolumn{1}{c|}{58.578±11.351} &
  \multicolumn{1}{c|}{45.133±6.651} &
  \multicolumn{1}{c|}{33.539±3.220} &
  92.727±10.763 &
  \multicolumn{1}{c|}{9}\\\hline
  \multicolumn{3}{|c|}{\textit{\begin{tabular}[c]{@{}c@{}}\textbf{\color{violet}CV}\textbf{= {\color{teal}SD}/{\color{magenta}PM}}\end{tabular}}} &
  \multicolumn{1}{c|}{0.143} &
  \multicolumn{1}{c|}{0.033} &
  \multicolumn{1}{c|}{0.035} &
  \multicolumn{1}{c|}{0.019} &
  \multicolumn{1}{c|}{0.071} &
  \multicolumn{1}{c|}{0.081} &
  \multicolumn{1}{c|}{0.194} &
  \multicolumn{1}{c|}{0.147} &
  \multicolumn{1}{c|}{0.096} &
  0.116 &
  \multicolumn{1}{c}{} \\ \cline{1-13}
\multicolumn{3}{|c|}{\textit{\textbf{\color{orange}PI} = \textbf{\color{magenta}PM}+\textbf{\color{violet}CV}}} &
  \multicolumn{1}{c|}{46.049} &
  \multicolumn{1}{c|}{29.642} &
  \multicolumn{1}{c|}{28.949} &
  \multicolumn{1}{c|}{30.097} &
  \multicolumn{1}{c|}{32.329} &
  \multicolumn{1}{c|}{31.855} &
  \multicolumn{1}{c|}{58.772} &
  \multicolumn{1}{c|}{45.280} &
  \multicolumn{1}{c|}{33.635} &
  92.843 &
  \multicolumn{1}{c}{}
 \\ \customCline{1.7pt}{1}{13}
\multicolumn{3}{|c|}{{\color[HTML]{FF0000} \begin{tabular}{c}
\textit{\textbf{\underline{Joint Robustness \& SaCGDI Rank}}}  \\ {\color{black}\textit{based on \textbf{\color{orange}PI}  (Asce.)}} 
\end{tabular}}}
&
  \multicolumn{1}{c|}{8} &
  \multicolumn{1}{c|}{2} &
  \multicolumn{1}{c|}{1} &
  \multicolumn{1}{c|}{3} &
  \multicolumn{1}{c|}{5} &
  \multicolumn{1}{c|}{4} &
  \multicolumn{1}{c|}{9} &
  \multicolumn{1}{c|}{7} &
  \multicolumn{1}{c|}{6} &
  10 &
  \multicolumn{1}{c}{} \\ \cline{1-13}
\end{tabular}
}
\label{area_percentage}
\end{table*}
Using the CGVs in Figs. \ref{homo_areas}-\ref{hete2_areas}, the SaCGDI results are shown in Table \ref{area_percentage}. For example, under Case 1 with Acc. 1 and TPFL topology, only 29.812\% of CGVs lead to Unsafe distances, compared to 84.625\% under SPTF. The Performance Index (PI) for SaCGDI in Table \ref{area_percentage} ranks the RCTs, with TPFL, MPF, PFL, BDL, and TPF having smaller PIs, indicating more Stable-Safe CGVs. SPTF, BD, PF, TBPF, and TPSF have larger PIs, indicating fewer Stable-Safe CGVs. The SaCGDI distribution for each communication is visualized in Fig. \ref{boxplot_StCGDI}. TPFL, MPF, and PFL show more robust performance, while SPTF, BD, and PF are more sensitive.
\begin{figure}[htbp!]
\begin{center}
\resizebox{1\hsize}{!}{\includegraphics*{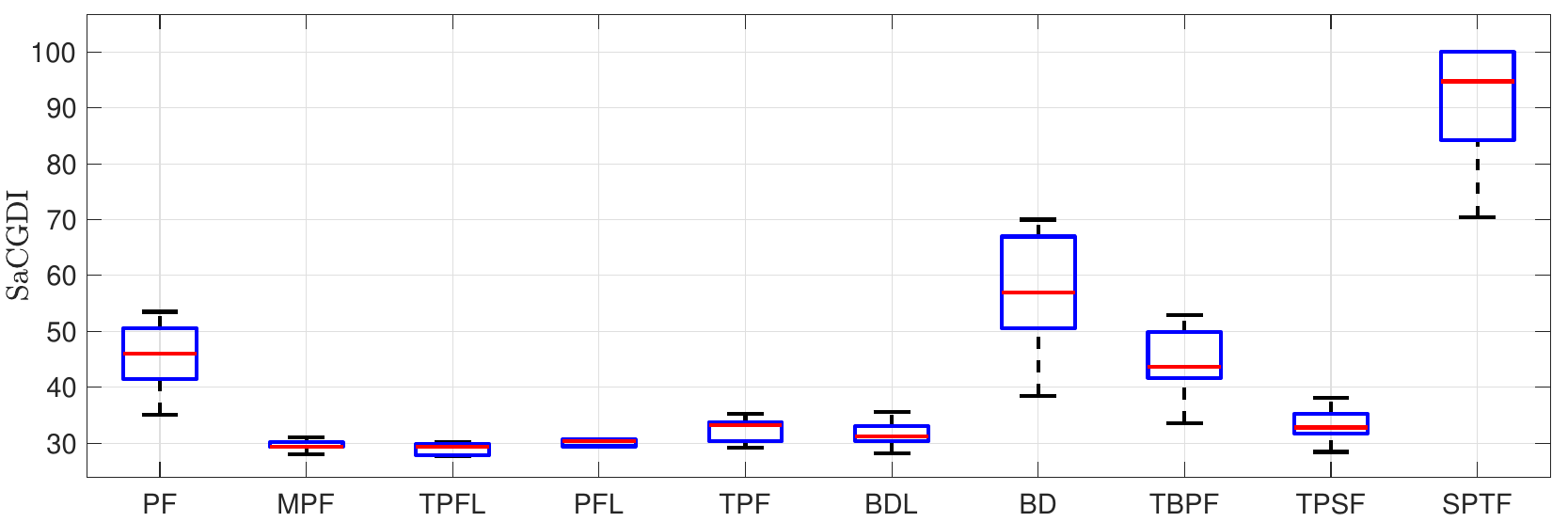}}
\caption{SaCGDI values across different RCTs is shown. Each box includes a median (red line), 25th and 75th percentiles (bottom and top edges), and whiskers extending to the extreme data points. Nine values per communication topology are detailed in Table \ref{area_percentage}.}
\label{boxplot_StCGDI}
\end{center}
\end{figure}

2. \textbf{Results for  Intervehicle Distance Safety}: To analyse this criterion, we have defined AAPMTTC [see \eqref{AAPMTTC}] and AAMDRAC [see \eqref{AADRAC}] metrics. Using SCGVs [see \ref{SCGVs} and Fig. \ref{commonCGs}], the results for AAPMTTC and AAMDRAC metrics across different scenarios are summarized in Table \ref{safety_table}. The `Weight' column in Table \ref{safety_table} shows the number of SCGVs (i.e., $\#(\mathbf{K}_{s}$)) for each scenario [see Fig. \ref{commonCGs}]. For example, under Case 3, Acc. 3, and the TBPF topology, the AAPMTTC and AAMDRAC values,  over the whole travel time, are $20.131\pm 17.741$ and $20.431\pm 10.462$, respectively. Note that the travel time is 25 s [see Fig. \ref{leader_av}], and thus AAPMTTC and AAMDRAC metric are considered up to $t=25$ s.

In Table \ref{safety_table}, the Performance Index (PI) for the AAPMTTC and AAMDRAC metrics is calculated, and the Average PI (API) is determined for the Intervehicle Distance Safety metric. Based on the API, the rank of each RCT is given in the last row of Table \ref{safety_table}. This ranking compares the performance of the platoon under various RCTs in terms of both robustness and safety. We observe that TPFL, BDL, MPF, PFL, and TPF are the top 5 topologies providing greater safety for the platoon, while SPTF, BD, PF, TBPF, and TPSF are the top 5 topologies providing less safety. As an example, Fig. \ref{metrics_plt_ncumsum} shows the APMTTC and AMDRAC trajectories under Case 1 with Acc. 1. Also, Fig. \ref{metrics_plt_platoon} shows the corresponding AAPMTTC and AAMDRAC trajectories, which their values at $t=25$ s indicate AAPMTTC and  AAMDRAC values over the whole travel time.
\setlength{\tabcolsep}{5pt}
\begin{table*}[]
\centering
\caption{Values of AAPMTTC (see \eqref{AAPMTTC}) and AAMDRAC (see \eqref{AADRAC}) metrics for different scenarios under SCGVs (see \ref{SCGVs} and Fig. \ref{commonCGs}) and different common RCTs. The RCTs are illustrated in Figs. \ref{typicalu}-\ref{typicalb}. Also, Acc. 1, Acc. 2, and Acc. 3; leader vehicle's accelerations, are depicted in Fig. \ref{leader_av} where their functions are provided in \eqref{accvel}. The term 'Rank (Asce.)' implies that the communications are ranked, based on API, in an ascending order. The number of SCGVs; $\#\mathbf{K}_{s}$ (see Fig. \ref{commonCGs}), are referred to as Weight in the table.}
\renewcommand{\arraystretch}{1.4}
\resizebox{1\textwidth}{!}{%
\begin{tabular}{ccc|cccccccccc|c}
\cline{4-14}
\multicolumn{1}{c}{} &
  \multicolumn{1}{c}{} &
  \multicolumn{1}{c|}{} &
  \multicolumn{10}{c|}{\textit{Rigid Communication Topologies (RCTs)}} &
  \multicolumn{1}{c|}{} \\ \cline{4-13}
\multicolumn{1}{c}{\multirow{-2}{*}{}} &
  \multicolumn{1}{c}{\multirow{-2}{*}{}} &
  \multicolumn{1}{c|}{\multirow{-2}{*}{}} &
  \multicolumn{1}{c|}{\textbf{PF}} &
  \multicolumn{1}{c|}{\textbf{MPF}} &
  \multicolumn{1}{c|}{\textbf{TPFL}} &
  \multicolumn{1}{c|}{\textbf{PFL}} &
  \multicolumn{1}{c|}{\textbf{TPF}} &
  \multicolumn{1}{c|}{\textbf{BDL}} &
  \multicolumn{1}{c|}{\textbf{BD}} &
  \multicolumn{1}{c|}{\textbf{TBPF}} &
  \multicolumn{1}{c|}{\textbf{TPSF}} &
  \textbf{SPTF} &
  \multicolumn{1}{c|}{\multirow{-2}{*}{Weight}} \\ \Xhline{2pt}
\multicolumn{1}{|c|}{} &
  \multicolumn{1}{c|}{} &
  \textbf{Acc. 1} &
  \multicolumn{1}{c|}{99.112±25.886} &
  \multicolumn{1}{c|}{50.672±38.287} &
  \multicolumn{1}{c|}{44.870±34.060} &
  \multicolumn{1}{c|}{44.835±32.290} &
  \multicolumn{1}{c|}{63.268±26.616} &
  \multicolumn{1}{c|}{43.255±31.720} &
  \multicolumn{1}{c|}{150.804±26.992} &
  \multicolumn{1}{c|}{70.443±27.817} &
  \multicolumn{1}{c|}{65.580±22.754} &
  181.180±105.008 &
  \multicolumn{1}{c|}{245} \\ \cline{3-14} 
\multicolumn{1}{|c|}{} &
  \multicolumn{1}{c|}{} &
  \textbf{Acc. 2} &
  \multicolumn{1}{c|}{114.226±53.891} &
  \multicolumn{1}{c|}{43.869±25.984} &
  \multicolumn{1}{c|}{42.680±25.700} &
  \multicolumn{1}{c|}{43.500±24.573} &
  \multicolumn{1}{c|}{53.276±17.812} &
  \multicolumn{1}{c|}{44.095±25.220} &
  \multicolumn{1}{c|}{179.687±121.986} &
  \multicolumn{1}{c|}{65.041±34.750} &
  \multicolumn{1}{c|}{59.212±17.822} &
  N.A. &
  \multicolumn{1}{c|}{497} \\ \cline{3-14} 
\multicolumn{1}{|c|}{} &
  \multicolumn{1}{c|}{\multirow{-3}{*}{\textbf{\begin{tabular}[c]{@{}c@{}}Case 1:\\ $\mbox{\boldmath{$\tau$}}=[1,\;1,\;1,\;1]$\end{tabular}}}} &
  \textbf{Acc. 3} &
  \multicolumn{1}{c|}{64.561±48.802} &
  \multicolumn{1}{c|}{32.119±26.903} &
  \multicolumn{1}{c|}{32.707±27.541} &
  \multicolumn{1}{c|}{33.468±26.710} &
  \multicolumn{1}{c|}{35.840±26.575} &
  \multicolumn{1}{c|}{34.658±25.483} &
  \multicolumn{1}{c|}{165.730±129.942} &
  \multicolumn{1}{c|}{52.633±40.286} &
  \multicolumn{1}{c|}{41.048±26.001} &
  N.A. &
  \multicolumn{1}{c|}{480} \\ \cline{2-14} 
\multicolumn{1}{|c|}{} &
  \multicolumn{1}{c|}{} &
  \textbf{Acc. 1} &
  \multicolumn{1}{c|}{96.857±35.157} &
  \multicolumn{1}{c|}{50.536±34.749} &
  \multicolumn{1}{c|}{45.960±31.689} &
  \multicolumn{1}{c|}{48.638±29.914} &
  \multicolumn{1}{c|}{63.402±31.154} &
  \multicolumn{1}{c|}{45.315±30.700} &
  \multicolumn{1}{c|}{154.096±31.163} &
  \multicolumn{1}{c|}{71.164±32.255} &
  \multicolumn{1}{c|}{65.112±28.853} &
  210.809±113.764 &
  \multicolumn{1}{c|}{271} \\ \cline{3-14} 
\multicolumn{1}{|c|}{} &
  \multicolumn{1}{c|}{} &
  \textbf{Acc. 2} &
  \multicolumn{1}{c|}{88.666±50.570} &
  \multicolumn{1}{c|}{38.917±23.088} &
  \multicolumn{1}{c|}{38.878±23.696} &
  \multicolumn{1}{c|}{40.702±21.871} &
  \multicolumn{1}{c|}{46.122±18.917} &
  \multicolumn{1}{c|}{40.882±21.417} &
  \multicolumn{1}{c|}{167.251±120.117} &
  \multicolumn{1}{c|}{61.871±31.511} &
  \multicolumn{1}{c|}{52.356±18.021} &
  N.A. &
  \multicolumn{1}{c|}{554} \\ \cline{3-14} 
\multicolumn{1}{|c|}{} &
  \multicolumn{1}{c|}{\multirow{-3}{*}{\textbf{\begin{tabular}[c]{@{}c@{}}Case 2:\\ $\mbox{\boldmath{$\tau$}}=[0.7,\;0.6,\;1,\;0.9]$\end{tabular}}}} &
  \textbf{Acc. 3} &
  \multicolumn{1}{c|}{55.155±41.946} &
  \multicolumn{1}{c|}{32.316±26.889} &
  \multicolumn{1}{c|}{32.737±27.305} &
  \multicolumn{1}{c|}{34.507±26.296} &
  \multicolumn{1}{c|}{35.739±26.275} &
  \multicolumn{1}{c|}{35.558±25.389} &
  \multicolumn{1}{c|}{152.061±129.453} &
  \multicolumn{1}{c|}{50.465±37.775} &
  \multicolumn{1}{c|}{39.939±25.803} &
  N.A. &
  \multicolumn{1}{c|}{539} \\ \cline{2-14} 
\multicolumn{1}{|c|}{} &
  \multicolumn{1}{c|}{} &
  \textbf{Acc. 1} &
  \multicolumn{1}{c|}{72.993±58.355} &
  \multicolumn{1}{c|}{45.163±42.270} &
  \multicolumn{1}{c|}{41.891±40.639} &
  \multicolumn{1}{c|}{41.772±41.365} &
  \multicolumn{1}{c|}{45.601±42.661} &
  \multicolumn{1}{c|}{41.200±40.901} &
  \multicolumn{1}{c|}{117.853±90.579} &
  \multicolumn{1}{c|}{55.062±56.007} &
  \multicolumn{1}{c|}{49.824±46.108} &
  176.053±152.719 &
  \multicolumn{1}{c|}{480} \\ \cline{3-14} 
\multicolumn{1}{|c|}{} &
  \multicolumn{1}{c|}{} &
  \textbf{Acc. 2} &
  \multicolumn{1}{c|}{57.332±22.919} &
  \multicolumn{1}{c|}{28.950±18.136} &
  \multicolumn{1}{c|}{28.256±18.370} &
  \multicolumn{1}{c|}{29.601±17.272} &
  \multicolumn{1}{c|}{29.658±15.896} &
  \multicolumn{1}{c|}{27.897±18.121} &
  \multicolumn{1}{c|}{45.915±20.626} &
  \multicolumn{1}{c|}{36.227±18.254} &
  \multicolumn{1}{c|}{32.185±15.016} &
  55.633±44.406 &
  \multicolumn{1}{c|}{84} \\ \cline{3-14} 
\multicolumn{1}{|c|}{\multirow{-9}{*}{\begin{tabular}[c]{@{}c@{}}\rotatebox[origin=c]{90}{$\mathbf{\color{blue}AAPMTTC(25)}$} \end{tabular}}} &
  \multicolumn{1}{c|}{\multirow{-3}{*}{\textbf{\begin{tabular}[c]{@{}c@{}}Case 3:\\ $\mbox{\boldmath{$\tau$}}=[0.7,\;0.8,\;0.4,\;0.5]$\end{tabular}}}} &
  \textbf{Acc. 3} &
  \multicolumn{1}{c|}{27.297±21.902} &
  \multicolumn{1}{c|}{18.660±18.724} &
  \multicolumn{1}{c|}{19.575±20.253} &
  \multicolumn{1}{c|}{19.213±19.861} &
  \multicolumn{1}{c|}{18.750±18.557} &
  \multicolumn{1}{c|}{19.218±20.053} &
  \multicolumn{1}{c|}{23.772±21.566} &
  \multicolumn{1}{c|}{20.131±17.741} &
  \multicolumn{1}{c|}{18.707±17.439} &
  50.140±47.877 &
  \multicolumn{1}{c|}{100} \\ \hline
\multicolumn{3}{|c|}{\textit{\begin{tabular}[c]{@{}c@{}}\textit{\textbf{{(\color{magenta}PM}±{\color{teal}PSD})}} \textit{of \textbf{\color{blue}AAPMMTTC}} \end{tabular}}} &
  \multicolumn{1}{c|}{79.914±47.007} &
  \multicolumn{1}{c|}{39.472±30.256} &
  \multicolumn{1}{c|}{38.148±29.479} &
  \multicolumn{1}{c|}{39.217±28.633} &
  \multicolumn{1}{c|}{45.364±27.363} &
  \multicolumn{1}{c|}{39.164±28.298} &
  \multicolumn{1}{c|}{149.226±106.723} &
  \multicolumn{1}{c|}{57.568±38.179} &
  \multicolumn{1}{c|}{49.805±27.424} &
  N.A. &
  \multicolumn{1}{c|}{3,250}\\
  \hline
\multicolumn{3}{|c|}{\textit{\begin{tabular}[c]{@{}c@{}}\textbf{\color{violet}CV}\textbf{= {\color{teal}PSD}/{\color{magenta}PM}}\end{tabular}}} &
  \multicolumn{1}{c|}{0.588} &
  \multicolumn{1}{c|}{0.767} &
  \multicolumn{1}{c|}{0.773} &
  \multicolumn{1}{c|}{0.730} &
  \multicolumn{1}{c|}{0.603} &
  \multicolumn{1}{c|}{0.723} &
  \multicolumn{1}{c|}{0.715} &
  \multicolumn{1}{c|}{0.663} &
  \multicolumn{1}{c|}{0.551} &
  N.A. &
  \multicolumn{1}{c}{} \\ \cline{1-13}
\multicolumn{3}{|c|}{\textit{\textbf{\color{orange}PI} = \textbf{\color{magenta}PM}+\textbf{\color{violet}CV}}} &
  \multicolumn{1}{c|}{80.502} &
  \multicolumn{1}{c|}{40.239} &
  \multicolumn{1}{c|}{38.921} &
  \multicolumn{1}{c|}{39.947} &
  \multicolumn{1}{c|}{45.967} &
  \multicolumn{1}{c|}{39.887} &
  \multicolumn{1}{c|}{149.941} &
  \multicolumn{1}{c|}{58.231} &
  \multicolumn{1}{c|}{50.356} &
  N.A. &
  \multicolumn{1}{c}{} 
  \\ \Xhline{2pt}
\multicolumn{1}{|c|}{} &
  \multicolumn{1}{c|}{} &
  \textbf{Acc. 1} &
  \multicolumn{1}{c|}{67.343±43.617} &
  \multicolumn{1}{c|}{37.114±22.866} &
  \multicolumn{1}{c|}{36.269±22.187} &
  \multicolumn{1}{c|}{36.791±22.003} &
  \multicolumn{1}{c|}{39.749±22.208} &
  \multicolumn{1}{c|}{36.423±22.198} &
  \multicolumn{1}{c|}{72.968±31.233} &
  \multicolumn{1}{c|}{47.218±31.591} &
  \multicolumn{1}{c|}{40.685±22.355} &
  107.670±96.557 &
  \multicolumn{1}{c|}{245} \\ \cline{3-14} 
\multicolumn{1}{|c|}{} &
  \multicolumn{1}{c|}{} &
  \textbf{Acc. 2} &
  \multicolumn{1}{c|}{82.561±78.474} &
  \multicolumn{1}{c|}{55.514±34.580} &
  \multicolumn{1}{c|}{55.125±33.978} &
  \multicolumn{1}{c|}{55.646±34.111} &
  \multicolumn{1}{c|}{57.293±35.045} &
  \multicolumn{1}{c|}{55.644±33.626} &
  \multicolumn{1}{c|}{118.521±95.089} &
  \multicolumn{1}{c|}{65.543±41.318} &
  \multicolumn{1}{c|}{58.338±34.092} &
  N.A. &
  \multicolumn{1}{c|}{497} \\ \cline{3-14} 
\multicolumn{1}{|c|}{} &
  \multicolumn{1}{c|}{\multirow{-3}{*}{\textbf{\begin{tabular}[c]{@{}c@{}}Case 1:\\ $\mbox{\boldmath{$\tau$}}=[1,\;1,\;1,\;1]$\end{tabular}}}} &
  \textbf{Acc. 3} &
  \multicolumn{1}{c|}{68.039±56.128} &
  \multicolumn{1}{c|}{53.114±33.964} &
  \multicolumn{1}{c|}{52.828±33.497} &
  \multicolumn{1}{c|}{53.334±33.622} &
  \multicolumn{1}{c|}{54.552±34.484} &
  \multicolumn{1}{c|}{53.307±33.122} &
  \multicolumn{1}{c|}{110.571±92.072} &
  \multicolumn{1}{c|}{61.025±40.444} &
  \multicolumn{1}{c|}{55.167±33.831} &
  N.A. &
  \multicolumn{1}{c|}{480} \\ \cline{2-14} 
\multicolumn{1}{|c|}{} &
  \multicolumn{1}{c|}{} &
  \textbf{Acc. 1} &
  \multicolumn{1}{c|}{62.156±44.594} &
  \multicolumn{1}{c|}{34.341±19.520} &
  \multicolumn{1}{c|}{33.954±19.518} &
  \multicolumn{1}{c|}{34.804±19.237} &
  \multicolumn{1}{c|}{37.826±19.976} &
  \multicolumn{1}{c|}{34.567±19.456} &
  \multicolumn{1}{c|}{83.418±40.380} &
  \multicolumn{1}{c|}{47.603±31.856} &
  \multicolumn{1}{c|}{38.846±19.876} &
  151.099±148.047 &
  \multicolumn{1}{c|}{271} \\ \cline{3-14} 
\multicolumn{1}{|c|}{} &
  \multicolumn{1}{c|}{} &
  \textbf{Acc. 2} &
  \multicolumn{1}{c|}{78.038±69.312} &
  \multicolumn{1}{c|}{56.680±34.914} &
  \multicolumn{1}{c|}{56.365±34.425} &
  \multicolumn{1}{c|}{57.185±34.778} &
  \multicolumn{1}{c|}{58.542±35.501} &
  \multicolumn{1}{c|}{57.094±34.441} &
  \multicolumn{1}{c|}{120.493±99.649} &
  \multicolumn{1}{c|}{66.583±42.585} &
  \multicolumn{1}{c|}{59.460±35.182} &
  N.A. &
  \multicolumn{1}{c|}{554} \\ \cline{3-14} 
\multicolumn{1}{|c|}{} &
  \multicolumn{1}{c|}{\multirow{-3}{*}{\textbf{\begin{tabular}[c]{@{}c@{}}Case 2:\\ $\mbox{\boldmath{$\tau$}}=[0.7,\;0.6,\;1,\;0.9]$\end{tabular}}}} &
  \textbf{Acc. 3} &
  \multicolumn{1}{c|}{67.952±52.159} &
  \multicolumn{1}{c|}{54.733±34.101} &
  \multicolumn{1}{c|}{54.463±33.678} &
  \multicolumn{1}{c|}{55.236±33.969} &
  \multicolumn{1}{c|}{56.227±34.637} &
  \multicolumn{1}{c|}{55.128±33.642} &
  \multicolumn{1}{c|}{111.823±94.532} &
  \multicolumn{1}{c|}{62.438±40.492} &
  \multicolumn{1}{c|}{56.856±34.417} &
  N.A. &
  \multicolumn{1}{c|}{539} \\ \cline{2-14} 
\multicolumn{1}{|c|}{} &
  \multicolumn{1}{c|}{} &
  \textbf{Acc. 1} &
  \multicolumn{1}{c|}{59.015±92.605} &
  \multicolumn{1}{c|}{50.568±59.172} &
  \multicolumn{1}{c|}{50.397±58.877} &
  \multicolumn{1}{c|}{50.963±59.574} &
  \multicolumn{1}{c|}{51.034±62.678} &
  \multicolumn{1}{c|}{50.775±59.323} &
  \multicolumn{1}{c|}{68.504±96.178} &
  \multicolumn{1}{c|}{54.747±76.415} &
  \multicolumn{1}{c|}{51.436±63.505} &
  108.338±222.662 &
  \multicolumn{1}{c|}{480} \\ \cline{3-14} 
\multicolumn{1}{|c|}{} &
  \multicolumn{1}{c|}{} &
  \textbf{Acc. 2} &
  \multicolumn{1}{c|}{22.561±10.310} &
  \multicolumn{1}{c|}{19.531±10.824} &
  \multicolumn{1}{c|}{19.445±10.858} &
  \multicolumn{1}{c|}{19.601±10.743} &
  \multicolumn{1}{c|}{20.361±10.829} &
  \multicolumn{1}{c|}{19.606±10.658} &
  \multicolumn{1}{c|}{27.504±11.313} &
  \multicolumn{1}{c|}{21.816±10.879} &
  \multicolumn{1}{c|}{20.413±10.523} &
  50.870±43.414 &
  \multicolumn{1}{c|}{84} \\ \cline{3-14} 
\multicolumn{1}{|c|}{\multirow{-9}{*}{\begin{tabular}[c]{@{}c@{}}\rotatebox[origin=c]{90}{$\mathbf{\color{olive}AAMDRAC(25)}$}\end{tabular}}} &
  \multicolumn{1}{c|}{\multirow{-3}{*}{\textbf{\begin{tabular}[c]{@{}c@{}}Case 3:\\ $\mbox{\boldmath{$\tau$}}=[0.7,\;0.8,\;0.4,\;0.5]$\end{tabular}}}} &
  \textbf{Acc. 3} &
  \multicolumn{1}{c|}{21.797±10.150} &
  \multicolumn{1}{c|}{18.885±10.794} &
  \multicolumn{1}{c|}{18.807±10.837} &
  \multicolumn{1}{c|}{19.023±10.699} &
  \multicolumn{1}{c|}{19.732±10.747} &
  \multicolumn{1}{c|}{19.043±10.575} &
  \multicolumn{1}{c|}{24.658±11.364} &
  \multicolumn{1}{c|}{20.431±10.462} &
  \multicolumn{1}{c|}{19.647±10.471} &
  46.313±38.139 &
  \multicolumn{1}{c|}{100} \\ \hline    
\multicolumn{3}{|c|}{\textit{\begin{tabular}[c]{@{}c@{}}\textit{\textbf{{(\color{magenta}PM}±{\color{teal}PSD})}} \textit{of \textbf{\color{olive}AAMDRAC}} \end{tabular}}} &
  \multicolumn{1}{c|}{67.476±65.281} &
  \multicolumn{1}{c|}{49.289±36.738} &
  \multicolumn{1}{c|}{48.963±36.342} &
  \multicolumn{1}{c|}{49.589±36.624} &
  \multicolumn{1}{c|}{50.944±37.887} &
  \multicolumn{1}{c|}{49.477±36.339} &
  \multicolumn{1}{c|}{97.583±86.012} &
  \multicolumn{1}{c|}{57.548±45.975} &
  \multicolumn{1}{c|}{51.669±37.784} &
  N.A. &
  \multicolumn{1}{c|}{3,250} \\\hline
\multicolumn{3}{|c|}{\textit{\begin{tabular}[c]{@{}c@{}} \textbf{\color{violet}CV} \textbf{= {\color{teal}PSD}/{\color{magenta}PM}}\end{tabular}}} &
  \multicolumn{1}{c|}{0.967} &
  \multicolumn{1}{c|}{0.745} &
  \multicolumn{1}{c|}{0.742} &
  \multicolumn{1}{c|}{0.739} &
  \multicolumn{1}{c|}{0.744} &
  \multicolumn{1}{c|}{0.734} &
  \multicolumn{1}{c|}{0.881} &
  \multicolumn{1}{c|}{0.799} &
  \multicolumn{1}{c|}{0.731} &
  N.A. &
  \multicolumn{1}{c}{} \\ 
  \cline{1-13}
\multicolumn{3}{|c|}{\textit{\textbf{\color{orange}PI} = \textbf{\color{magenta}PM}+\textbf{\color{violet}CV}}} &
  \multicolumn{1}{c|}{68.443} &
  \multicolumn{1}{c|}{50.034} &
  \multicolumn{1}{c|}{49.705} &
  \multicolumn{1}{c|}{50.328} &
  \multicolumn{1}{c|}{51.688} &
  \multicolumn{1}{c|}{50.211} &
  \multicolumn{1}{c|}{98.464} &
  \multicolumn{1}{c|}{58.347} &
  \multicolumn{1}{c|}{52.400} &
  N.A. &
  \multicolumn{1}{c}{}
  \\\customCline{1.7pt}{1}{13}
  \multicolumn{3}{|c|}{\textit{Avergae \textbf{\color{orange}PI} of  \textbf{\color{blue}AAPMMTC} \& \textbf{\color{olive}AAMDRAC} }} &
  \multicolumn{1}{c|}{74.473} &
  \multicolumn{1}{c|}{45.136} &
  \multicolumn{1}{c|}{44.313} &
  \multicolumn{1}{c|}{45.137} &
  \multicolumn{1}{c|}{48.827} &
  \multicolumn{1}{c|}{45.049} &
  \multicolumn{1}{c|}{124.203} &
  \multicolumn{1}{c|}{58.289} &
  \multicolumn{1}{c|}{51.378} &
  N.A. &
  \multicolumn{1}{c}{} 
  \\ \cline{1-13}
\multicolumn{3}{|c|}{{\color[HTML]{FF0000} \begin{tabular}{c}
\textit{\textbf{\underline{Joint Robustness \& Safety Rank}}}  \\ {\color{black}\textit{based on Average \textbf{\color{orange}PI}  (Asce.)}}
\end{tabular}}} &
  \multicolumn{1}{c|}{8} &
  \multicolumn{1}{c|}{3} &
  \multicolumn{1}{c|}{1} &
  \multicolumn{1}{c|}{4} &
  \multicolumn{1}{c|}{5} &
  \multicolumn{1}{c|}{2} &
  \multicolumn{1}{c|}{9} &
  \multicolumn{1}{c|}{7} &
  \multicolumn{1}{c|}{6} &
  10 &
  \multicolumn{1}{c}{} \\ \cline{1-13}
\end{tabular}
}
\label{safety_table}
\end{table*}

3. \textbf{Results for Energy Consumption}:
For this metric, we have defined AAMEEI [see \eqref{AAIE}], which requires SCGVs [see \ref{SCGVs}] for each scenario shown in Figs. \ref{homo_areas}-\ref{hete2_areas}. The SCGVs [see \ref{SCGVs} and Fig. \ref{commonCGs}] for each scenario are depicted in Fig. \ref{commonCGs}, and the corresponding AAMEEI metric values are summarized in Table \ref{energy_table}. For example, under Case 2, Acc. 1, and the BD topology, the AAMEEI value (over the whole travel time) is $369.222\pm 326.517$. In Table \ref{energy_table}, the Performance Index (PI) for the AAMEEI metric is calculated, and the rank of each RCT is provided in the last row of the table. This ranking compares the performance of the platoon under various RCTs in terms of energy  consumption and robustness. We observe that TPFL, MPF, PFL, BDL, and TPF are the top 5 topologies with the lowest energy  consumption, while SPTF, BD, PF, TPSF, and TBPF are the top 5 topologies with the highest energy  consumption. As an example, for Case 1 and Acc. 1, the AMEEI  trajectory is shown in Fig. \ref{metrics_plt_ncumsum}. Also, for AAMEEI, its trajectory is shown in Fig. \ref{metrics_plt_platoon}.
\setlength{\tabcolsep}{5pt}
\begin{table*}[]
\centering
\caption{Values of AAMEA (see \eqref{AAAE}) and AAMEJ (see \eqref{AAJE}) metrics for different scenarios under SCGVs (see \ref{SCGVs} and Fig. \ref{commonCGs}) and different common RCTs. The RCTs are illustrated in Figs. \ref{typicalu}-\ref{typicalb}. Also, Acc. 1, Acc. 2, and Acc. 3; leader vehicle's accelerations, are depicted in Fig. \ref{leader_av} where their functions are provided in \eqref{accvel}. The term 'Rank (Asce.)' implies that the communications are ranked, based on API, in an ascending order. The number of SCGVs; $\#\mathbf{K}_{s}$ (see Fig. \ref{commonCGs}), are referred to as Weight in the table.}
\renewcommand{\arraystretch}{1.4}
\resizebox{1\textwidth}{!}{%
\begin{tabular}{ccc|cccccccccc|c}
\cline{4-14}
\multicolumn{1}{l}{} &
  \multicolumn{1}{c}{} &
  \multicolumn{1}{c|}{} &
  \multicolumn{10}{c|}{\textit{Rigid Communication Topologies (RCTs)}} &
  \multicolumn{1}{c|}{} \\ \cline{4-13}
\multicolumn{1}{c}{\multirow{-2}{*}{}} &
  \multicolumn{1}{c}{\multirow{-2}{*}{}} &
  \multicolumn{1}{c|}{\multirow{-2}{*}{}} &
  \multicolumn{1}{c|}{\textbf{PF}} &
  \multicolumn{1}{c|}{\textbf{MPF}} &
  \multicolumn{1}{c|}{\textbf{TPFL}} &
  \multicolumn{1}{c|}{\textbf{PFL}} &
  \multicolumn{1}{c|}{\textbf{TPF}} &
  \multicolumn{1}{c|}{\textbf{BDL}} &
  \multicolumn{1}{c|}{\textbf{BD}} &
  \multicolumn{1}{c|}{\textbf{TBPF}} &
  \multicolumn{1}{c|}{\textbf{TPSF}} &
  \textbf{SPTF} &
  \multicolumn{1}{c|}{\multirow{-2}{*}{Weight}} \\ \Xhline{2pt}
\multicolumn{1}{|c|}{} &
  \multicolumn{1}{c|}{} &
  \textbf{Acc. 1} &
  \multicolumn{1}{c|}{75.575±62.973} &
  \multicolumn{1}{c|}{5.775±2.974} &
  \multicolumn{1}{c|}{4.028±1.211} &
  \multicolumn{1}{c|}{4.212±2.126} &
  \multicolumn{1}{c|}{11.925±9.404} &
  \multicolumn{1}{c|}{3.808±2.257} &
  \multicolumn{1}{c|}{337.696±260.869} &
  \multicolumn{1}{c|}{27.072±16.817} &
  \multicolumn{1}{c|}{16.923±13.042} &
  977.243±648.804 &
  \multicolumn{1}{c|}{245} \\ \cline{3-14} 
\multicolumn{1}{|c|}{} &
  \multicolumn{1}{c|}{} &
  \textbf{Acc. 2} &
  \multicolumn{1}{c|}{90.205±126.418} &
  \multicolumn{1}{c|}{5.544±6.107} &
  \multicolumn{1}{c|}{4.502±4.602} &
  \multicolumn{1}{c|}{6.851±8.137} &
  \multicolumn{1}{c|}{13.618±17.122} &
  \multicolumn{1}{c|}{8.116±9.128} &
  \multicolumn{1}{c|}{581.698±736.178} &
  \multicolumn{1}{c|}{43.888±46.125} &
  \multicolumn{1}{c|}{23.801±29.907} &
  N.A. &
  \multicolumn{1}{c|}{497} \\ \cline{3-14} 
\multicolumn{1}{|c|}{} &
  \multicolumn{1}{c|}{\multirow{-3}{*}{\textbf{\begin{tabular}[c]{@{}c@{}}Case 1:\\ $\mbox{\boldmath{$\tau$}}=[1,\;1,\;1,\;1]$\end{tabular}}}} &
  \textbf{Acc. 3} &
  \multicolumn{1}{c|}{65.219±100.989} &
  \multicolumn{1}{c|}{4.102±5.168} &
  \multicolumn{1}{c|}{3.464±4.107} &
  \multicolumn{1}{c|}{5.573±7.273} &
  \multicolumn{1}{c|}{10.044±14.016} &
  \multicolumn{1}{c|}{6.741±8.237} &
  \multicolumn{1}{c|}{524.985±718.508} &
  \multicolumn{1}{c|}{36.616±43.341} &
  \multicolumn{1}{c|}{18.630±25.878} &
  N.A. &
  \multicolumn{1}{c|}{480} \\ \cline{2-14} 
\multicolumn{1}{|c|}{} &
  \multicolumn{1}{c|}{} &
  \textbf{Acc. 1} &
  \multicolumn{1}{c|}{91.791±80.240} &
  \multicolumn{1}{c|}{6.859±3.691} &
  \multicolumn{1}{c|}{5.036±2.330} &
  \multicolumn{1}{c|}{5.898±3.710} &
  \multicolumn{1}{c|}{15.164±13.768} &
  \multicolumn{1}{c|}{6.617±4.340} &
  \multicolumn{1}{c|}{624.329±478.219} &
  \multicolumn{1}{c|}{45.694±28.997} &
  \multicolumn{1}{c|}{24.651±20.556} &
  2460.737±1507.486 &
  \multicolumn{1}{c|}{271} \\ \cline{3-14} 
\multicolumn{1}{|c|}{} &
  \multicolumn{1}{c|}{} &
  \textbf{Acc. 2} &
  \multicolumn{1}{c|}{89.326±134.853} &
  \multicolumn{1}{c|}{6.481±8.267} &
  \multicolumn{1}{c|}{5.627±7.186} &
  \multicolumn{1}{c|}{9.069±13.089} &
  \multicolumn{1}{c|}{14.622±20.885} &
  \multicolumn{1}{c|}{11.935±16.667} &
  \multicolumn{1}{c|}{695.455±984.738} &
  \multicolumn{1}{c|}{56.434±68.633} &
  \multicolumn{1}{c|}{27.997±39.858} &
  N.A. &
  \multicolumn{1}{c|}{554} \\ \cline{3-14} 
\multicolumn{1}{|c|}{} &
  \multicolumn{1}{c|}{\multirow{-3}{*}{\textbf{\begin{tabular}[c]{@{}c@{}}Case 2:\\ $\mbox{\boldmath{$\tau$}}=[0.7,\;0.6,\;1,\;0.9]$\end{tabular}}}} &
  \textbf{Acc. 3} &
  \multicolumn{1}{c|}{6.672±11.053} &
  \multicolumn{1}{c|}{0.498±0.707} &
  \multicolumn{1}{c|}{0.443±0.624} &
  \multicolumn{1}{c|}{0.735±1.137} &
  \multicolumn{1}{c|}{1.109±1.728} &
  \multicolumn{1}{c|}{0.974±1.452} &
  \multicolumn{1}{c|}{61.519±93.754} &
  \multicolumn{1}{c|}{4.612±6.208} &
  \multicolumn{1}{c|}{2.191±3.399} &
  N.A. &
  \multicolumn{1}{c|}{539} \\ \cline{2-14} 
\multicolumn{1}{|c|}{} &
  \multicolumn{1}{c|}{} &
  \textbf{Acc. 1} &
  \multicolumn{1}{c|}{93.584±132.789} &
  \multicolumn{1}{c|}{7.483±7.262} &
  \multicolumn{1}{c|}{5.505±5.251} &
  \multicolumn{1}{c|}{6.691±8.014} &
  \multicolumn{1}{c|}{15.276±19.234} &
  \multicolumn{1}{c|}{6.455±8.674} &
  \multicolumn{1}{c|}{386.086±636.151} &
  \multicolumn{1}{c|}{36.874±54.843} &
  \multicolumn{1}{c|}{24.126±36.028} &
  1240.491±1741.188 &
  \multicolumn{1}{c|}{480} \\ \cline{3-14} 
\multicolumn{1}{|c|}{} &
  \multicolumn{1}{c|}{} &
  \textbf{Acc. 2} &
  \multicolumn{1}{c|}{40.105±39.960} &
  \multicolumn{1}{c|}{2.808±2.124} &
  \multicolumn{1}{c|}{2.055±1.213} &
  \multicolumn{1}{c|}{2.840±2.317} &
  \multicolumn{1}{c|}{7.261±7.535} &
  \multicolumn{1}{c|}{3.103±2.431} &
  \multicolumn{1}{c|}{214.168±181.283} &
  \multicolumn{1}{c|}{19.636±15.378} &
  \multicolumn{1}{c|}{11.413±11.454} &
  823.291±408.590 &
  \multicolumn{1}{c|}{84} \\ \cline{3-14} 
\multicolumn{1}{|c|}{\multirow{-9}{*}{\begin{tabular}[c]{@{}c@{}}\rotatebox[origin=c]{90}{$\mathbf{\color{cyan}AAMEA(25)}$} \end{tabular}}} &
  \multicolumn{1}{c|}{\multirow{-3}{*}{\textbf{\begin{tabular}[c]{@{}c@{}}Case 3:\\ $\mbox{\boldmath{$\tau$}}=[0.7,\;0.8,\;0.4,\;0.5]$\end{tabular}}}} &
  \textbf{Acc. 3} &
  \multicolumn{1}{c|}{26.038±30.156} &
  \multicolumn{1}{c|}{1.897±1.737} &
  \multicolumn{1}{c|}{1.486±1.274} &
  \multicolumn{1}{c|}{2.333±2.444} &
  \multicolumn{1}{c|}{5.100±6.101} &
  \multicolumn{1}{c|}{2.589±2.542} &
  \multicolumn{1}{c|}{153.919±133.809} &
  \multicolumn{1}{c|}{13.445±11.881} &
  \multicolumn{1}{c|}{7.885±9.099} &
  680.977±337.387 &
  \multicolumn{1}{c|}{100} \\ \hline
\multicolumn{3}{|c|}{\textit{\begin{tabular}[c]{@{}c@{}}\textit{\textbf{{(\color{magenta}PM}±{\color{teal}PSD})}} of \textbf{\color{cyan}AAMEA} \end{tabular}}} &
  \multicolumn{1}{c|}{68.770±102.888} &
  \multicolumn{1}{c|}{4.884±5.586} &
  \multicolumn{1}{c|}{3.868±4.397} &
  \multicolumn{1}{c|}{5.481±7.661} &
  \multicolumn{1}{c|}{11.006±15.117} &
  \multicolumn{1}{c|}{6.385±9.162} &
  \multicolumn{1}{c|}{440.052±641.507} &
  \multicolumn{1}{c|}{34.722±44.267} &
  \multicolumn{1}{c|}{18.959±27.470} &
  N.A. &
  \multicolumn{1}{c|}{3,250} \\ \hline
    \multicolumn{3}{|c|}{\textit{\begin{tabular}[c]{@{}c@{}}\textbf{\color{violet}CV} \textbf{= {\color{teal}PSD}/{\color{magenta}PM}}\end{tabular}}} &
  \multicolumn{1}{c|}{1.496} &
  \multicolumn{1}{c|}{1.144} &
  \multicolumn{1}{c|}{1.137} &
  \multicolumn{1}{c|}{1.398} &
  \multicolumn{1}{c|}{1.374} &
  \multicolumn{1}{c|}{1.435} &
  \multicolumn{1}{c|}{1.458} &
  \multicolumn{1}{c|}{1.275} &
  \multicolumn{1}{c|}{1.449} &
  N.A. &
  \multicolumn{1}{c}{} \\ \cline{1-13}
\multicolumn{3}{|c|}{\textit{\textbf{\color{orange}PI} = \textbf{\color{magenta}PM}+\textbf{\color{violet}CV}}} &
  \multicolumn{1}{c|}{70.266} &
  \multicolumn{1}{c|}{6.028} &
  \multicolumn{1}{c|}{5.005} &
  \multicolumn{1}{c|}{6.879} &
  \multicolumn{1}{c|}{12.380} &
  \multicolumn{1}{c|}{7.820} &
  \multicolumn{1}{c|}{441.510} &
  \multicolumn{1}{c|}{35.997} &
  \multicolumn{1}{c|}{20.408} &
  N.A. &
  \multicolumn{1}{c}{}
  \\ \Xhline{2pt}
\multicolumn{1}{|c|}{} &
  \multicolumn{1}{c|}{} &
  \textbf{Acc. 1} &
  \multicolumn{1}{c|}{250.634±338.885} &
  \multicolumn{1}{c|}{38.853±54.342} &
  \multicolumn{1}{c|}{34.354±52.687} &
  \multicolumn{1}{c|}{43.494±71.685} &
  \multicolumn{1}{c|}{70.222±112.201} &
  \multicolumn{1}{c|}{40.633±65.892} &
  \multicolumn{1}{c|}{362.175±377.728} &
  \multicolumn{1}{c|}{85.483±105.848} &
  \multicolumn{1}{c|}{75.570±113.588} &
  449.682±332.875 &
  \multicolumn{1}{c|}{245} \\ \cline{3-14} 
\multicolumn{1}{|c|}{} &
  \multicolumn{1}{c|}{} &
  \textbf{Acc. 2} &
  \multicolumn{1}{c|}{575.385±1068.414} &
  \multicolumn{1}{c|}{90.859±154.920} &
  \multicolumn{1}{c|}{82.764±138.126} &
  \multicolumn{1}{c|}{106.622±181.510} &
  \multicolumn{1}{c|}{158.639±280.957} &
  \multicolumn{1}{c|}{110.481±184.068} &
  \multicolumn{1}{c|}{770.480±1348.776} &
  \multicolumn{1}{c|}{221.283±338.163} &
  \multicolumn{1}{c|}{190.466±330.367} &
  N.A. &
  \multicolumn{1}{c|}{497} \\ \cline{3-14} 
\multicolumn{1}{|c|}{} &
  \multicolumn{1}{c|}{\multirow{-3}{*}{\textbf{\begin{tabular}[c]{@{}c@{}}Case 1:\\ $\mbox{\boldmath{$\tau$}}=[1,\;1,\;1,\;1]$\end{tabular}}}} &
  \textbf{Acc. 3} &
  \multicolumn{1}{c|}{376.922±768.717} &
  \multicolumn{1}{c|}{73.773±131.969} &
  \multicolumn{1}{c|}{69.386±121.543} &
  \multicolumn{1}{c|}{90.337±159.114} &
  \multicolumn{1}{c|}{122.619±230.531} &
  \multicolumn{1}{c|}{93.391±161.071} &
  \multicolumn{1}{c|}{668.995±1278.252} &
  \multicolumn{1}{c|}{175.024±293.737} &
  \multicolumn{1}{c|}{145.516±270.280} &
  N.A. &
  \multicolumn{1}{c|}{480} \\ \cline{2-14} 
\multicolumn{1}{|c|}{} &
  \multicolumn{1}{c|}{} &
  \textbf{Acc. 1} &
  \multicolumn{1}{c|}{266.616±371.192} &
  \multicolumn{1}{c|}{47.122±69.617} &
  \multicolumn{1}{c|}{43.279±68.439} &
  \multicolumn{1}{c|}{55.185±91.826} &
  \multicolumn{1}{c|}{78.509±127.680} &
  \multicolumn{1}{c|}{60.354±95.561} &
  \multicolumn{1}{c|}{608.214±601.782} &
  \multicolumn{1}{c|}{129.390±151.751} &
  \multicolumn{1}{c|}{100.391±148.730} &
  1016.165±708.623 &
  \multicolumn{1}{c|}{271} \\ \cline{3-14} 
\multicolumn{1}{|c|}{} &
  \multicolumn{1}{c|}{} &
  \textbf{Acc. 2} &
  \multicolumn{1}{c|}{740.309±1404.430} &
  \multicolumn{1}{c|}{169.091±297.248} &
  \multicolumn{1}{c|}{160.851±281.730} &
  \multicolumn{1}{c|}{200.750±365.739} &
  \multicolumn{1}{c|}{250.574±463.267} &
  \multicolumn{1}{c|}{218.950±393.527} &
  \multicolumn{1}{c|}{1250.717±2319.080} &
  \multicolumn{1}{c|}{399.836±654.323} &
  \multicolumn{1}{c|}{315.319±579.031} &
  N.A. &
  \multicolumn{1}{c|}{554} \\ \cline{3-14} 
\multicolumn{1}{|c|}{} &
  \multicolumn{1}{c|}{\multirow{-3}{*}{\textbf{\begin{tabular}[c]{@{}c@{}}Case 2:\\ $\mbox{\boldmath{$\tau$}}=[0.7,\;0.6,\;0.1,\;0.9]$\end{tabular}}}} &
  \textbf{Acc. 3} &
  \multicolumn{1}{c|}{536.243±1081.236} &
  \multicolumn{1}{c|}{143.924±259.357} &
  \multicolumn{1}{c|}{138.857±248.660} &
  \multicolumn{1}{c|}{172.411±319.447} &
  \multicolumn{1}{c|}{205.836±391.233} &
  \multicolumn{1}{c|}{186.093±341.124} &
  \multicolumn{1}{c|}{1067.364±2127.014} &
  \multicolumn{1}{c|}{320.217±557.498} &
  \multicolumn{1}{c|}{251.284±480.122} &
  N.A. &
  \multicolumn{1}{c|}{539} \\ \cline{2-14} 
\multicolumn{1}{|c|}{} &
  \multicolumn{1}{c|}{} &
  \textbf{Acc. 1} &
  \multicolumn{1}{c|}{868.439±1691.181} &
  \multicolumn{1}{c|}{184.795±342.450} &
  \multicolumn{1}{c|}{171.494±309.827} &
  \multicolumn{1}{c|}{209.750±375.201} &
  \multicolumn{1}{c|}{323.966±630.425} &
  \multicolumn{1}{c|}{221.918±384.320} &
  \multicolumn{1}{c|}{1099.430±2077.154} &
  \multicolumn{1}{c|}{415.295±738.724} &
  \multicolumn{1}{c|}{400.250±762.156} &
  1270.699±1966.408 &
  \multicolumn{1}{c|}{480} \\ \cline{3-14} 
\multicolumn{1}{|c|}{} &
  \multicolumn{1}{c|}{} &
  \textbf{Acc. 2} &
  \multicolumn{1}{c|}{364.235±465.332} &
  \multicolumn{1}{c|}{117.476±161.184} &
  \multicolumn{1}{c|}{110.280±160.341} &
  \multicolumn{1}{c|}{138.691±205.630} &
  \multicolumn{1}{c|}{178.083±263.959} &
  \multicolumn{1}{c|}{128.789±195.621} &
  \multicolumn{1}{c|}{500.560±580.600} &
  \multicolumn{1}{c|}{187.679±239.945} &
  \multicolumn{1}{c|}{178.710±263.411} &
  581.971±478.532 &
  \multicolumn{1}{c|}{84} \\ \cline{3-14} 
\multicolumn{1}{|c|}{\multirow{-9}{*}{\begin{tabular}[c]{@{}c@{}}\rotatebox[origin=c]{90}{$\mathbf{\color{gray}AAMEJ(25)}$}\end{tabular}}} &
  \multicolumn{1}{c|}{\multirow{-3}{*}{\textbf{\begin{tabular}[c]{@{}c@{}}Case 3:\\ $\mbox{\boldmath{$\tau$}}=[0.7,\;0.8,\;0.4,\;0.5]$\end{tabular}}}} &
  \textbf{Acc. 3} &
  \multicolumn{1}{c|}{321.795±468.590} &
  \multicolumn{1}{c|}{113.526±159.573} &
  \multicolumn{1}{c|}{107.020±158.098} &
  \multicolumn{1}{c|}{135.135±202.107} &
  \multicolumn{1}{c|}{173.765±264.545} &
  \multicolumn{1}{c|}{125.845±192.457} &
  \multicolumn{1}{c|}{364.684±497.318} &
  \multicolumn{1}{c|}{166.711±234.370} &
  \multicolumn{1}{c|}{169.202±263.134} &
  432.112±400.087 &
  \multicolumn{1}{c|}{100} \\ \hline    
\multicolumn{3}{|c|}{\textit{\begin{tabular}[c]{@{}c@{}}\textit{\textbf{{(\color{magenta}PM}±{\color{teal}PSD})}} of \textbf{\color{gray}AAMEJ} \end{tabular}}} &
  \multicolumn{1}{c|}{547.489±1117.015} &
  \multicolumn{1}{c|}{118.163±227.794} &
  \multicolumn{1}{c|}{111.022±212.411} &
  \multicolumn{1}{c|}{139.062±269.764} &
  \multicolumn{1}{c|}{188.857±383.313} &
  \multicolumn{1}{c|}{146.946±282.468} &
  \multicolumn{1}{c|}{871.401±1698.033} &
  \multicolumn{1}{c|}{269.502±491.390} &
  \multicolumn{1}{c|}{229.049±464.251} &
  \multicolumn{1}{c|}{N.A.} &
  \multicolumn{1}{c|}{3,250} 
  \\ \hline
  \multicolumn{3}{|c|}{\textit{\begin{tabular}[c]{@{}c@{}}\textbf{\color{violet}CV} \textbf{= {\color{teal}PSD}/{\color{magenta}PM}}\end{tabular}}} &
  \multicolumn{1}{c|}{2.040} &
  \multicolumn{1}{c|}{1.928} &
  \multicolumn{1}{c|}{1.913} &
  \multicolumn{1}{c|}{1.940} &
  \multicolumn{1}{c|}{2.030} &
  \multicolumn{1}{c|}{1.922} &
  \multicolumn{1}{c|}{1.949} &
  \multicolumn{1}{c|}{1.823} &
  \multicolumn{1}{c|}{2.027} &
  N.A. &
  \multicolumn{1}{c}{} \\ \cline{1-13}
\multicolumn{3}{|c|}{\textit{\textbf{\color{orange}PI} = \textbf{\color{magenta}PM}+\textbf{\color{violet}CV}}} &
  \multicolumn{1}{c|}{549.529} &
  \multicolumn{1}{c|}{120.091} &
  \multicolumn{1}{c|}{112.935} &
  \multicolumn{1}{c|}{141.002} &
  \multicolumn{1}{c|}{190.887} &
  \multicolumn{1}{c|}{148.868} &
  \multicolumn{1}{c|}{873.350} &
  \multicolumn{1}{c|}{271.325} &
  \multicolumn{1}{c|}{231.076} &
  N.A. & \multicolumn{1}{c}{}
  \\\customCline{1.7pt}{1}{13}
  \multicolumn{3}{|c|}{\textit{Avergae \textbf{\color{orange}PI} of  \textbf{\color{cyan}AAMEA} \& \textbf{\color{gray}AAMEJ} }} &
  \multicolumn{1}{c|}{309.898} &
  \multicolumn{1}{c|}{63.059} &
  \multicolumn{1}{c|}{58.970} &
  \multicolumn{1}{c|}{73.940} &
  \multicolumn{1}{c|}{101.633} &
  \multicolumn{1}{c|}{78.344} &
  \multicolumn{1}{c|}{657.430} &
  \multicolumn{1}{c|}{153.661} &
  \multicolumn{1}{c|}{125.742} &
  N.A. 
   &
  \multicolumn{1}{c}{} 
  \\
 \cline{1-13}
\multicolumn{3}{|c|}{{\color[HTML]{FF0000} \begin{tabular}{c}
\textit{\textbf{\underline{Joint Robustness \& Comfort Rank}}}  \\ {\color{black}\textit{based on Average \textbf{\color{orange}PI}  (Asce.)}}
\end{tabular}}} &
  \multicolumn{1}{c|}{8} &
  \multicolumn{1}{c|}{2} &
  \multicolumn{1}{c|}{1} &
  \multicolumn{1}{c|}{3} &
  \multicolumn{1}{c|}{5} &
  \multicolumn{1}{c|}{4} &
  \multicolumn{1}{c|}{9} &
  \multicolumn{1}{c|}{7} &
  \multicolumn{1}{c|}{6} &
  10 &
  \multicolumn{1}{c}{} \\ \cline{1-13}
\end{tabular}
}
\label{comfort_table}
\end{table*}
\setlength{\tabcolsep}{5pt}
\begin{table*}[]
\centering
\caption{Values of AAMEEI (see \eqref{AAIE}) metric for different scenarios under SCGVs (see \ref{SCGVs} and Fig. \ref{commonCGs}) and different common RCTs. The RCTs are illustrated in Figs. \ref{typicalu}-\ref{typicalb}. Also, Acc. 1, Acc. 2, and Acc. 3; leader vehicle's accelerations, are depicted in Fig. \ref{leader_av} where their functions are provided in \eqref{accvel}. The term 'Rank (Asce.)' implies that the communications are ranked, based on API, in an ascending order. The number of SCGVs; $\#\mathbf{K}_{s}$ (see Fig. \ref{commonCGs}), are referred to as Weight in the table.}
\renewcommand{\arraystretch}{1.4}
\resizebox{1\textwidth}{!}{%
\begin{tabular}{ccc|cccccccccc|c}
\cline{4-14}
\multicolumn{1}{c}{} &
  \multicolumn{1}{c}{} &
  \multicolumn{1}{c|}{} &
  \multicolumn{10}{c|}{\textit{Rigid Communication Topologies (RCTs)}} &
  \multicolumn{1}{c|}{} \\ \cline{4-13}
\multicolumn{1}{c}{\multirow{-2}{*}{}} &
  \multicolumn{1}{c}{\multirow{-2}{*}{}} &
  \multicolumn{1}{c|}{\multirow{-2}{*}{}} &
  \multicolumn{1}{c|}{\textbf{PF}} &
  \multicolumn{1}{c|}{\textbf{MPF}} &
  \multicolumn{1}{c|}{\textbf{TPFL}} &
  \multicolumn{1}{c|}{\textbf{PFL}} &
  \multicolumn{1}{c|}{\textbf{TPF}} &
  \multicolumn{1}{c|}{\textbf{BDL}} &
  \multicolumn{1}{c|}{\textbf{BD}} &
  \multicolumn{1}{c|}{\textbf{TBPF}} &
  \multicolumn{1}{c|}{\textbf{TPSF}} &
  \textbf{SPTF} &
  \multicolumn{1}{c|}{\multirow{-2}{*}{Weight}} \\ \Xhline{2pt}
\multicolumn{1}{|c|}{} &
  \multicolumn{1}{c|}{} &
  \textbf{Acc. 1} &
  \multicolumn{1}{c|}{124.736±153.624} &
  \multicolumn{1}{c|}{16.950±21.457} &
  \multicolumn{1}{c|}{14.450±20.519} &
  \multicolumn{1}{c|}{18.045±28.144} &
  \multicolumn{1}{c|}{31.370±46.284} &
  \multicolumn{1}{c|}{16.750±25.999} &
  \multicolumn{1}{c|}{262.150±238.954} &
  \multicolumn{1}{c|}{42.421±46.852} &
  \multicolumn{1}{c|}{35.118±48.221} &
  525.807±356.943 &
  \multicolumn{1}{c|}{245} \\ \cline{3-14} 
\multicolumn{1}{|c|}{} &
  \multicolumn{1}{c|}{} &
  \textbf{Acc. 2} &
  \multicolumn{1}{c|}{255.028±465.361} &
  \multicolumn{1}{c|}{36.343±62.724} &
  \multicolumn{1}{c|}{32.686±55.081} &
  \multicolumn{1}{c|}{42.853±73.635} &
  \multicolumn{1}{c|}{65.868±116.821} &
  \multicolumn{1}{c|}{44.651±74.563} &
  \multicolumn{1}{c|}{507.182±790.517} &
  \multicolumn{1}{c|}{99.917±148.184} &
  \multicolumn{1}{c|}{81.484±139.325} &
  N.A. &
  \multicolumn{1}{c|}{497} \\ \cline{3-14} 
\multicolumn{1}{|c|}{} &
  \multicolumn{1}{c|}{\multirow{-3}{*}{\textbf{\begin{tabular}[c]{@{}c@{}}Case 1:\\ $\mbox{\boldmath{$\tau$}}=[1,\;1,\;1,\;1]$\end{tabular}}}} &
  \textbf{Acc. 3} &
  \multicolumn{1}{c|}{169.510±338.688} &
  \multicolumn{1}{c|}{29.345±53.279} &
  \multicolumn{1}{c|}{27.334±48.481} &
  \multicolumn{1}{c|}{36.278±64.580} &
  \multicolumn{1}{c|}{50.708±95.863} &
  \multicolumn{1}{c|}{37.755±65.336} &
  \multicolumn{1}{c|}{447.671±756.506} &
  \multicolumn{1}{c|}{79.752±129.799} &
  \multicolumn{1}{c|}{62.438±114.530} &
  N.A. &
  \multicolumn{1}{c|}{480} \\ \cline{2-14} 
\multicolumn{1}{|c|}{} &
  \multicolumn{1}{c|}{} &
  \textbf{Acc. 1} &
  \multicolumn{1}{c|}{116.478±147.635} &
  \multicolumn{1}{c|}{16.193±23.100} &
  \multicolumn{1}{c|}{14.275±22.672} &
  \multicolumn{1}{c|}{18.886±31.761} &
  \multicolumn{1}{c|}{30.859±48.219} &
  \multicolumn{1}{c|}{17.912±28.708} &
  \multicolumn{1}{c|}{369.222±326.517} &
  \multicolumn{1}{c|}{47.923±50.841} &
  \multicolumn{1}{c|}{35.967±50.085} &
  1095.558±698.046 &
  \multicolumn{1}{c|}{271} \\ \cline{3-14} 
\multicolumn{1}{|c|}{} &
  \multicolumn{1}{c|}{} &
  \textbf{Acc. 2} &
  \multicolumn{1}{c|}{231.312±458.950} &
  \multicolumn{1}{c|}{37.503±69.433} &
  \multicolumn{1}{c|}{34.741±63.445} &
  \multicolumn{1}{c|}{49.088±94.026} &
  \multicolumn{1}{c|}{67.642±130.731} &
  \multicolumn{1}{c|}{51.925±96.958} &
  \multicolumn{1}{c|}{523.832±892.434} &
  \multicolumn{1}{c|}{106.757±175.901} &
  \multicolumn{1}{c|}{85.197±161.616} &
  N.A. &
  \multicolumn{1}{c|}{554} \\ \cline{3-14} 
\multicolumn{1}{|c|}{} &
  \multicolumn{1}{c|}{\multirow{-3}{*}{\textbf{\begin{tabular}[c]{@{}c@{}}Case 2:\\ $\mbox{\boldmath{$\tau$}}=[0.7,\;0.6,\;1,\;0.9]$\end{tabular}}}} &
  \textbf{Acc. 3} &
  \multicolumn{1}{c|}{165.282±348.684} &
  \multicolumn{1}{c|}{31.542±59.960} &
  \multicolumn{1}{c|}{29.815±55.778} &
  \multicolumn{1}{c|}{41.898±81.549} &
  \multicolumn{1}{c|}{54.394±108.900} &
  \multicolumn{1}{c|}{44.011±83.804} &
  \multicolumn{1}{c|}{454.433±828.795} &
  \multicolumn{1}{c|}{85.704±150.494} &
  \multicolumn{1}{c|}{67.365±133.404} &
  N.A. &
  \multicolumn{1}{c|}{539} \\ \cline{2-14} 
\multicolumn{1}{|c|}{} &
  \multicolumn{1}{c|}{} &
  \textbf{Acc. 1} &
  \multicolumn{1}{c|}{147.032±258.603} &
  \multicolumn{1}{c|}{32.407±58.419} &
  \multicolumn{1}{c|}{29.940±53.900} &
  \multicolumn{1}{c|}{35.769±63.946} &
  \multicolumn{1}{c|}{50.768±94.330} &
  \multicolumn{1}{c|}{36.727±64.752} &
  \multicolumn{1}{c|}{308.808±527.606} &
  \multicolumn{1}{c|}{73.157±124.547} &
  \multicolumn{1}{c|}{61.510±111.562} &
  681.158±982.504 &
  \multicolumn{1}{c|}{480} \\ \cline{3-14} 
\multicolumn{1}{|c|}{} &
  \multicolumn{1}{c|}{} &
  \textbf{Acc. 2} &
  \multicolumn{1}{c|}{53.170±57.984} &
  \multicolumn{1}{c|}{11.087±13.036} &
  \multicolumn{1}{c|}{9.886±12.465} &
  \multicolumn{1}{c|}{12.413±16.480} &
  \multicolumn{1}{c|}{17.700±23.217} &
  \multicolumn{1}{c|}{12.650±16.934} &
  \multicolumn{1}{c|}{150.012±139.222} &
  \multicolumn{1}{c|}{28.170±28.581} &
  \multicolumn{1}{c|}{21.441±26.953} &
  400.183±222.225 &
  \multicolumn{1}{c|}{84} \\ \cline{3-14} 
\multicolumn{1}{|c|}{\multirow{-9}{*}{\begin{tabular}[c]{@{}c@{}}\rotatebox[origin=c]{90}{$\mathbf{\color{brown}AAMEEI(25)}$} \end{tabular}}} &
  \multicolumn{1}{c|}{\multirow{-3}{*}{\textbf{\begin{tabular}[c]{@{}c@{}}Case 3:\\ $\mbox{\boldmath{$\tau$}}=[0.7,\;0.8,\;0.4,\;0.5]$\end{tabular}}}} &
  \textbf{Acc. 3} &
  \multicolumn{1}{c|}{40.401±53.889} &
  \multicolumn{1}{c|}{10.136±13.067} &
  \multicolumn{1}{c|}{9.208±12.460} &
  \multicolumn{1}{c|}{11.783±16.416} &
  \multicolumn{1}{c|}{16.006±22.876} &
  \multicolumn{1}{c|}{12.080±16.847} &
  \multicolumn{1}{c|}{104.792±107.265} &
  \multicolumn{1}{c|}{22.457±26.631} &
  \multicolumn{1}{c|}{18.515±26.099} &
  320.089±179.733 &
  \multicolumn{1}{c|}{100} \\ \hline
\multicolumn{3}{|c|}{\begin{tabular}[c]{@{}c@{}}\textit{\textbf{{(\color{magenta}PM}±{\color{teal}PSD})}} \textit{of \textbf{\color{brown}AAMEEI} } \end{tabular}} &
  \multicolumn{1}{c|}{174.325±346.221} &
  \multicolumn{1}{c|}{29.528±55.083} &
  \multicolumn{1}{c|}{27.143±50.275} &
  \multicolumn{1}{c|}{36.129±69.439} &
  \multicolumn{1}{c|}{51.499±100.156} &
  \multicolumn{1}{c|}{37.434±70.759} &
  \multicolumn{1}{c|}{411.596±696.639} &
  \multicolumn{1}{c|}{78.888±132.685} &
  \multicolumn{1}{c|}{63.232±120.761} &
  N.A. &
  \multicolumn{1}{c|}{3,250} \\ \Xhline{2pt}
\multicolumn{3}{|c|}{\textit{\begin{tabular}[c]{@{}c@{}}\textbf{\color{violet}CV} \textbf{= {\color{teal}PSD}/{\color{magenta}PM}}\end{tabular}}} &
  \multicolumn{1}{c|}{1.986} &
  \multicolumn{1}{c|}{1.865} &
  \multicolumn{1}{c|}{1.852} &
  \multicolumn{1}{c|}{1.922} &
  \multicolumn{1}{c|}{1.945} &
  \multicolumn{1}{c|}{1.890} &
  \multicolumn{1}{c|}{1.693} &
  \multicolumn{1}{c|}{1.682} &
  \multicolumn{1}{c|}{1.910} &
  N.A. &
  \multicolumn{1}{c}{} \\ \cline{1-13}
\multicolumn{3}{|c|}{\textit{\textbf{\color{orange}PI} = \textbf{\color{magenta}PM}+\textbf{\color{violet}CV}}} &
  \multicolumn{1}{c|}{176.311} &
  \multicolumn{1}{c|}{31.393} &
  \multicolumn{1}{c|}{28.995} &
  \multicolumn{1}{c|}{38.051} &
  \multicolumn{1}{c|}{53.444} &
  \multicolumn{1}{c|}{39.324} &
  \multicolumn{1}{c|}{413.289} &
  \multicolumn{1}{c|}{80.570} &
  \multicolumn{1}{c|}{65.142} &
  N.A. &
  \multicolumn{1}{c}{} \\ \cline{1-13}
\multicolumn{3}{|c|}{{\color[HTML]{FF0000} \begin{tabular}{c}
\textit{\textbf{\underline{Joint Robustness \& Energy Rank}}}  \\ {\color{black}\textit{based on \textbf{\color{orange}PI}  (Asce.)}}
\end{tabular}}} &
  \multicolumn{1}{c|}{8} &
  \multicolumn{1}{c|}{2} &
  \multicolumn{1}{c|}{1} &
  \multicolumn{1}{c|}{3} &
  \multicolumn{1}{c|}{5} &
  \multicolumn{1}{c|}{4} &
  \multicolumn{1}{c|}{9} &
  \multicolumn{1}{c|}{7} &
  \multicolumn{1}{c|}{6} &
  10 &
  \multicolumn{1}{c}{} \\ \cline{1-13}
\end{tabular}
}
\label{energy_table}
\end{table*}

4. \textbf{Results for Passenger Comfort}:
For this metric, we have defined AAMEA [see \eqref{AAAE}] and AAMEJ [see \eqref{AAJE}]. Using SCGVs [see \ref{SCGVs} and Fig. \ref{commonCGs}], the results for the AAMEA and AAMEJ metrics for each scenario are shown in Table \ref{comfort_table}. For example, under Case 1, Acc. 3, and the PF topology, the AAMEA and AAMEJ values are $65.219\pm 100.989$ and $376.922\pm 768.717$, respectively. In Table \ref{comfort_table}, the Performance Index (PI) for the AAMEA and AAMEJ metrics is calculated, and the Average PI (API) for the Passenger Comfort is determined. Based on the API, the rank of each RCT is provided in the last row of Table \ref{comfort_table}. This ranking compares the performance of the platoon under various RCTs in terms of passenger comfort and robustness for the platoon. We observe that TPFL, MPF, PFL, BDL, and TPF are the top 5 topologies providing greater comfort for the passengers, while SPTF, BD, PF, TBPF, and TPSF are the top 5 topologies providing less comfort. As an example, Fig. \ref{metrics_plt_ncumsum} shows the AMEA and AMEJ trajectories for Case 1 with Acc. 1. Also, Fig. \ref{metrics_plt_platoon} shows the corresponding AAMEA and AAMEJ trajectories.
\begin{figure}[htbp!]
\centering
\subfloat[APMTTC (see \eqref{APMTTC}), AMDRAC (see \eqref{ADRAC}), AMEEI (see \eqref{enCT}), AMEA (see \eqref{enJA1}), and AMEJ (see \eqref{enJA2}) metrics for Case 1 with Acc. 1. \label{metrics_plt_ncumsum}]{\includegraphics[width=1\linewidth]{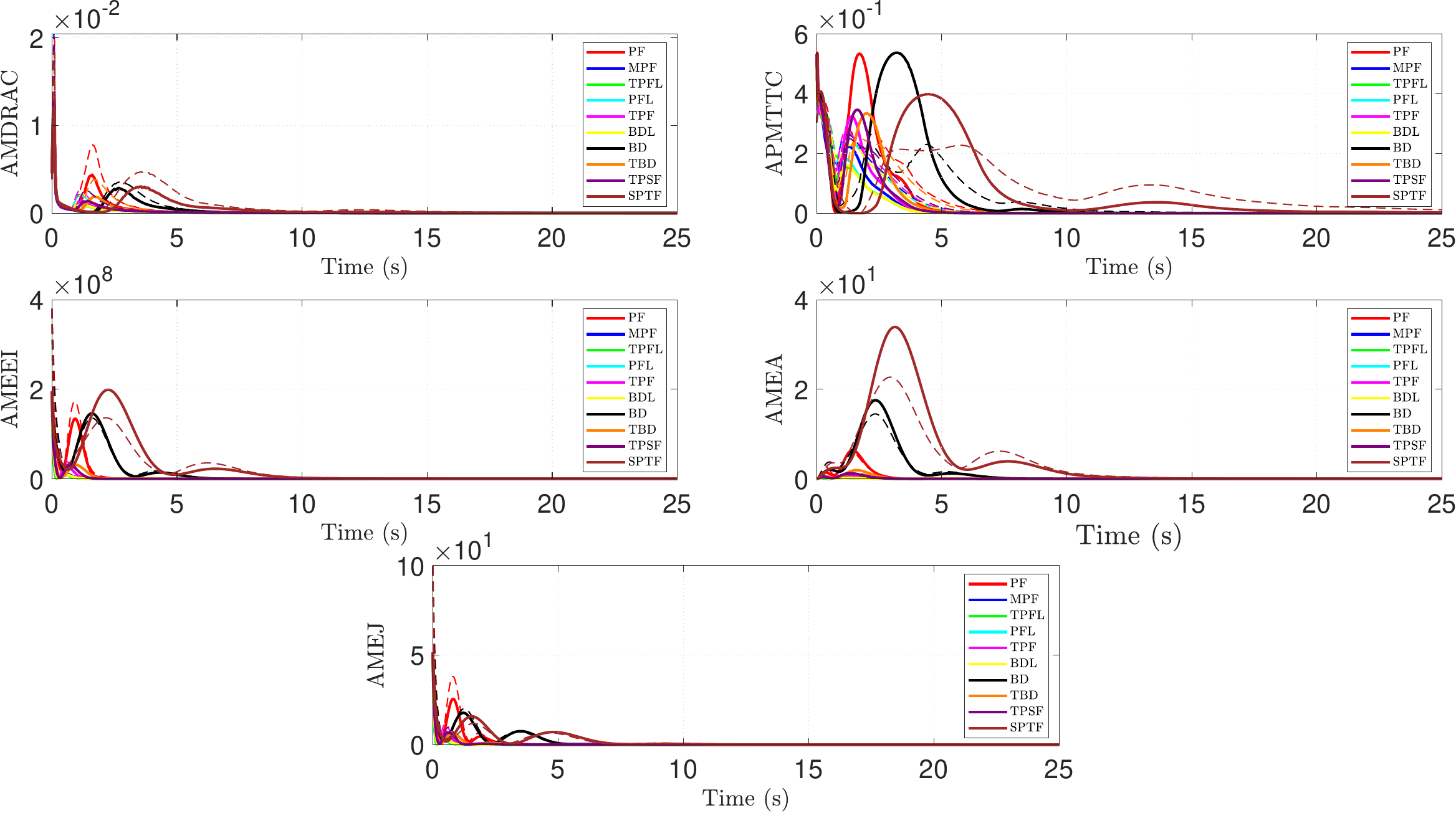}}
\hfill
\subfloat[AAPMTTC (see \eqref{AAPMTTC}), AAMDRAC (see \eqref{AADRAC}), AAMEEI (see \eqref{AAIE}), AAMEA (see \eqref{AAAE}), and AAMEJ (see \eqref{AAJE}) metrics for Case 1 with Acc. 1.\label{metrics_plt_platoon}]{\includegraphics[width=1\linewidth]{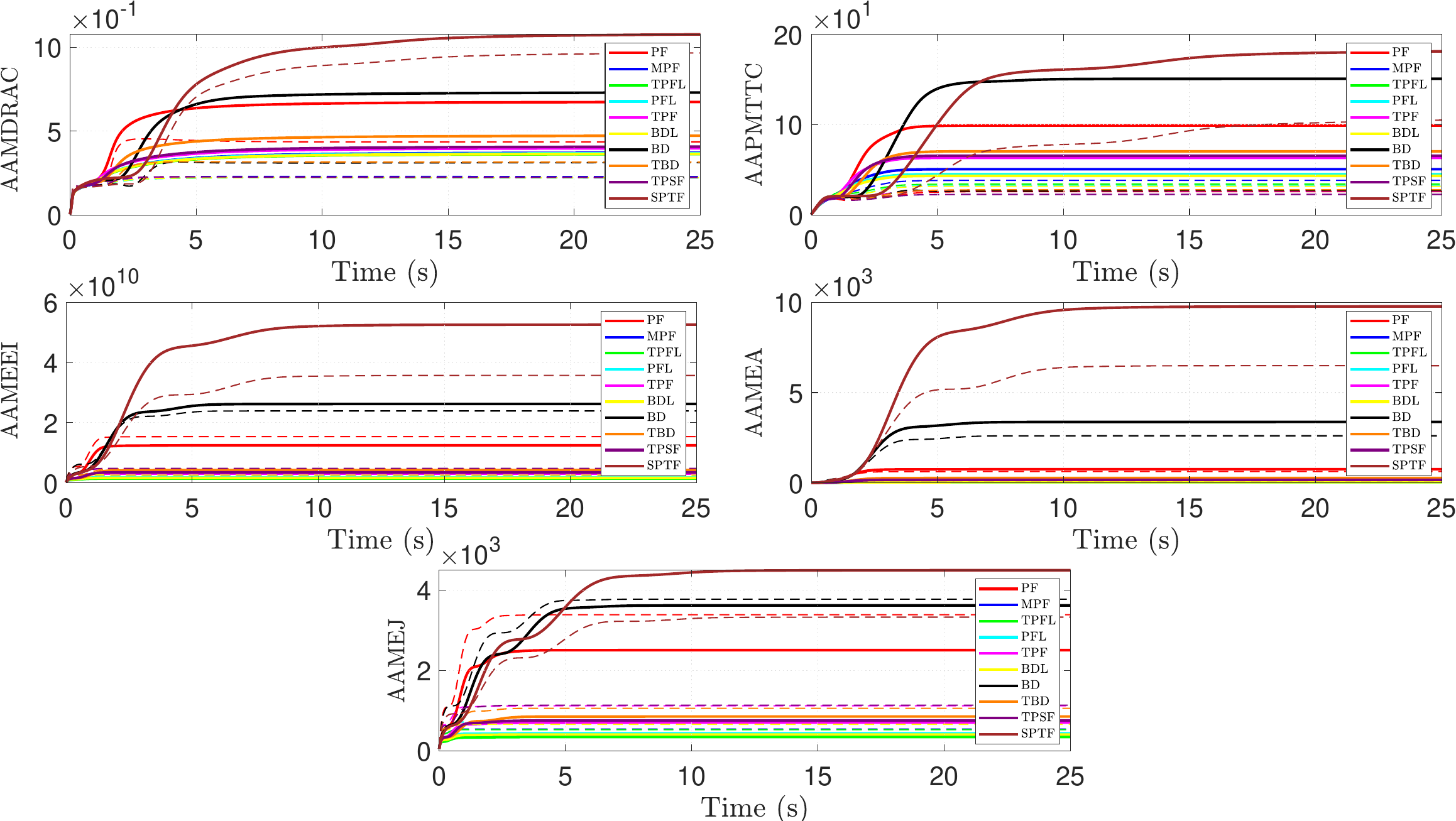}}
\caption{Platoon momentary and accumulated performance metrics over travel time under common RCTs, and under, as a sample, Case 1 \& Acc. 1 (see Table \ref{table1} and Fig. \ref{leader_av}). The dashed lines are standard deviations.}
\label{MEtrics}
\end{figure}
\subsection{\textbf{Overall Ranking}}
To integrate all metrics, we first normalize the PI values of the Energy Consumption  [see Table \ref{energy_table}] and SaCGDI  [see Table \ref{area_percentage}] metrics, and the API values of  Intervehicle Distance Safety [see Table \ref{safety_table}] and Passenger Comfort [see Table \ref{comfort_table}] metrics between 0 and 1. These normalized values are shown in Table \ref{overall_table}. Based on ANV (Average Normalized values) for each RCT, we obtain the overall ranking of the RCTs. We observe that the communications, from best to worst, in terms of their performance, are as follows: TPFL, MPF, PFL, BDL, TPF, TPSF, TBPF, PF, BD, and SPTF.     
\subsection{\textbf{Beneficial and Adversarial Features}}
The following sections elaborate on communication features that can positively or adversely impact platoon performance:
\setlength{\tabcolsep}{2.8pt}
\begin{table}[]
\centering
\caption{Normalized values (between $0$ and $1$) and overall ranking of common RCTs. For consistency with other metrics, for SaCGDI, normalization is obtained without considering SPTF Topology and so replaced with N.A. in the table.}
\renewcommand{\arraystretch}{1.3}
\resizebox{0.5\textwidth}{!}{%
\begin{tabular}{ccc|cccccccccc|}
\cline{4-13}
\multicolumn{1}{c}{} &
   &
   &
  \multicolumn{10}{c|}{\textit{Rigid Communication Topologies (RCTs)}}  \\ \cline{4-13}
\multicolumn{1}{c}{} &
   &
   &
  \multicolumn{1}{c|}{\textbf{PF}} &
  \multicolumn{1}{c|}{\textbf{MPF}} &
  \multicolumn{1}{c|}{\textbf{TPFL}} &
  \multicolumn{1}{c|}{\textbf{PFL}} &
  \multicolumn{1}{c|}{\textbf{TPF}} &
  \multicolumn{1}{c|}{\textbf{BDL}} &
  \multicolumn{1}{c|}{\textbf{BD}} &
  \multicolumn{1}{c|}{\textbf{TBPF}} &
  \multicolumn{1}{c|}{\textbf{TPSF}} &
  \textbf{SPTF} \\ \Xhline{2pt}
\multicolumn{3}{|c|}{\begin{tabular}[c]{@{}c@{}}\textit{Normalized  \textbf{\color{orange}PI} of}  \textit{\textbf{\color{bazaar}SaCGDI}} \end{tabular}} &
  \multicolumn{1}{c|}{0.573} &
  \multicolumn{1}{c|}{0.023} &
  \multicolumn{1}{c|}{0.000} &
  \multicolumn{1}{c|}{0.038} &
  \multicolumn{1}{c|}{0.113} &
  \multicolumn{1}{c|}{0.097} &
  \multicolumn{1}{c|}{1.000} &
  \multicolumn{1}{c|}{0.548} &
  \multicolumn{1}{c|}{0.157} &
  N.A. \\ \cline{1-13}
\multicolumn{3}{|c|}{\begin{tabular}[c]{@{}c@{}}\textit{Normalized \textbf{\color{orange}API} of} \\\textit{\textbf{\color{blue}AAPMMTTC}} \& \textit{\textbf{\color{olive}AAMDRAC}}\\\end{tabular}} &
  \multicolumn{1}{c|}{0.378} &
  \multicolumn{1}{c|}{0.010} &
  \multicolumn{1}{c|}{0.000} &
  \multicolumn{1}{c|}{0.010} &
  \multicolumn{1}{c|}{0.057} &
  \multicolumn{1}{c|}{0.009} &
  \multicolumn{1}{c|}{1.000} &
  \multicolumn{1}{c|}{0.175} &
  \multicolumn{1}{c|}{0.088} &
  N.A.  \\ \cline{1-13}
\multicolumn{3}{|c|}{\begin{tabular}[c]{@{}c@{}}\textit{Normalized \textbf{\color{orange}PI} of} \textit{\textbf{\color{brown}AAMEEI}} \end{tabular}} &
  \multicolumn{1}{c|}{0.383} &
  \multicolumn{1}{c|}{0.006} &
  \multicolumn{1}{c|}{0.000} &
  \multicolumn{1}{c|}{0.024} &
  \multicolumn{1}{c|}{0.064} &
  \multicolumn{1}{c|}{0.027} &
  \multicolumn{1}{c|}{1.000} &
  \multicolumn{1}{c|}{0.134} &
  \multicolumn{1}{c|}{0.094} &
  N.A. \\ \cline{1-13}
\multicolumn{3}{|c|}{\begin{tabular}[c]{@{}c@{}}\textit{Normalized \textbf{\color{orange}API} of} \\\textit{\textbf{\color{cyan}AAMEA}} \& \textit{\textbf{\color{gray}AAMEJ}}\end{tabular}}  &
  \multicolumn{1}{c|}{0.419} &
  \multicolumn{1}{c|}{0.007} &
  \multicolumn{1}{c|}{0.000} &
  \multicolumn{1}{c|}{0.025} &
  \multicolumn{1}{c|}{0.071} &
  \multicolumn{1}{c|}{0.032} &
  \multicolumn{1}{c|}{1.000} &
  \multicolumn{1}{c|}{0.158} &
  \multicolumn{1}{c|}{0.112} &
  N.A.  \\ \cline{1-13}
\multicolumn{3}{|c|}{\begin{tabular}[c]{@{}c@{}}\textit{\textbf{\color{orange}ANV} = Average  of Normalized Values}\end{tabular}}&
  \multicolumn{1}{c|}{0.438} &
  \multicolumn{1}{c|}{0.012} &
  \multicolumn{1}{c|}{0.000} &
  \multicolumn{1}{c|}{0.024} &
  \multicolumn{1}{c|}{0.076} &
  \multicolumn{1}{c|}{0.041} &
  \multicolumn{1}{c|}{1.000} &
  \multicolumn{1}{c|}{0.254} &
  \multicolumn{1}{c|}{0.113} &
  N.A. \\ \cline{1-13}
\multicolumn{3}{|c|}{{\color[HTML]{FF0000} \begin{tabular}{c}
\textit{\textbf{\underline{Overall Joint Robustness \& Metrics Rank}}} \\{\color{black}\textit{based on \textbf{\color{orange}ANV}  (Asce.)}}
\end{tabular}}} &
  \multicolumn{1}{c|}{8} &
  \multicolumn{1}{c|}{2} &
  \multicolumn{1}{c|}{1} &
  \multicolumn{1}{c|}{3} &
  \multicolumn{1}{c|}{5} &
  \multicolumn{1}{c|}{4} &
  \multicolumn{1}{c|}{9} &
  \multicolumn{1}{c|}{7} &
  \multicolumn{1}{c|}{6} &
  10\\ \cline{1-13}
\end{tabular}
}
\label{overall_table}
\end{table}

1. \textbf{Broadcasting the Leader's State}:
To highlight this effect, we investigate SaCGDI, Intervehicle Distance Safety (AAPMTTC \& AAMDRAC), Energy Consumption (AAMEEI), and Passenger Comfort (AAMEA \& AAMEJ) metrics, for which the results are obtained under SCGVs [see \ref{SCGVs} and Fig. \ref{commonCGs}]. The highlights (based on comparisons between PIs) are as follows:

\textbf{BD vs BDL}: The BDL topology is similar to BD, except that in BDL, FVs also receive the leader's state. Broadcasting the leader's state to the FVs significantly improved platoon performance. The SaCGDI metric decreased by 45.799\% (see Table \ref{area_percentage}). The AAPMTTC and AAMDRAC metrics decreased by 73.398\% and 49.005\%, respectively, indicating enhanced platoon safety (see Table \ref{safety_table}). The AAMEEI metric dropped by 90.485\%, showing a substantial reduction in the energy  consumption metric (see Table \ref{energy_table}). Lastly, the AAMEA and AAMEJ metrics decreased by 98.228\% and 82.954\%, respectively, demonstrating significant improvements in passenger comfort (see Table \ref{comfort_table}).

\textbf{PF vs PFL}: The PFL topology is similar to PF, except that in PFL, FVs also receive the leader's state. Broadcasting the leader's state to the FVs enhanced the platoon's performance. The SaCGDI metric decreased by 34.641\% (see Table \ref{area_percentage}). The AAPMTTC and AAMDRAC metrics were reduced by 50.377\% and 26.467\%, respectively, indicating improved platoon safety (see Table \ref{safety_table}). The AAMEEI metric dropped by 78.418\%, highlighting a reduction in the energy  consumption metric (see Table \ref{energy_table}). Finally, the AAMEA and AAMEJ metrics decreased by 90.210\% and 74.341\%, respectively, showing improvements in passenger comfort (see Table \ref{comfort_table}).

\textbf{TPF vs TPFL}: The TPFL topology is similar to TPF, except that in TPFL, FVs also receive the leader's state. Broadcasting the leader's state to the FVs improved the platoon's performance. The SaCGDI metric decreased by 10.455\% (see Table \ref{area_percentage}).  The AAPMTTC and AAMDRAC metrics saw reductions of 15.328\% and 3.836\%, respectively, indicating enhanced platoon safety (see Table \ref{safety_table}). The AAMEEI metric dropped by 45.746\%, highlighting a reduction in the energy  consumption metric (see Table \ref{energy_table}). Lastly, the AAMEA and AAMEJ metrics decreased by 59.571\% and 40.836\%, respectively, showing improvements in passenger comfort (see Table \ref{comfort_table}).

2. \textbf{Receiving States of Vehicles Ahead}:
To highlight this effect, we investigate SaCGDI, Intervehicle Distance Safety (AAPMTTC and AAMDRAC), Energy Consumption (AAMEEI), and Passenger Comfort (AAMEA and AAMEJ) metrics, for which the results are obtained under SCGVs [see \ref{SCGVs} and Fig. \ref{commonCGs}]. The highlights (based on comparisons between PIs) are as follows:

\textbf{PF vs TPF \& MPF}: In the PF, TPF, and MPF topologies, FVs receive information from one, two, and three immediate vehicles ahead, respectively. The SaCGDI metric decreased by 29.794\% (PF vs TPF) and 35.629\% (PF vs MPF) (see Table \ref{area_percentage}). The AAPMTTC and AAMDRAC metrics are reduced by 42.899\% and 24.480\%, respectively, when FVs receive data from one more vehicle ahead (PF vs TPF). These reductions increase to 50.014\% and 26.896\% when receiving data from two more vehicles ahead (PF vs MPF) (see Table \ref{safety_table}). The AAMEEI metric drops by 69.887\% (PF vs TPF), increasing to 82.194\% (PF vs MPF) (see Table \ref{energy_table}). Lastly, the AAMEA and AAMEJ metrics for passenger comfort decrease by 82.381\% and 65.263\% (PF vs TPF), increasing to 91.421\% and 78.146\% (PF vs MPF) (see Table \ref{comfort_table}).

\textbf{BD vs TPSF}: The TPSF topology is similar to the BD, except in TPSF, each FV also receives information from the vehicle two positions ahead. The SaCGDI metric decreased by 42.770\% (see Table \ref{area_percentage}). The AAPMTTC and AAMDRAC metrics decrease by 66.416\% and 46.782\%, respectively, when FVs receive data from one more vehicle ahead (see Table \ref{safety_table}). The AAMEEI metric drops by 84.238\% under the same conditions (see Table \ref{energy_table}). Finally, passenger comfort improves, with the AAMEA and AAMEJ metrics reducing by 95.377\% and 73.541\%, respectively (see Table \ref{comfort_table}).
\begin{figure}[htbp!]
\begin{center}
\resizebox{0.95\hsize}{!}{\includegraphics*{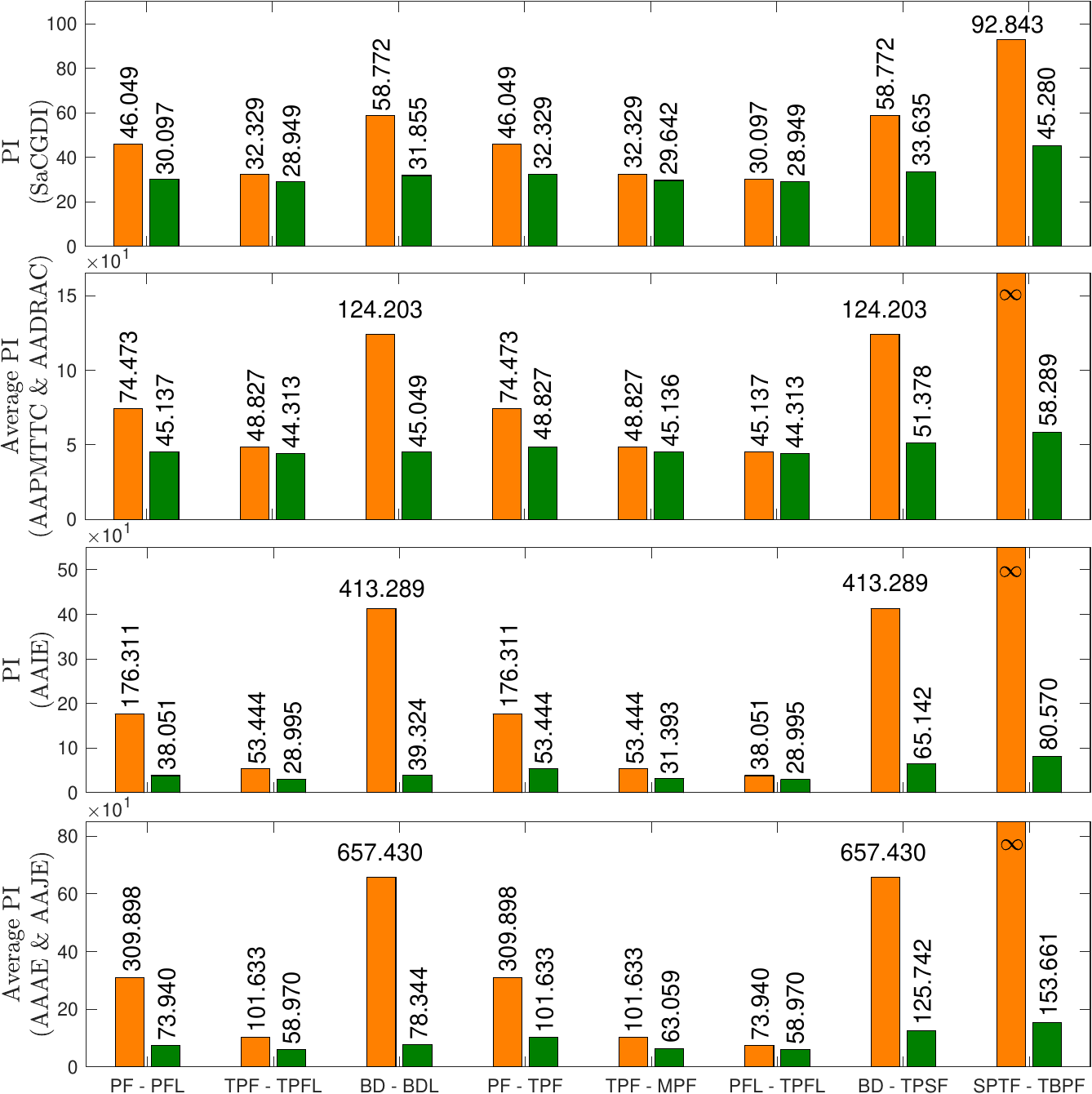}}
\caption{The effect of broadcasting additional information from vehicles ahead to follower vehicles (FVs) on platoon performance is shown. Each comparison (e.g., PF vs. PFL) involves the second topology providing state information from an extra vehicle further ahead. This additional information improves performance, highlighted by green bars. See Figs. \ref{typicalu}-\ref{typicalb} and Tables \ref{area_percentage}, \ref{safety_table}, \ref{energy_table}, and \ref{comfort_table} for detailed values.}
\label{barplots_aheadeffect}
\end{center}
\end{figure}
\begin{remark}    
To illustrate the effectiveness of adding an additional communication avenue for each FV to receive information from more vehicles ahead, Fig. \ref{barplots_aheadeffect} presents comparisons between pairs of communication topologies, such as (PF - PFL) and (TPF - MPF). In each pair, the second communication (PFL and MPF) includes an extra communication link, allowing vehicles to receive data from an additional vehicle ahead, compared to the first one (PF and TPF). We observe that these second communications result in smaller performance metric values, as indicated by the PIs for SaCGDI and AAMEEI, and the average PIs for intervehicle distance safety (AAPMTTC \& AAMDRAC) and passenger comfort (AAMEA \& AAMEJ). The superior communications are shown in green bars.
\end{remark}
3. \textbf{Receiving States of Vehicles Behind}
To highlight this effect, we analyze the metrics SaCGDI, Intervehicle Distance Safety (AAPMTTC \& AAMDRAC), Energy Consumption (AAMEEI), and Passenger Comfort (AAMEA \& AAMEJ). The results are derived considering SCGVs [see \ref{SCGVs} and Fig. \ref{commonCGs}], and the main findings (based on comparisons between PIs) are as follows:

\textbf{PF vs BD}: The BD topology is similar to the PF topology, except in BD, each FV also receives information from its immediate FV. The SaCGDI metric increased by 27.629\% (see Table \ref{area_percentage}). The AAPMTTC and AAMDRAC metrics increase by 86.257\% and 43.862\%, respectively, when FVs receive data from their following vehicle (see Table \ref{safety_table}). The AAMEEI metric for  input energy expands by 134.409\% (see Table \ref{energy_table}). Finally, passenger comfort degrades, with the AAMEA and AAMEJ metrics increasing by 528.340\% and 58.927\%, respectively (see Table \ref{comfort_table}).

\textbf{TPF vs TPSF}: The TPSF topology is similar to TPF, except in TPSF, each FV also receives information from its immediate following vehicle. The SaCGDI metric increased by 4.039\% (see Table \ref{area_percentage}). The AAPMTTC and AAMDRAC metrics increase by 9.548\% and 1.377\% when FVs receive data from the vehicle behind them (see Table \ref{safety_table}). The AAMEEI metric for energy  consumption increases by 21.888\% (see Table \ref{energy_table}). Passenger comfort degrades, with the AAMEA and AAMEJ metrics increasing by 64.846\% and 21.053\%, respectively (see Table \ref{comfort_table}).

\textbf{TPSF vs TBPF}: The TBPF topology is similar to TPSF, except in TBPF, each FV also receives information from its second immediate following vehicle. The SaCGDI metric increased by 34.621\% (see Table \ref{area_percentage}). The AAPMTTC and AAMDRAC metrics increase by 15.638\% and 11.349\%, respectively, when FVs receive data from one more vehicle behind (see Table \ref{safety_table}). The AAMEEI metric for energy  consumption increases by 23.683\% under the same conditions (see Table \ref{energy_table}). Finally, passenger comfort degrades, with the AAMEA and AAMEJ metrics increasing by 76.386\% and 17.418\%, respectively (see Table \ref{comfort_table}).

\textbf{BD vs SPTF}: The SPTF topology is similar to BD, except in SPTF, each FV also receives information from its second immediate following vehicle. This significantly degrades the performance of the platoon. The SaCGDI metric increased by 57.971\% (see Table \ref{area_percentage}). Or, for instance, under Case 3 and Acc. 1, the AAPMTTC and AAMDRAC metrics increase by 49.383\% and 58.148\%, respectively, when FVs receive data from one more vehicle behind (see Table \ref{safety_table}). The AAMEEI metric for energy  consumption increases by 120.576\% (see Table \ref{energy_table}). Finally, passenger comfort degrades, with the AAMEA and AAMEJ metrics increasing by 219.642\% and 15.577\%, respectively (see Table \ref{comfort_table}).
\begin{figure}[htbp!]
\begin{center}
\resizebox{0.95\hsize}{!}{\includegraphics*{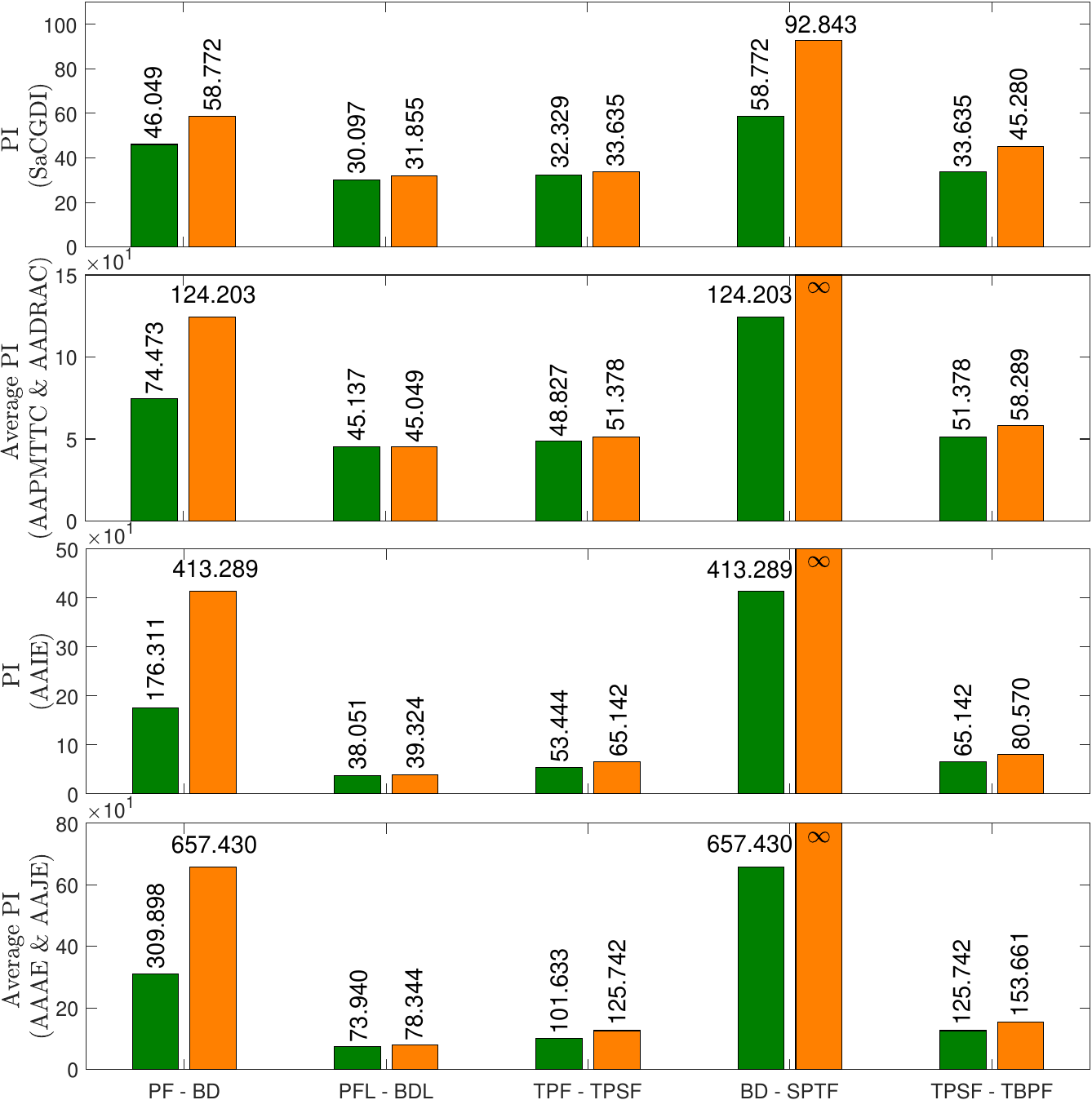}}
\caption{The effect of broadcasting information from additional vehicles behind to follower vehicles (FVs) on platoon performance is illustrated. Each comparison (e.g., PF vs. BD) shows that the second topology includes state information from an extra vehicle further behind, which degrades performance. This is indicated by orange bars. Refer to Figs. \ref{typicalu}-\ref{typicalb} and Tables \ref{area_percentage}, \ref{safety_table}, \ref{energy_table}, and \ref{comfort_table} for detailed values.}
\label{barplots_behindeffect}
\end{center}
\end{figure}
\begin{remark}    
To demonstrate the detrimental impact of adding an additional communication link for each following vehicle (FV) to receive information from more vehicles behind, Fig. \ref{barplots_behindeffect} presents comparisons between pairs of communication topologies, such as (PF - BD) and (TPSF - TBPF). In each pair, the second communication (BD and TBPF) includes an extra link, enabling vehicles to receive data from an additional vehicle behind compared to the first communication (PF and TPSF). We observe that these second communications lead to larger performance metric values, as indicated by the PIs for SaCGDI and AAMEEI, and the average PIs for intervehicle distance safety (AAPMTTC and AAMDRAC) and passenger comfort (AAMEA and AAMEJ). The inferior communications in each pair are represented by the orange bars.
\end{remark}

4. \textbf{A Vehicle Ahead vs. the Leader Vehicle}: 
To highlight this effect, we analyze the metrics SaCGDI, Intervehicle Distance Safety (AAPMTTC and AAMDRAC), Energy Consumption (AAMEEI), and Passenger Comfort (AAMEA and AAMEJ). The main findings (based on comparisons between PIs) are:

\textbf{TPF vs PFL}: In PFL, each FV receives the state of the leader instead of the second immediate preceding vehicle, which improves platoon performance. The SaCGDI metric decreased by 6.904\% (see Table \ref{area_percentage}). The AAPMTTC and AAMDRAC metrics decrease by 13.096\% and 2.631\%, respectively (see Table \ref{safety_table}). The AAMEEI metric for engine input energy decreases by 28.802\% (see Table \ref{energy_table}). Passenger comfort improves, with the AAMEA and AAMEJ metrics decreasing by 44.434\% and 26.133\%, respectively (see Table \ref{comfort_table}).

\textbf{MPF vs TPFL}: In TPFL, each FV receives the state of the leader vehicle instead of the third immediate preceding vehicle, which enhances the platoon performance. The SaCGDI metric decreased by 2.337\% (see Table \ref{area_percentage}). The AAPMTTC and AAMDRAC metrics decrease by 3.275\% and 0.657\% (see Table \ref{safety_table}). The AAMEEI metric for engine input energy decreases by 7.638\% (see Table \ref{energy_table}). Passenger comfort improves, with the AAMEA and AAMEJ metrics decreasing by 16.970\% and 5.958\% (see Table \ref{comfort_table}).

\textbf{TPSF vs BDL}: In BDL, each FV receives the state of the leader vehicle instead of the second immediate preceding vehicle, which improves platoon performance. The SaCGDI metric decreased by 5.292\% (see Table \ref{area_percentage}). The AAPMTTC and AAMDRAC metrics decrease by 20.789\% and 4.177\% (see Table \ref{safety_table}). The AAMEEI metric for energy  consumption decreases by 39.633\% (see Table \ref{energy_table}). Passenger comfort improves, with the AAMEA and AAMEJ metrics decreasing by 61.681\% and 35.576\% (see Table \ref{comfort_table}).
\begin{figure}[htbp!]
\begin{center}
\resizebox{0.85\hsize}{!}{\includegraphics*{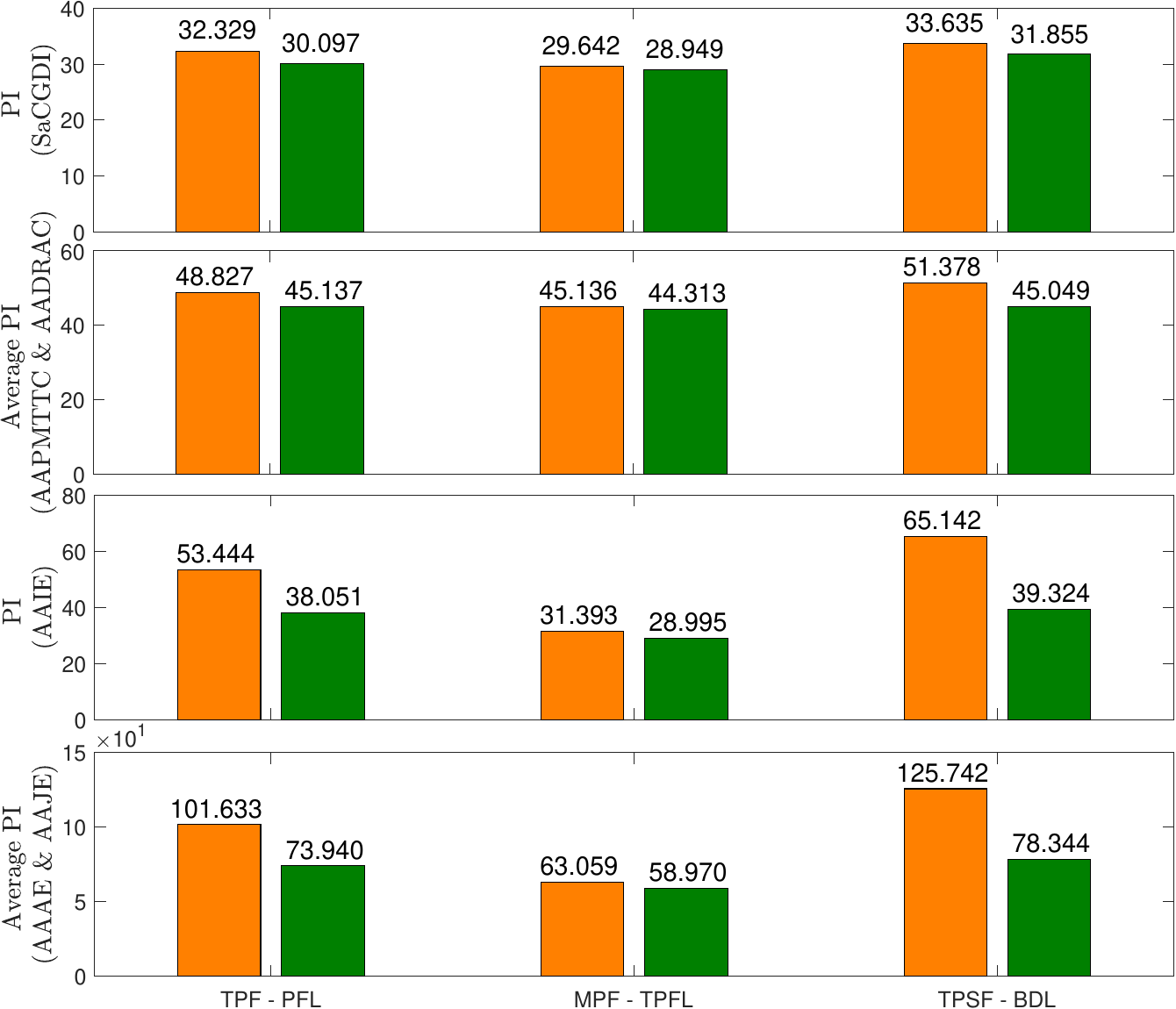}}
\caption{Comparison between broadcasting information from an additional vehicle ahead and broadcasting the leader's state to follower vehicles (FVs) is shown. Each comparison (e.g., TPF vs. PFL) highlights that broadcasting the leader's state improves performance. See Figs. \ref{typicalu}–\ref{typicalb} and Tables \ref{area_percentage}, \ref{safety_table}, \ref{energy_table}, and \ref{comfort_table} for detailed values.}
\label{barplots_leader_vs_ahead}
\end{center}
\end{figure}
\begin{remark}    
Fig. \ref{barplots_leader_vs_ahead} compares communication topologies such as (TPF - PFL) and (MPF - TPFL), where the second topology (PFL and TPFL) replaces the link to a farther vehicle with a direct link to the leader. This change leads to better performance, shown by lower PI values for SaCGDI and AAMEEI, and improved average PIs for safety (AAPMTTC and AAMDRAC) and comfort (AAMEA and AAMEJ).
\end{remark}
\begin{remark}
Since receiving information from vehicles ahead, especially the leader, greatly improves platoon performance, larger platoons will likely benefit most from forward-looking and leader-centric communication topologies.
\end{remark}
\section{\textbf{Conclusion}}
\label{conclusions}
In this paper, we investigated the impact of rigid communication topologies (RCTs) on vehicular platoon performance, focusing on four key metrics: Safe Control Gain Deficiency Index (SaCGDI), Intervehicle Distance Safety, Energy Consumption, and Passenger Comfort. 

We developed platoon dynamics based on distance errors, relative velocities, and accelerations between neighboring vehicles, considering the leader’s acceleration trajectory and initial conditions. Our results show that broadcasting the leader’s state to follower vehicles significantly improves safety, energy efficiency, and passenger comfort. Additionally, receiving information from more vehicles ahead enhances performance, while receiving data from vehicles behind degrades it. 

The RCTs were ranked from best to worst as: TPFL, MPF, PFL, BDL, TPF, TPSF, TBPF, PF, BD, and SPTF. Communication from vehicles ahead enables proactive adjustments, improving performance, whereas communication from vehicles behind leads to poorer outcomes. These findings highlight the critical role of communication topologies in improving platoon efficiency and safety. Future work will explore the impact of communication delays and identify topologies that are more robust against such delays.
\label{section6}
\bibliographystyle{elsarticle-num}        
\bibliography{name,IEEEabrv}           

\begin{thebibliography}{10}
\expandafter\ifx\csname url\endcsname\relax
  \def\url#1{\texttt{#1}}\fi
\expandafter\ifx\csname urlprefix\endcsname\relax\def\urlprefix{URL }\fi
\expandafter\ifx\csname href\endcsname\relax
  \def\href#1#2{#2} \def\path#1{#1}\fi

\bibitem{alam2013cooperative}
N.~Alam, et~al., Cooperative positioning for vehicular networks: Facts and future, IEEE Trans. Intell. Transp. Syst. 14~(4) (2013) 1708--1717.

\bibitem{xu2014impact}
L.~Xu, et~al., Impact of communication erasure channels on the safety of highway vehicle platoons, IEEE Trans. Intell. Transp. Syst. 16~(3) (2014) 1456--1468.

\bibitem{balador2022survey}
A.~Balador, et~al., A survey on vehicular communication for cooperative truck platooning application, Vehicular Communications 35 (2022) 100460.

\bibitem{xu2014communication}
L.~Xu, et~al., Communication information structures and contents for enhanced safety of highway vehicle platoons, IEEE Trans. Veh. Technol. 63~(9) (2014) 4206--4220.

\bibitem{taylor2022vehicular}
S.~J. Taylor, et~al., Vehicular platoon communication: architecture, security threats and open challenges, Sensors 23~(1) (2022) 134.

\bibitem{razzaghpour2022impact}
M.~Razzaghpour, et~al., Impact of information flow topology on safety of tightly-coupled connected and automated vehicle platoons utilizing stochastic control, in: 2022 European Control Conference (ECC), IEEE, 2022, pp. 27--33.

\bibitem{amoozadeh2015platoon}
M.~Amoozadeh, et~al., Platoon management with cooperative adaptive cruise control enabled by vanet, Vehicular communications 2~(2) (2015) 110--123.

\bibitem{razzaghpour2021impact}
M.~Razzaghpour, et~al., Impact of communication loss on mpc based cooperative adaptive cruise control and platooning, in: 2021 IEEE 94th Vehicular Technology Conference (VTC2021-Fall), IEEE, 2021, pp. 01--07.

\bibitem{sidorenko2021safety}
G.~Sidorenko, et~al., Safety of automatic emergency braking in platooning, IEEE Trans. Veh. Technol. 71~(3) (2021) 2319--2332.

\bibitem{prayitno2021v2v}
A.~Prayitno, et~al., V2v network topologies for vehicle platoons with cooperative state variable feedback control, in: 2021 Second International Symposium on Instrumentation, Control, Artificial Intelligence, and Robotics (ICA-SYMP), IEEE, 2021, pp. 1--4.

\bibitem{milanes2013cooperative}
V.~Milan{\'e}s, et~al., Cooperative adaptive cruise control in real traffic situations, IEEE Trans. Intell. Transp. Syst. 15~(1) (2013) 296--305.

\bibitem{milanes2014modeling}
V.~Milan\'es, et~al., Modeling cooperative and autonomous adaptive cruise control dynamic responses using experimental data, Transportation Research Part C: Emerging Technologies 48 (2014) 285--300.

\bibitem{wang2024stability}
X.~Wang, et~al., Stability and safety analysis of connected and automated vehicle platoon considering dynamic communication topology, IEEE Trans. Intell. Transp. Syst. (2024).

\bibitem{he2021distributed}
X.~He, et~al., Distributed control under compromised measurements: Resilient estimation, attack detection, and vehicle platooning, Automatica 134 (2021) 109953.

\bibitem{segata2015toward}
M.~Segata, et~al., Toward communication strategies for platooning: Simulative and experimental evaluation, IEEE Trans. Veh. Technol. 64~(12) (2015) 5411--5423.

\bibitem{nguyen2017impact}
T.~Nguyen, et~al., Impact of communication erasure channels on control performance of connected and automated vehicles, IEEE Trans. Veh. Technol. 67~(1) (2017) 29--43.

\bibitem{khattak2023impact}
Z.~H. Khattak, et~al., Impact of communications delay on safety and stability of connected and automated vehicle platoons: Empirical evidence from experimental data, IEEE access (2023).

\bibitem{zhang2023impacts}
Y.~Zhang, et~al., Impacts of communication delay on vehicle platoon string stability and its compensation strategy: A review, Journal of traffic and transportation engineering (English edition) (2023).

\bibitem{sun2020impacts}
M.~Sun, et~al., Impacts of constrained sensing and communication based attacks on vehicular platoons, IEEE Trans. Veh. Technol. 69~(5) (2020) 4773--4787.

\bibitem{devika2023impact}
K.~Devika, et~al., Impact of v2v communication on energy consumption of connected electric trucks in stable platoon formation, in: 2023 15th International Conference on communication Systems \& networks (COMSNETS), IEEE, 2023, pp. 42--47.

\bibitem{qin2019impact}
Y.~Qin, et~al., Impact of connected and automated vehicles on passenger comfort of traffic flow with vehicle-to-vehicle communications, KSCE journal of civil engineering 23 (2019) 821--832.

\bibitem{pirani2022impact}
M.~Pirani, et~al., Impact of network topology on the resilience of vehicle platoons, IEEE Trans. Intell. Transp. Syst. 23~(9) (2022) 15166--15177.

\bibitem{ruan2022impacts}
T.~Ruan, et~al., Impacts of information flow topology on traffic dynamics of cav-mv heterogeneous flow, IEEE Trans. Intell. Transp. Syst. (2022).

\bibitem{huang1998longitudinal}
S.~Huang, et~al., Longitudinal control with time delay in platooning, IEE Proceedings-Control Theory and Applications 145~(2) (1998) 211--217.

\bibitem{zheng2015stabilitym}
Y.~Zheng, et~al., Stability margin improvement of vehicular platoon considering undirected topology and asymmetric control, IEEE Trans. Control Syst. Technol. 24~(4) (2015) 1253--1265.

\bibitem{zheng2015stability}
Y.~Zheng, {et~al.}, Stability and scalability of homogeneous vehicular platoon: Study on the influence of information flow topologies, IEEE Trans. Intell. Transp. Syst. 17~(1) (2015) 14--26.

\bibitem{li2017distributed}
S.~E. Li, et~al., Distributed platoon control under topologies with complex eigenvalues: Stability analysis and controller synthesis, IEEE Trans. Control Syst. Technol. 27~(1) (2017) 206--220.

\bibitem{ge2022scalable}
X.~Ge, {et~al.}, Scalable and resilient platooning control of cooperative automated vehicles, IIEEE Trans. Veh. Technol. 71~(4) (2022) 3595--3608.

\bibitem{lan2021data}
J.~Lan, et~al., Data-driven robust predictive control for mixed vehicle platoons using noisy measurement, IEEE Trans. Intell. Transp. Syst. (2021).

\bibitem{huang2022design}
D.~Huang, et~al., Design and analysis of longitudinal controller for the platoon with time-varying delay, IEEE Trans. Intell. Transp. Syst. (2022).

\bibitem{zakerimanesh2021heterogeneous}
A.~Zakerimanesh, et~al., Heterogeneous vehicular platooning with stable decentralized linear feedback control, in: 2021 IEEE International Conference on Autonomous Systems (ICAS), IEEE, 2021, pp. 1--5.

\bibitem{zakerimanesh2024stability}
A.~Zakerimanesh, {et~al.}, Stability and intervehicle distance analysis of vehicular platoons: Highlighting the impact of bidirectional communication topologies, IEEE Trans. Control Syst. Technol. (2024).

\bibitem{green2000long}
M.~Green, " how long does it take to stop?" methodological analysis of driver perception-brake times, Transportation human factors 2~(3) (2000) 195--216.

\bibitem{al2021driving}
Y.~Al-Wreikat, et~al., Driving behaviour and trip condition effects on the energy consumption of an electric vehicle under real-world driving, Applied Energy 297 (2021) 117096.

\bibitem{ma2021eco}
F.~Ma, et~al., Eco-driving-based cooperative adaptive cruise control of connected vehicles platoon at signalized intersections, Transportation Research Part D: Transport and Environment 92 (2021) 102746.

\bibitem{yao2020vehicle}
Y.~Yao, et~al., Vehicle fuel consumption prediction method based on driving behavior data collected from smartphones, Journal of Advanced Transportation 2020~(1) (2020) 9263605.

\bibitem{sun2020optimal}
C.~Sun, et~al., Optimal eco-driving control of connected and autonomous vehicles through signalized intersections, IEEE Internet of Things Journal 7~(5) (2020) 3759--3773.

\bibitem{yu2022eco}
M.~Yu, et~al., An eco-driving strategy for partially connected automated vehicles at a signalized intersection, IEEE Trans. Intell. Transp. Syst. 23~(9) (2022) 15780--15793.

\bibitem{liu2021dynamic}
Y.~Liu, et~al., Dynamic lane-changing trajectory planning for autonomous vehicles based on discrete global trajectory, IEEE Trans. Intell. Transp. Syst. 23~(7) (2021) 8513--8527.

\bibitem{mohajer2020enhancing}
N.~Mohajer, et~al., Enhancing passenger comfort in autonomous vehicles through vehicle handling analysis and optimization, IEEE Intell. Transp. Syst. Mag. 13~(3) (2020) 156--173.

\bibitem{wang2020research}
C.~Wang, et~al., Research on the comfort of vehicle passengers considering the vehicle motion state and passenger physiological characteristics: Improving the passenger comfort of autonomous vehicles, International journal of environmental research and public health 17~(18) (2020) 6821.

\bibitem{yan2022cooperative}
Y.~Yan, et~al., A cooperative trajectory planning system based on the passengers' individual preferences of aggressiveness, IEEE Trans. Veh. Technol. 72~(1) (2022) 395--406.

\bibitem{hang2020integrated}
P.~Hang, et~al., An integrated framework of decision making and motion planning for autonomous vehicles considering social behaviors, IEEE Trans. Veh. Technol. 69~(12) (2020) 14458--14469.

\bibitem{borenstein2021introduction}
M.~Borenstein, et~al., Introduction to meta-analysis, John Wiley \& Sons, 2021.

\end{thebibliography}

\section{Acknowledgment}
Funding for this project was provided by the Government of Alberta through Grant RCP-19-001-MIF to the Centre for Autonomous Systems in Strengthening Future Communities.
\appendices
\section{Proof of Theorem 1}
\label{appendixA}
First, let's define $\Tilde{\mathbf{X}}_{j}^{i-1}\triangleq\Tilde{\mathbf{X}}_{j}-\Tilde{\mathbf{X}}_{i-1}$ and $\Tilde{\mathbf{X}}_{i}^{j}\triangleq\Tilde{\mathbf{X}}_{i}-\Tilde{\mathbf{X}}_{j}$. For $i=2,\dots,n$ and using \it{the equation (12) of the paper}, \normalfont we have\begin{equation}
\begin{cases}
\begin{split}
    \dddot{\Tilde{x}}_{i-1}=-\frac{1}{\tau_{i-1}}\mathbb{K}_{i-1}\Tilde{\mathbf{X}}_{i-1}+\frac{1}{\tau_{i-1}}\sum_{j\in\mathbb    {I}_{i-1}}\mathbf{K}\Tilde{\mathbf{X}}_{j}+\epsilon_{i-1}
    \end{split}
    \\
    \begin{split}      
    \dddot{\Tilde{x}}_{i}=-\frac{1}{\tau_{i}}\mathbb{K}_{i}\Tilde{\mathbf{X}}_{i}+\frac{1}{\tau_{i}}\sum_{j\in\mathbb    {I}_{i}}\mathbf{K}\Tilde{\mathbf{X}}_{j}+\epsilon_{i}
    \end{split}
    \end{cases}
    \label{st-spa-f}
\end{equation}
Reminding that $\Tilde{\mathbb{j}}_{i-1}^{i}\triangleq \dddot{\Tilde{x}}_{i-1}-\dddot{\Tilde{x}}_{i}$, we get
\begin{equation}
\begin{split}
\Tilde{\mathbb{j}}_{i-1}^{i}=&
-\frac{1}{\tau_{i-1}}\mathbb{K}_{i-1}\Tilde{\mathbf{X}}_{i-1}^{i}-\left(\frac{1}{\tau_{i-1}}\mathbb{K}_{i-1}-\frac{1}{\tau_{i}}\mathbb{K}_{i}\right)\Tilde{\mathbf{X}}_{i}\\&+\sum_{j\in\mathbb    {I}_{i-1}}\frac{1}{\tau_{i-1}}
\mathbf{K}\Tilde{\mathbf{X}}_{j}-\sum\limits_{\mathclap{j\in\mathbb    {I}_{i}}}\frac{1}{\tau_{i}}
\mathbf{K}\Tilde{\mathbf{X}}_{j}+\frac{\tau_{i-1}-\tau_{i}}{\tau_{i-1}\tau_{i}}a_{0}
\end{split}
\label{xx01}
\end{equation}
Since $-\left(\frac{\mathbb{K}_{i-1}}{\tau_{i-1}}-\frac{\mathbb{K}_{i}}{\tau_{i}}
\right)=\frac{\mbox{\boldmath$\tau$}_{i-1}^{i}}{\tau_{i-1}}+\sum_{j\in\mathbb    {I}_{i}}\frac{\mathbf{K}}{\tau_{i}}
-\sum_{j\in\mathbb    {I}_{i-1}}\frac{\mathbf{K}}{\tau_{i-1}}
$, 
then \eqref{xx01} can be expressed as 
\begin{equation}
\begin{split}
\Tilde{\mathbb{j}}_{i-1}^{i}=&
-\frac{1}{\tau_{i-1}}\left(\mathbb{K}_{i-1}+\mbox{\boldmath$\tau$}_{i-1}^{i}\right)\Tilde{\mathbf{X}}_{i-1}^{i}+\sum_{j\in\mathbb    {I}_{i-1}}\frac{1}{\tau_{i-1}}
\mathbf{K}\Tilde{\mathbf{X}}_{j}^{i}\\&-\sum\limits_{\mathclap{j\in\mathbb    {I}_{i}}}\frac{1}{\tau_{i}}
\mathbf{K}\Tilde{\mathbf{X}}_{j}^{i}-\frac{1}{\tau_{i-1}}\mbox{\boldmath$\tau$}_{i-1}^{i}\sum_{\kappa=1}^{i-1}\Tilde{\mathbf{X}}_{\kappa-1}^{\kappa}+\frac{\tau_{i-1}-\tau_{i}}{\tau_{i-1}\tau_{i}}a_{0}
\end{split}
\label{xx02}
\end{equation}
which derived utilizing the fact that $\tilde{\mathbf{X}}_{0}=[0;0;0]$ and 
\begin{equation}
\frac{\mbox{\boldmath$\tau$}_{i-1}^{i}}{\tau_{i-1}}\tilde{\mathbf{X}}_{i}=-\frac{\mbox{\boldmath$\tau$}_{i-1}^{i}}{\tau_{i-1}}\left(\tilde{\mathbf{X}}_{0}-\tilde{\mathbf{X}}_{i}\right)=-\frac{\mbox{\boldmath$\tau$}_{i-1}^{i}}{\tau_{i-1}}\sum_{\kappa=1}^{i}\tilde{\mathbf{X}}_{\kappa-1}^{\kappa}  
\end{equation}
Now, since $\Tilde{\mathbf{X}}_{j}-\Tilde{\mathbf{X}}_{i}=\Tilde{\mathbf{X}}_{j}-\Tilde{\mathbf{X}}_{i-1}+\Tilde{\mathbf{X}}_{i-1}-\Tilde{\mathbf{X}}_{i}$, \eqref{xx02} becomes
\begin{equation}
\begin{split}
\Tilde{\mathbb{j}}_{i-1}^{i}=&-\left(\mbox{\boldmath$\bar{\tau}$}_{i-1}^{i}+\frac{|\mathbb    {I}_{i}|}{\tau_{i}}\mathbf{K}\right)\Tilde{\mathbf{X}}_{i-1}^{i}+\sum_{j\in\mathbb    {I}_{i-1}}\frac{1}{\tau_{i-1}}
\mathbf{K}\Tilde{\mathbf{X}}_{j}^{i-1}\\&-\sum\limits_{\mathclap{j\in\mathbb    {I}_{i}}}\frac{1}{\tau_{i}}
\mathbf{K}\Tilde{\mathbf{X}}_{j}^{i-1}-\frac{1}{\tau_{i-1}}\mbox{\boldmath$\tau$}_{i-1}^{i}\sum_{\kappa=1}^{i-1}\Tilde{\mathbf{X}}_{\kappa-1}^{\kappa}+\frac{\tau_{i-1}-\tau_{i}}{\tau_{i-1}\tau_{i}}a_{0}
\end{split}
\label{xx03}
\end{equation}
in which $\mbox{\boldmath$\bar{\tau}$}_{i-1}^{i}=[0,\;0,\;\frac{1}{\tau_{i}}]$. Note that $\mbox{\boldmath$\bar{\tau}$}_{i-1}^{i}+\frac{|\mathbb    {I}_{i}|}{\tau_{i}}\mathbf{K}=\frac{1}{\tau_{i}}\mathbb{K}_{i}$ [\it{see  equation (13) of the paper}]. \normalfont Now, based on how the connection is between the neighboring vehicles [\it{see  subsection `E. Connection Types Between Neighboring Vehicles' of the paper}], \normalfont equation \eqref{xx03} can be reformulated as:
\begin{equation}
\begin{split}
\Tilde{\mathbb{j}}_{i-1}^{i}=&-\frac{1}{\tau_{i}}\left(\mathbb{K}_{i}+\frac{\zeta_{i-1}^{i}\tau_{i}}{\tau_{i-1}}\mathbf{K}\right)\Tilde{\mathbf{X}}_{i-1}^{i}+\sum_{j\in\mathbb    {R}_{i-1}}\frac{1}{\tau_{i-1}}
\mathbf{K}\Tilde{\mathbf{X}}_{j}^{i-1}\\&-\sum\limits_{\mathclap{j\in\mathbb    {R}_{i}}}\frac{1}{\tau_{i}}
\mathbf{K}\Tilde{\mathbf{X}}_{j}^{i-1}-\frac{1}{\tau_{i-1}}\mbox{\boldmath$\tau$}_{i-1}^{i}\sum_{\kappa=1}^{i-1}\Tilde{\mathbf{X}}_{\kappa-1}^{\kappa}+\frac{\tau_{i-1}-\tau_{i}}{\tau_{i-1}\tau_{i}}a_{0}
\end{split}
\label{po898}
\end{equation}
Splitting $j\in\mathbb{R}_{i-1}$ and $j\in\mathbb{R}_{i}$ into parts $j<i-1$ and $j>i$, and utilizing the fact that for $j>i$, we have $\mathbf{X}_{j}-\Tilde{\mathbf{X}}_{i-1}=\mathbf{X}_{j}-\Tilde{\mathbf{X}}_{i}-\left(\Tilde{\mathbf{X}}_{i-1}-\Tilde{\mathbf{X}}_{i}\right)$, we can reformulate \eqref{po898} as:
\begin{equation}
\begin{split}
\Tilde{\mathbb{j}}_{i-1}^{i}=&\underbrace{\left(-\frac{1}{\tau_{i}}\mathbb{K}_{i}-\frac{z_{i-1}^{i}}{\tau_{i-1}}\mathbf{K}-\frac{1}{\tau_{i-1}}\sum\limits_{\mathclap{j\in\mathbb    {R}_{i-1}^{>i}}}
\mathbf{K}+\frac{1}{\tau_{i}}\sum\limits_{\mathclap{j\in\mathbb    {R}_{i}^{>i}}}
\mathbf{K}\right)}_{\triangleq\mathbf{A}_{i-1,i}^{31}}\Tilde{\mathbf{X}}_{i-1}^{i}\\&+\sum\limits_{j\in\mathbb    {R}_{i-1}^{<i-1}}\frac{1}{\tau_{i-1}}
\mathbf{K}\Tilde{\mathbf{X}}_{j}^{i-1}-\sum\limits_{j\in\mathbb    {R}_{i}^{<i-1}}
\frac{1}{\tau_{i}}\mathbf{K}\Tilde{\mathbf{X}}_{j}^{i-1}+\sum\limits_{j\in\mathbb    {R}_{i}^{>i}}
\frac{1}{\tau_{i}}\mathbf{K}\Tilde{\mathbf{X}}_{i}^{j}\\&-\sum\limits_{j\in\mathbb    {R}_{i-1}^{>i}}
\frac{1}{\tau_{i-1}}\mathbf{K}\Tilde{\mathbf{X}}_{i}^{j}-\frac{1}{\tau_{i-1}}\mbox{\boldmath$\tau$}_{i-1}^{i}\sum_{\kappa=1}^{i-1}\Tilde{\mathbf{X}}_{\kappa-1}^{\kappa}+\frac{\tau_{i-1}-\tau_{i}}{\tau_{i-1}\tau_{i}}a_{0}
\end{split}
\label{784poaq}
\end{equation} 
in which $\mathbf{A}_{i-1,i}^{31}$ is the third element of $\mathbf{A}_{i-1}^{i}$ [see \it{equation (14) of the paper}], \normalfont if we assume $\mathbf{A}_{i-1}^{i}=\left[\mathbf{A}_{i-1,i}^{11};\mathbf{A}_{i-1,i}^{21};\mathbf{A}_{i-1,i}^{31}\right]$, and is equivalent to [see \it{subsection `H. Accumulative Control Gains of Neighboring Vehicles}' of the paper]
\begin{equation}
\begin{split}
\mathbf{A}_{i-1,i}^{31}&=-\frac{1}{\tau_{i}}\mathbb{K}_{i}^{j\leq i-1}-\frac{1}{\tau_{i-1}}\mathbf{K}_{i-1}^{j\geq i}=-\left[\bar{k}_{i-1}^{i},\;\bar{b}_{i-1}^{i},\;\bar{h}_{i-1}^{i}\right]   
\end{split}
\end{equation}
\normalfont Now, by expressing $\Tilde{\mathbf{X}}_{j}^{i-1}$ and $\Tilde{\mathbf{X}}_{i}^{j}$ in the following forms:
\begin{equation}
\begin{cases}
\Tilde{\mathbf{X}}_{j}^{i-1}=\sum_{\kappa=j+1}^{i-1}\Tilde{\mathbf{X}}_{\kappa-1}^{\kappa}; & \quad j<i-1
\\
\Tilde{\mathbf{X}}_{i}^{j}=\sum_{\kappa=i+1}^{j}\Tilde{\mathbf{X}}_{\kappa-1}^{\kappa}; & \quad j>i
\end{cases}
\label{er4455}
\end{equation}
and substituting  which into  \eqref{784poaq} yields
\begin{equation}
\begin{split}
\Tilde{\mathbb{j}}&_{i-1}^{i}
=\left(-\frac{1}{\tau_{i}}\mathbb{K}_{i}^{j\leq i-1}-\frac{1}{\tau_{i-1}}\mathbf{K}_{i-1}^{j\geq i}   \right)\Tilde{\mathbf{X}}_{i-1}^{i}-\frac{1}{\tau_{i-1}}\mbox{\boldmath$\tau$}_{i-1}^{i}\sum_{\kappa=1}^{i-1}\Tilde{\mathbf{X}}_{\kappa-1}^{\kappa}\\&+\frac{1}{\tau_{i-1}}
\sum_{j\in\mathbb{R}_{i-1}^{<i-1}}\sum_{\kappa=j+1}^{i-1}\mathbf{K}\Tilde{\mathbf{X}}_{\kappa-1}^{\kappa}-\frac{1}{\tau_{i-1}}
\sum_{j\in\mathbb{R}_{i-1}^{>i}}\sum_{\kappa=i+1}^{j}\mathbf{K}\Tilde{\mathbf{X}}_{\kappa-1}^{\kappa}\\&+\frac{1}{\tau_{i}}
\sum_{j\in\mathbb{R}_{i}^{>i}}\sum_{\kappa=i+1}^{j}\mathbf{K}\Tilde{\mathbf{X}}_{\kappa-1}^{\kappa}
-\frac{1}{\tau_{i}}
\sum_{j\in\mathbb{R}_{i}^{<i-1}}\sum_{\kappa=j+1}^{i-1}\mathbf{K}\Tilde{\mathbf{X}}_{\kappa-1}^{\kappa}+\frac{\tau_{i-1}-\tau_{i}}{\tau_{i-1}\tau_{i}}a_{0} 
\end{split}
\label{BCTsT}
\end{equation}


On the other hand, for $i=1$, i.e., the pair $(0,1)$, since $\Tilde{\mathbf{X}}_{0}=0$, utilizing equation \eqref{st-spa-f} gives:
\begin{equation}
\begin{split}
\Tilde{\mathbb{j}}_{i-1}^{i} &= \frac{1}{\tau_{i}}\mathbb{K}_{i}\tilde{\mathbf{X}}_{i}-\frac{1}{\tau_{i}}\sum_{j\in\mathbb{R}_{i}}\mathbf{K}\tilde{\mathbf{X}}_{j} -\epsilon_{1}  
\\&=-\frac{1}{\tau_{i}}\mathbb{K}_{i}\tilde{\mathbf{X}}_{i-1}^{i}+\frac{1}{\tau_{i}}\sum_{j\in\mathbb{R}_{i}}\mathbf{K}\left(\tilde{\mathbf{X}}_{i-1}^{i}+\tilde{\mathbf{X}}_{i}^{j}\right)-\epsilon_{1}
\\&=-\frac{1}{\tau_{i}}\mathbb{K}_{i}^{j\leq i-1}\tilde{\mathbf{X}}_{i-1}^{i}+\frac{1}{\tau_{i}}\sum_{j\in\mathbb{R}_{i}^{>i}}\mathbf{K}\tilde{\mathbf{X}}_{i}^{j}-\epsilon_{1}
\end{split}
\label{01pair}
\end{equation}
Given equation \eqref{er4455} and the fact that the leader does not receive any information, we rewrite \eqref{01pair} in the following form:
\begin{equation}
\Tilde{\mathbb{j}}_{i-1}^{i} = -\frac{1}{\tau_{i}}\mathbb{K}_{i}^{j\leq i-1}\tilde{\mathbf{X}}_{i-1}^{i} +\frac{1}{\tau_{i}}
\sum_{j\in\mathbb{R}_{i}^{>i}}\sum_{\kappa=i+1}^{j}\mathbf{K}\Tilde{\mathbf{X}}_{\kappa-1}^{\kappa}-\epsilon_{i}
\label{pair01_}
\end{equation}
For the pair $(0,1)$, since the sets $\mathbb{R}_{i-1}^{<i-1}$, $\mathbb{R}_{i-1}^{>i}$, and $\mathbb{R}_{i}^{<i-1}$ are all empty, and the leader does not receive any information, \eqref{BCTsT} holds for the pair $(0,1)$ and is equivalent to \eqref{pair01_}. Therefore, utilizing \eqref{BCTsT} and noting that $d/dt\{p_{i-1}^{i}\}=v_{i-1}^{i}$, $d/dt\{v_{i-1}^{i}\}=a_{i-1}^{i}$, and $\Tilde{y}_{i-1}^{i}=a_{i-1}^{i}$, we confirm that the state-space representation, \it{the equation (14) of the paper}, \normalfont is valid. Thus, the proof of \it{Theorem 1} \normalfont is complete.

\section{Proof of Theorem 2}
\label{appendixB}
Given the coupled dynamics [\it{the equation (14) of the paper}], \normalfont $a_{i-1}^{i}(s)$ would be the summation of zero state; $\mathbf{\Xi}_{i-1}^{i}=[0;\;0;\;0]$, and zero input; $\Tilde{u}_{i-1}^{i}=0$,  responses. As such, the zero input response, denoted as $a_{i-1}^{i,zi}(s)$, is 
\begin{equation}
\begin{split}
 a_{i-1}^{i,zi}(s)&=\mathbf{C}_{i-1}^{i}\left(s\mathbf{I}_{3}-\mathbf{A}_{i-1}^{i}\right)^{-1} \mathbf{\Xi}_{i-1}^{i}=\mathbf{K}_{i-1,i}^{ini}\mathbf{T}_{1}s^{2}\Upsilon_{i-1}^{i}(s)
 \end{split}  
 \label{eq49c}
\end{equation}
where $\mathbf{K}_{i-1,i}^{ini}\triangleq \left[-\nu_{i-1}^{i}\bar{k}_{i-1}^{i},\; -\theta_{i-1}^{i}\bar{k}_{i-1}^{i} -\nu_{i-1}^{i}\bar{b}_{i-1}^{i},\;\varphi_{i-1}^{i}\right]$, and $\mathbf{I}_{3}$ is the identity matrix of size $3$. 
Also, the zero state response, denoted by $a_{i-1}^{i,zs}(s)$, would be $a_{i-1}^{i,zs}(s)=\mathcal{G}_{i-1}^{i}(s)\Tilde{u}_{i-1}^{i}(s)$ which, considering \it{the equation (17) of the paper}, \normalfont can be explored as
\begin{equation}
\begin{split}
a_{i-1}^{i,zs}(s)=&(hs^2+bs+k)\Pi_{i-1}^{i}(s)\Upsilon_{i-1}^{i}(s)+\left(\Theta_{i-1}^{i}s+\Gamma_{i-1}^{i}\right)\Upsilon_{i-1}^{i}(s)\\&-\frac{U(i-\gamma)(\tau_{i-1}-\tau_{i})}{\tau_{i-1}\tau_{i}}(\mathcal{M}_{i-1}^{i}(s)-a_{0}(s))s^{2}\Upsilon_{i-1}^{i}(s)\\&+\frac{U(\gamma-i)(1+\tau_{i} s)a_{0}(s)s^{2}\Upsilon_{i-1}^{i}(s)}{\tau_{i}}
\end{split}  
\label{kj78q}
\end{equation}
Since $a_{i-1}^{i}(s)=a_{i-1}^{i,zi}(s)+a_{i-1}^{i,zs}(s)$, using \eqref{eq49c}-\eqref{kj78q} and noting that $\mathcal{G}_{i-1}^{i}(s)=s^{2}\Upsilon_{i-1}^{i}(s)$, results in \it{the equation (21) of the paper}. \normalfont Thus, the proof of \it{the Theorem 2} \normalfont completed.
\end{document}